\pdfoutput=1
\PassOptionsToPackage{unicode}{hyperref}
\PassOptionsToPackage{hyphens}{url}
\PassOptionsToPackage{dvipsnames,svgnames,x11names}{xcolor}
\documentclass[
12pt]{article}

\usepackage{setspace}
\usepackage{amsmath,amssymb}
\usepackage{iftex}
\usepackage{SJ_Definitions}
\ifPDFTeX
\usepackage[T1]{fontenc}
\usepackage[utf8]{inputenc}
\usepackage{textcomp} % provide euro and other symbols
\else % if luatex or xetex
\usepackage{unicode-math}
\defaultfontfeatures{Scale=MatchLowercase}
\defaultfontfeatures[\rmfamily]{Ligatures=TeX,Scale=1}
\fi
\usepackage{lmodern}
\ifPDFTeX\else  
\fi
\IfFileExists{upquote.sty}{\usepackage{upquote}}{}
\IfFileExists{microtype.sty}{% use microtype if available
\usepackage[]{microtype}
\UseMicrotypeSet[protrusion]{basicmath} % disable protrusion for tt fonts
}{}
\makeatletter
\@ifundefined{KOMAClassName}{% if non-KOMA class
\IfFileExists{parskip.sty}{%
\usepackage{parskip}
}{% else
\setlength{\parindent}{0pt}
\setlength{\parskip}{6pt plus 2pt minus 1pt}}
}{% if KOMA class
\KOMAoptions{parskip=half}}
\makeatother
\usepackage{xcolor}
\makeatletter
\ifx\paragraph\undefined\else
\let\oldparagraph\paragraph
\renewcommand{\paragraph}{
\@ifstar
\xxxParagraphStar
\xxxParagraphNoStar
}
\newcommand{\xxxParagraphStar}[1]{\oldparagraph*{#1}\mbox{}}
\newcommand{\xxxParagraphNoStar}[1]{\oldparagraph{#1}\mbox{}}
\fi
\ifx\subparagraph\undefined\else
\let\oldsubparagraph\subparagraph
\renewcommand{\subparagraph}{
\@ifstar
\xxxSubParagraphStar
\xxxSubParagraphNoStar
}
\newcommand{\xxxSubParagraphStar}[1]{\oldsubparagraph*{#1}\mbox{}}
\newcommand{\xxxSubParagraphNoStar}[1]{\oldsubparagraph{#1}\mbox{}}
\fi
\makeatother

\usepackage{longtable,booktabs,array}
\usepackage{calc} % for calculating minipage widths
\usepackage{etoolbox}
\makeatletter
\patchcmd\longtable{\par}{\if@noskipsec\mbox{}\fi\par}{}{}
\makeatother
\IfFileExists{footnotehyper.sty}{\usepackage{footnotehyper}}{\usepackage{footnote}}
\makesavenoteenv{longtable}
\usepackage{graphicx}
\makeatletter
\def\maxwidth{\ifdim\Gin@nat@width>\linewidth\linewidth\else\Gin@nat@width\fi}
\def\maxheight{\ifdim\Gin@nat@height>\textheight\textheight\else\Gin@nat@height\fi}
\makeatother
\setkeys{Gin}{width=\maxwidth,height=\maxheight,keepaspectratio}
\makeatletter
\def\fps@figure{htbp}
\makeatother

\makeatletter
\@ifpackageloaded{caption}{}{\usepackage{caption}}
\AtBeginDocument{%
\ifdefined\contentsname
\renewcommand*\contentsname{Table of contents}
\else
\newcommand\contentsname{Table of contents}
\fi
\ifdefined\listfigurename
\renewcommand*\listfigurename{List of Figures}
\else
\newcommand\listfigurename{List of Figures}
\fi
\ifdefined\listtablename
\renewcommand*\listtablename{List of Tables}
\else
\newcommand\listtablename{List of Tables}
\fi
\ifdefined\figurename
\renewcommand*\figurename{Figure}
\else
\newcommand\figurename{Figure}
\fi
\ifdefined\tablename
\renewcommand*\tablename{Table}
\else
\newcommand\tablename{Table}
\fi
}
\@ifpackageloaded{float}{}{\usepackage{float}}
\floatstyle{ruled}
\@ifundefined{c@chapter}{\newfloat{codelisting}{h}{lop}}{\newfloat{codelisting}{h}{lop}[chapter]}
\floatname{codelisting}{Listing}

\makeatother
\makeatletter
\@ifpackageloaded{caption}{}{\usepackage{caption}}
\@ifpackageloaded{subcaption}{}{\usepackage{subcaption}}
\makeatother

\ifLuaTeX
\usepackage{selnolig}  % disable illegal ligatures
\fi
\usepackage[]{natbib}
\usepackage{bookmark}

\IfFileExists{xurl.sty}{\usepackage{xurl}}{} % add URL line breaks if available
\hypersetup{
pdftitle={Title},
pdfauthor={Author 1; Author 2},
pdfkeywords={3 to 6 keywords, that do not appear in the title},
colorlinks=true,
linkcolor={blue},
filecolor={Maroon},
citecolor={Blue},
urlcolor={Blue},
pdfcreator={LaTeX via pandoc}}

\newcommand{\anon}{1}

\newcommand{\inn}{\operatorname{in}}
\newcommand{\Sighat}{\widehat{{\bSigma}}}
\newcommand{\op}{\operatorname{op}}
\newcommand{\sig}{\sigma}
\newcommand{\Sbar}{\mathbf{S}}
\newcommand{\Shat}{\widehat{\mathbf{S}}}
\newcommand{\final}{\operatorname{final}}

\newcommand{\de}{\mathrm{d}}
\newcommand{\Sigmatreated}{\bSigma_{\mathcal{S},1}}
\newcommand{\Sigmacontrol}{\bSigma_{\mathcal{S},0}}
\newcommand{\event}{\mathscr{E}}
\usepackage{etoc}
\etocdepthtag.toc{mtchapter}
\etocsettagdepth{mtchapter}{subsection}
\etocsettagdepth{mtappendix}{none}
\usepackage{enumitem}
\begin{document}

\def\spacingset#1{\renewcommand{\baselinestretch}%
{#1}\small\normalsize} \spacingset{1}

\setlength{\bibsep}{0pt}
\renewcommand{\bibfont}{\small}

%%%%%%%%%%%%%%%%%%%%%%%%%%%%%%%%%%%%%%%%%%%%%%%%%%%%%%%%%%%%%%%%%%%%%%%%%%%%%%

\if1\anon
{
\title{\bf Transfer Learning of CATE with Kernel Ridge Regression}
\author{Seok-Jin Kim 
% \thanks{
% The authors gratefully acknowledge \textit{please remember to list all relevant funding sources in the unblinded version}}\hspace{.2cm}
\\
Department of IEOR, Columbia University\\
Hongjie Liu \\
Department of Statistics, Purdue University \\
Molei Liu \\
Department of Biostatistics, Peking University Health Science Center;  \\ Beijing International Center for Mathematical Research, Peking University\\
and \\
Kaizheng Wang \\
Department of IEOR and Data Science Institute, Columbia University}
\maketitle
} \fi

\if0\anon
{
\bigskip
\bigskip
\bigskip
\begin{center}
{\LARGE\bf Transfer Learning of CATE with Kernel Ridge Regression}
\end{center}
\medskip
} \fi

\bigskip
\begin{abstract}
\noindent The proliferation of data has sparked significant interest in leveraging findings from one study to estimate treatment effects in a different target population without direct outcome observations. However, the transfer learning process is frequently hindered by substantial covariate shift and limited overlap between (i) the source and target populations, as well as (ii) the treatment and control groups within the source. We propose a novel method for overlap-adaptive transfer learning of conditional average treatment effect (CATE) using kernel ridge regression (KRR). Our approach involves partitioning the labeled source data into two subsets. The first one is used to train candidate CATE models based on regression adjustment and pseudo-outcomes. An optimal model is then selected using the second subset and unlabeled target data, employing another pseudo-outcome-based strategy. We provide a theoretical justification for our method through sharp non-asymptotic MSE bounds, highlighting its adaptivity to both weak overlaps and the complexity of the CATE function. Extensive numerical studies confirm that our method achieves superior finite-sample efficiency and adaptability. We conclude by demonstrating the effectiveness and superiority of our approach on two real-world datasets.
\end{abstract}

\noindent%
\vfill

\newpage
\spacingset{1.8} % DON'T change the spacing!

\section{Introduction}

%\subsection{Background}
%\noindent
Predicting the conditional average treatment effect (CATE) given the covariates of an individual enables more precise and personalized decision-making in various application fields \citep{kent2018personalized,dube2023personalized}. Meanwhile, there have emerged increasing needs in transporting one particular causal study to infer the effect of interest on a different population without observations of the targeted treatment and outcome. For example, for a new drug studied on some source cohort, one may be interested in generalizing its treatment effect to another target cohort and comparing it with some existing treatment appearing in the target sample \citep{signorovitch2012matching}. There is also great interest in transferring the results in randomized controlled trials (RCT) conducted on some less representative subjects to more general real-world observational populations \citep{colnet2024causal}. 

In these data integration setups, covariate shift adaptation plays an important role in adjusting for the potential bias caused by the heterogeneity between the source and target \citep{pan2009survey}. For this purpose, the most frequently used strategy is to match the source sample with the target through importance weighting (IW) \citep{huang2007correcting,liu2023augmented}. However, the classic IW method tends to result in low effective sample sizes \citep{signorovitch2010comparative} when the covariate distribution of the target population deviates from the source excessively. Such a weak overlap between the source and target is a common issue in observational studies on which one could barely design or control the data collection mechanism.  For example, consider a source sample with its age normally ranging between 30 and 65 and a target with the mean age 62 \citep{ishak2015simulation}. In this case, matching the two samples on age could introduce extremely high variance to the IW function and cause inefficiency. Moreover, a similar weak overlapping issue often occurs between the non-completely-randomly assigned treated and control groups on a single source, referred as the violation of the positivity assumption \citep{cole2008constructing}. These challenges can impede effective transfer learning of the causal effects and models, especially when the outcome models are highly complex.

In this paper, we aim at addressing the aforementioned two folds of weak overlap in the transfer learning setup where the treatment and outcome are only observed on some source sample but not on the target. In particular, we are interested in the transfer learning of CATE through kernel ridge regression (KRR) within the framework of reproducing kernel Hilbert space (RKHS) response functions.

\paragraph*{Related Literature}
%\subsection{Related Literature}
%\noindent
With the increasing interests in personalized decision making in various fields, CATE has caught great attention in recent literature. For example, \cite{kunzel2019metalearners} proposed a meta-learner for CATE that allows the use of machine learning (ML) and adapts to the structural properties of the targeted treatment effect function. \cite{lan2024causal} developed a model selection approach for CATE using Q-aggregation. \cite{kennedy2020towards} extended the double machine learning (DML) framework \citep{chernozhukov2018double} for the estimation of CATE with general ML methods and showed that their estimator achieves the rate-double-robustness. \cite{kennedy2022minimax} further studied the minimax property of the semiparametric estimation of CATE. A central theme in \citet{kennedy2020towards,kennedy2022minimax} is that the nuisance propensity and outcome models often have higher structural complexity than their contrast (the CATE function). In such cases, the goal is to exploit the reduced complexity of the contrast function to achieve rates faster than naive plug-in methods and, when conditions are favorable, to obtain error bounds adaptive to the contrast function's complexity. Despite the recent advances, there still lack approaches for the transfer learning of CATE targeting some population with drastic covariate shift from the source study.

Related to this methodological gap, \cite{hartman2015sample} proposed a calibration weighting (CW) approach to adjust for the covariate shift between RCTs and real-world data sets (RWDs) and generalize the RCT for inference of the average treatment effect (ATE) on more general populations. \cite{dahabreh2019generalizing} established a causal framework with identification conditions in a similar setup and developed an augmented inverse probability weighting estimator for the ATE on the target population. \cite{lee2023improving} proposed an augmented CW (ACW) method for transfer learning of ATE that is less prone to misspecification of the calibration models. More close to our CATE problem, \cite{wu2023transfer} developed a method to generalize the individual treatment rules for RCTs to RWDs. Nevertheless, all these methods are developed under the classic semiparametric inference framework with the positivity and strong overlap assumption \citep{rosenbaum1983central}. To the best of our knowledge, no existing work on this track can adapt to the unknown degree of covariate shift between the source and target and stay effective under weak overlap.

To address such unknown and poor overlap issues in transfer learning, \cite{ma2023optimally} proposed a reweighted KRR method for the adaptive and efficient estimation of RKHS outcome models that weights the source samples with truncated density ratio functions. Recently, \cite{wang2026pseudo} developed a pseudo-labeling transfer learning approach that achieves the optimal model selection for KRR on some target population with strong and unknown covariate shift to the source data. In broader contexts, we also notice related work on optimal kernel-based treatment effect estimation in the absence of the strong overlap or positivity assumption on the treatment regime \citep{mou2023kernel}, as well as pessimistic policy learning for the optimal treatment rule in a similar scenario \citep{jin2022policy}. Nevertheless, none of the existing work can fully adapt to our setup with more complicated data structure involving two folds of distributional shift between (i) the source and target populations, and (ii) the treatment and control groups on the source, as well as (iii) the unknown function complexity of the CATE and nuisance outcome models. 

\paragraph*{Contributions}
%\subsection{Our Contribution}
%\noindent
We propose a novel approach for the transfer learning of {\bf C}ATE with {\bf O}verlap-adaptive {\bf KE}rnel ridge regression (\texttt{COKE}). It first splits the source data and uses one subset to derive candidate regression adjustment (RA) learners for CATE, through KRR with various regularization parameters. Then it fits another KRR with small regularization on the holdout source data to impute the counterfactual outcomes on the unlabeled target data. Finally, the best candidate RA learner is selected using the target data with pseudo-outcomes. Through theoretical analyses, we establish the rate-optimality of our regression adjustment and model selection strategies. Consequently, \texttt{COKE} simultaneously achieves the adaptivity to (i) unknown and weak overlap between the source and target distributions; (ii) unknown and weak overlap between the treated and control groups on the source; (iii) unknown complexity of the CATE function, as well as (iv) rate robustness to errors in the highly complex nuisance outcome models. 

For the issue of weak overlap, we notice that recent work like \cite{wang2026pseudo} and \cite{mou2023kernel} can only address simpler setups. Our setting is fundamentally different from the covariate shift problem considered by \citet{wang2026pseudo} in the sense that (i)--(iii) can co-occur in our case. Simultaneously addressing them and achieving sharp error rates are technically more involving due to the complication of our missing data structure. To accomplish this goal, we propose a {\em two-stage} KRR procedure accompanied by model selection based on pseudo-outcomes, while \citet{wang2026pseudo} only employs a \emph{single-stage} KRR and, thus, cannot be adaptive to the issues of two-folded weak overlap and unknown complexity, which are predominant in the transfer learning problem of CATE.

For (iii) and (iv), \texttt{COKE} is shown to maintain milder dependency on the complexity of the nuisance models compared to the targeted CATE function. This is in a similar spirit to recent works by \cite{kennedy2020towards}, \cite{kennedy2022minimax} and \cite{kato2023cate}. In contrast to them, our method achieves (iii) and (iv) without the common assumptions of positivity or strong overlap \citep{rosenbaum1983central}. In addition, we introduce an RKHS framework via the Hilbert norm to capture the ``structural complexity'' of the nuisance and CATE functions, which moves beyond the H\"older class studied in \citet{kennedy2022minimax} and reflects both smoothness and magnitude. 
This advance corresponds to recent literature on causal inference that emphasizes the need to study notions of complexity beyond smoothness \citep{cinelli2025challenges}.

In both theoretical and numerical studies, \texttt{COKE} consistently outperforms existing strategies like the separate regression and DML estimation, under various settings on the degree of distributional overlap and the complexity of the models. We also illustrate the superiority of \texttt{COKE} in two real-world examples of transfer learning with significant covariate shift.

\paragraph*{Notation}
%\subsection{Notations}\noindent
We define \( [n] := \{1, 2, \dots, n\} \).  
For an operator \( A \) on a Hilbert space \( \HH \), we define the operator norm as \( \|A\|_{\op} \), and for \( \theta \in \HH \), we write the Hilbert norm as \( \|\theta\|_{\HH} \).  
We denote the identity operator as \( \Ib \), and the outer product operator of \( x, y \in \HH \) as \( x \otimes y \).  
We use the notation \( \cO(\cdot) \) or \( \lesssim \) to hide constants, and \( \tilde{\cO}(\cdot) \) to hide constants and logarithmic terms.  
Whenever additional factors are hidden, we explicitly note them.  
We use the notation \( a \asymp b \) when \( a \lesssim b \) and \( b \lesssim a \).  
For the inner product in the Hilbert space \( \HH \), we use \( \langle a, b \rangle_\HH \); however, when the context is clear, we use \( a^\top b \) or \( b^\top a \) for simplicity and to improve readability.
The constants \(c, C, c_1, c_2, C_1, C_2, \dots \) \textit{may differ from line to line.}

\section{Problem Setup}\label{section: problem setup}

\subsection{Treatment Regime and Covariate Shift}\label{sec-setup-1}\noindent
We introduce the treatment regime and the conditional average treatment effect (CATE) under covariate shift.  
Suppose we observe $n$ i.i.d.\ samples from a \emph{source} distribution, denoted by $\cD = \{(z_i, a_i, y_i)\}_{i=1}^n$. Here, $z_i \in \cZ$ are covariates, $a_i \in \{0,1\}$ is the binary treatment indicator, and $y_i \in \RR$ is the response.
For the \emph{target} distribution of primary interest, assume there are \(n_\cT\) i.i.d.\ samples
\(\{z_{0i}, a_{0i}, y_{0i}\}_{i=1}^{n_\cT}\) with \(z_{0i} \in \cZ\), \(a_{0i} \in \{0,1\}\), and \(y_{0i} \in \RR\).  
We focus on the setting where \(\{(a_{0i}, y_{0i})\}_{i=1}^{n_\cT}\) are \emph{unobserved}, so only the covariates
\(\cD_{\cT} := \{z_{0i}\}_{i=1}^{n_\cT}\) are available.

The covariate distributions may differ, giving rise to \emph{covariate shift}.  
We write \(z_i \sim \cQ_{\mathcal{S}}\) for source covariates and \(z_{0i} \sim \cQ_{\mathcal{T}}\) for target covariates.  
Treatment assignment in the source follows \(a_i \mid z_i \sim \mathrm{Bernoulli}(\pi(z_i))\), where \(\pi: \cZ \to [0,1]\) is the propensity score.  
Rather than the usual positivity condition (which requires \(\pi\) to be bounded away from \(\{0,1\}\)), we adopt a weaker assumption that allows singular cases; see Section~\labelcref{subsection: weak overlap}.  
We impose no model on \(\pi\).

Let \(\cQ_\cS^\star\) and \(\cQ_\cT^\star\) denote the joint distributions of \((z_i, a_i, y_i)\) and \((z_{0i}, a_{0i}, y_{0i})\), respectively.  
For a function space \(\cF\), assume there exist \(f_0^\star, f_1^\star \in \cF\) such that
\begin{align*}
f_{1}^\star(z) &= \EE_{(z,a,y)\sim \cQ_\cS^\star}\bigl[y \mid a=1, z\bigr] 
= \EE_{(z,a,y)\sim \cQ_\cT^\star}\bigl[y \mid a=1, z\bigr],\\
f_{0}^\star(z) &= \EE_{(z,a,y)\sim \cQ_\cS^\star}\bigl[y \mid a=0, z\bigr]
= \EE_{(z,a,y)\sim \cQ_\cT^\star}\bigl[y \mid a=0, z\bigr].
\end{align*}
Thus, the response functions for treated and control units are shared across source and target populations, while the covariate distributions may differ.

Our goal is to estimate the CATE function,
\[
h^\star(z) := f_1^\star(z) - f_0^\star(z),
\]
and to evaluate any estimator \(h \in \cF\) by its mean squared error (MSE) under the target distribution:
\[
\cE_{\cT}(h) = \EE_{z \sim \cQ_\cT}|h(z) - h^\star(z)|^2.
\]
We also define the noise variables \(\varepsilon_i := y_i - f^\star_{a_i}(z_i)\) for \(i \in [n]\); by construction,
\(\EE[\varepsilon_i \mid z_i,a_i] = 0\). Next, we focus on the case where \(\cF\) is an RKHS induced by a symmetric, positive semidefinite kernel \(K(\cdot,\cdot): \cZ \times \cZ \to \RR\) \citep{wainwright2019high}.  
By the Moore--Aronszajn Theorem \citep{aronszajn1950theory}, there exists a Hilbert space \(\HH\) and a feature map \(\phi: \cZ \to \HH\) such that \(\langle \phi(z), \phi(w)\rangle_{\HH} = K(z,w)\). Define
\[
\mathcal{F} = \bigl\{
f_\theta: \cZ \to \RR 
\mid 
f_\theta(z) = \langle \phi(z), \theta \rangle_{\HH}
\text{ for some } \theta \in \HH 
\bigr\}.
\]
This class includes Sobolev and Besov spaces as special cases \citep{zhang2023optimality,fischer2020sobolev}.  
Moreover, \(\cF\) is isomorphic to \(\HH\), and we denote its norm by \(\|\cdot\|_{\cF}\).  
Finally, assume the kernel \(K\) is bounded, i.e., \(\sup_{z \in \cZ} K(z,z) \leq \xi\) for some \(\xi > 0\).  
This assumption is common in KRR analyses.

\subsection{Kernel Ridge Regression}\label{section: preliminaries kernel ridge regression}
\noindent
We briefly review KRR in the setting with $N$ covariate--response pairs.  
Suppose we observe data \(\{(u_i, r_i)\}_{i=1}^N\), with \(u_i \in \cZ\) and \(r_i \in \RR\).
KRR estimates a function \(\hat{f}\) by solving the penalized least-squares problem
\begin{align}\label{equation: KRR program}
\hat{f}=  \arg \min_{f \in \cF} 
\bigg\{
\frac{1}{N}\sum_{i=1}^N (r_i - f(u_i))^2 + \lambda \| f\|_{\cF}^2
\bigg\}
.
\end{align}
We refer to \(\lambda>0\) as the ridge regularizer.  
By the representer theorem, the solution depends only on the kernel evaluations \(\{K(u_i, u_j)\}_{1\le i,j \le N}\), so \(\hat{f}\) can be computed via an equivalent finite-dimensional quadratic program \citep{wainwright2019high}.  
A primary challenge in KRR is selecting \(\lambda\) due to the inherent bias--variance trade-off.  
We provide additional details and the closed-form solution in Appendix~\labelcref{section: groundwork}.

Before presenting our main algorithm, we describe a naive approach for CATE estimation, referred to as \emph{separate regression}.  
We split the data \(\cD\) into the treated group \(\cD(a=1)\) and the control group \(\cD(a=0)\), where
\[
\cD(a=1) := \{(z_i, a_i, y_i) \in \cD \mid a_i =1 \}, 
\quad
\cD(a=0) := \{(z_i, a_i, y_i) \in \cD \mid a_i =0 \}.
\]
We choose regularizers \(\lambda_0\) and \(\lambda_1\), estimate each \(f_k^\star\) by applying KRR to \(\cD(a=k)\) with regularizer \(\lambda_k\), and obtain \(\hat{f}_0\) and \(\hat{f}_1\).  
The resulting CATE estimator is \(\hat{h}_{\operatorname{sep}} = \hat{f}_1 - \hat{f}_0\).  
This plug-in procedure is attractive for its simplicity, since \(\lambda_0\) and \(\lambda_1\) can be chosen to minimize the estimation errors of \(\hat{f}_0\) and \(\hat{f}_1\) separately.  
However, because our main interest lies in the contrast function \(h^\star = f^\star_1 - f^\star_0\), there is no direct rule of thumb for selecting \(\lambda_0\) and \(\lambda_1\) to optimize the estimation of \(h^\star\).

We close this section by noting the challenges posed by covariate shift in KRR.  
Even when the response model is shared across source and target data, an unknown degree of covariate shift makes the choice of regularization both crucial and difficult.  
This has prompted extensive research on covariate shift in KRR \citep{ma2023optimally,wang2026pseudo,PDT24}.

\section{Methodology}\label{section: methodology}
\noindent
In this section, we present our transfer learning methodology for CATE estimation. It has three steps: 
(1) splitting the data, 
(2) training candidate models via a regression adjustment (RA) learner, 
and (3) selecting the best model.

To begin, we choose positive integers $n_1, n_{2}$ such that $n = n_1 + n_{2}$, and randomly split the source data $\cD$ into two subsets $\{\cD_1, \cD_2\}$ of sizes $\{ n_1, n_{2} \}$.  
We denote each $\cD_j$ as $\{ (z_{ji}, a_{ji}, y_{ji}) \}_{i=1}^{n_j}$.  
We will use $\cD_1$ to train candidate models, and then use $\cD_2 \cup \cD_{\cT}$ for model selection.  
For simplicity, we assume below that $n$ is a multiple of $2$ and set $n_1 = n_{2} = n/2$.

\begingroup 
We provide a high-level summary of our algorithm below.
\begin{enumerate}[label=(\arabic*)]
\item Partition the training dataset into two disjoint subsets, $\cD_1$ and $\cD_2$.

\item For a grid of regularization parameters $\bm\lambda \in \bm{\Lambda}$, apply the RA learner to estimate the CATE on $\cD_1$, obtaining a set of candidate estimators $\cH =\{\hat{h}_{\bm\lambda}\}_{\bm \lambda \in \bm{\Lambda}}$.

\item Learn imputation models for the treated and untreated outcomes on $\mathcal{D}_2$ via separate, undersmoothed KRR. Then, construct (pseudo-)test outcomes for the target dataset as the difference between the imputed treated and control outcomes.

\item Perform model selection over the candidate models derived in Step 2, using the test outcomes from Step 3.
\end{enumerate}
\endgroup

\subsection{RA Learner}
\noindent
We now describe the \emph{Regression Adjustment (RA)} learner. 
Our goal is to estimate the contrast $h^{\star}$, but the treatment structure leaves one potential outcome unobserved, so direct regression on $h^\star$ is not possible.
Instead, we construct pseudo-outcomes and regress on them.
The RA learner proceeds in three steps: (1) estimate imputation models on $\cD_1$; (2) generate pseudo-outcomes on $\cD_1$; and (3) run KRR on these pseudo-outcomes in \(\cD_1\) to learn \(h^\star\).

Several works \citep{kunzel2019metalearners,kennedy2020towards,curth2021nonparametric} have studied a similar methodology; however, our approach differs in three ways: (i) we use KRR, (ii) we do not split the data between the first-stage and second-stage regressions, instead using the same dataset for both, and (iii) in our main algorithm, we provide guidance on selecting the regularizer for both stages.

Our RA learner (Algorithm~\labelcref{algorithm: RA learner}) takes as input the dataset \(\cD_1\) and a tuple of three regularizers \(\bm{\lambda }  = (\lambda_{0,0}, \lambda_{0,1}, \lambda_{1}) \in \RR^3\) with \(\lambda_{0,0}, \lambda_{0,1}, \lambda_{1} >0\).  
First, we estimate nuisance functions by performing KRR separately on \(\cD_1\) with regularizers \(\lambda_{0,0}\) and \(\lambda_{0,1}\), which yields the following estimators:
\begin{align}\label{equation: RA learner algorithm}
\hat{f}_0 &:= \arg\min_{f \in \cF} \bigg\{ \frac{1}{n_1} \sum_{i=1}^{n_1} (y_{1i} -f(z_{1i}))^2 \one(a_{1i}=0) + \lambda_{0,0} \| f\|_{\cF}^2 \bigg\}, \\
\hat{f}_1 &:= \arg \min_{f \in \cF} \bigg\{\frac{1}{n_1} \sum_{i=1}^{n_1} (y_{1i} -f(z_{1i}))^2\one(a_{1i}=1)  + \lambda_{0,1} \| f\|_{\cF}^2\bigg\}. 
\label{equation: RA learner algorithm 2}
\end{align}
We then form pseudo-outcomes \(\{m_{1i}\}_{i=1}^{n_{1}}\) on $\cD_{1}$ by 
\begin{align*}
m_{1i} := 
\begin{cases}
y_{1i} - \hat{f}_0(z_{1i}), & \text{if } a_{1i} = 1,\\[6pt]
\hat{f}_1(z_{1i}) - y_{1i}, & \text{if } a_{1i} = 0.
\end{cases}
\end{align*}
Next, we run KRR with regularizer \(\lambda_{1}\) and return the CATE estimator:
\begin{align*}
\hat{h}_{\bm{\lambda}} := \arg \min_{h \in \cF} \bigg\{ \frac{1}{n_{1}} \sum_{i=1}^{n_{1}} (m_{1i}-h(z_{1i}))^2  + \lambda_{1} \| h\|_\cF^2 \bigg\}.
\end{align*}
The pseudocode of this procedure is given in Algorithm~\ref{algorithm: RA learner}.

\begin{algorithm}[]
\caption{RA Learner}
\label{algorithm: RA learner}
\begin{algorithmic}
\Require  Dataset \(\cD_1\), regularizers  \(\bm{\lambda }  = (\lambda_{0,0}, \lambda_{0,1}, \lambda_{1}) \in \RR^3\) where \(\lambda_{0,0}, \lambda_{0,1}, \lambda_{1} >0 \).
\State Using \(\cD_1\), run KRR to get nuisance estimators \(\hat{f}_0, \hat{f}_1\) by solving \eqref{equation: RA learner algorithm} and \eqref{equation: RA learner algorithm 2}.
\State Using \(\cD_{1}\), set the pseudo-outcome \(m_{1i} := (y_{1i} - \hat{f}_0(z_{1i}) )\bm{1}(a_{1i}=1) + (\hat{f}_1(z_{1i}) - y_{1i} )\bm{1}(a_{1i}=0)\) for all \(i \in [n_{1}]\).
\State On \(\cD_{1}\), apply KRR using the independent variables \(\{z_{1i}\}_{i=1}^{n_{1}}\) and responses \(\{m_{1i}\}_{i=1}^{n_{1}}\) with regularizer \(\lambda_{1} > 0\), and get estimator \(\hat{h}_{\bm{\lambda}}\).
\Ensure \(\hat{h}_{\bm{\lambda}}\).

\end{algorithmic}
\end{algorithm}

Recall that \emph{separate regression} estimates two nuisance functions independently and outputs $\hat{f}_1 - \hat{f}_0$ as the CATE estimator.  
By contrast, Algorithm~\labelcref{algorithm: RA learner} fits an additional regression on the pseudo-outcomes, which, as we will see, brings notable benefits.

\subsection{Model Selection}
\noindent
We now present our model selection procedure, which constructs test outcomes and selects the best model using them.  
Since \(\cD_\cT\) contains only target covariates, Algorithm~\labelcref{algorithm: model selection} first constructs test outcomes and then selects the best model with respect to these outcomes.
The inputs are the dataset \(\cD_2\), the target covariates \(\cD_\cT\), and a set of candidate CATE estimators denoted by \(\Hcal_0 := \{\hat{h}_1, \dots, \hat{h}_L\}\) for some \(L > 0\).

First, we set two regularizers \(\tilde{\lambda}_0, \tilde{\lambda}_1 >0\).  
% Their specific choices are deferred to the theoretical results in Section~\labelcref{section: main results}.  
Next, we run KRR on \(\cD_2\) to obtain
\begin{align}\label{equation: test outcome algorithm}
\tilde{f}_0 &:= \arg\min_{f \in \cF} \bigg\{\frac{1}{n_{2}} \sum_{i=1}^{n_{2}} (y_{2i} -f(z_{2i}))^2 \one(a_{2i}=0) + \tilde{\lambda}_{0} \| f\|_{\cF}^2 \bigg\}, \\
\tilde{f}_1 &:= \arg \min_{f \in \cF} \bigg\{ \frac{1}{n_{2}} \sum_{i=1}^{n_{2}} (y_{2i} -f(z_{2i}))^2\one(a_{2i}=1)  + \tilde{\lambda}_{1} \| f\|_{\cF}^2 \bigg\}. 
\label{equation: test outcome algorithm 2}
\end{align}
% \begin{align}\label{equation: test outcome algorithm}
% \tilde{f}_a &:= \arg\min_{f \in \cF} \bigg\{\frac{1}{n_{2}} \sum_{i=1}^{n_{2}} (y_{2i} -f(z_{2i}))^2 \one(a_{2i}=a) + \tilde{\lambda}_{a} \| f\|_{\cF}^2 \bigg\}\quad a \in \{0,1\} .
% % \label{equation: test outcome algorithm 2}
% \end{align}
We use $\tilde{h} = \tilde{f}_1 - \tilde{f}_0$ to generate test outcomes on $\cD_{\cT}$, and then select the final model from $\Hcal_0$ by empirical squared-loss minimization.
Algorithm~\labelcref{algorithm: model selection} provides a detailed description.  
Here we choose $\tilde{\lambda}_0, \tilde{\lambda}_1$ to be small (undersmoothing), at a scale such as $\asymp \log n/n$, to serve as low-bias proxies; precise choices are deferred to the theory.

\begin{algorithm}[]
\caption{Model Selection}
\label{algorithm: model selection}
\begin{algorithmic}
\Require  Dataset \(\cD_2, \cD_\cT\); set of CATE estimators $\Hcal_0 = \{ \hat{h}_1, \dots, \hat{h}_L\}$.
\State Run KRR to obtain $\tilde{f}_0, \tilde{f}_1$ by solving \eqref{equation: test outcome algorithm} and \eqref{equation: test outcome algorithm 2}.
\State Define $\tilde{h} = \tilde{f}_1 -\tilde{f}_0$, and form the test outcomes for \(\cD_\cT\) as \(\{\tilde{h}(z_{0i})\}_{i=1}^{n_{\cT}}\).
\State For each $h \in \Hcal_0$, compute
\[
L(h) = \frac{1}{n_\cT} \sum_{i=1}^{n_\cT} \bigl(\tilde{h}(z_{0i}) -h(z_{0i})\bigr)^2.
\]
Choose the final model $\hat{h}_{\final} = \arg \min_{h \in \Hcal_0} L(h)$.
\Ensure \(\hat{h}_{\final}\).
\end{algorithmic}
\end{algorithm}

Next, we combine Algorithm~\labelcref{algorithm: RA learner} (the RA learner) and Algorithm~\labelcref{algorithm: model selection} to present our main algorithm, which is described in the following section.

\subsection{The Final Procedure}
\noindent
We are now ready to present our method, Transfer learning of {\bf C}ATE with {\bf O}verlap-adaptive {\bf KE}rnel ridge regression (COKE), in Algorithm~\labelcref{algorithm: main}. It first uses an RA learner to generate multiple candidate estimators, and then employs Algorithm~\labelcref{algorithm: model selection} to select the final one.

The algorithm takes as input the labeled source data \(\cD\) and unlabeled target data \(\cD_\cT\).  
We use fixed small nuisance regularizers \(\lambda_{0,0}, \lambda_{0,1}\) (specified in Section~\labelcref{section: main results}) on the order of $\log n/n$.  
This choice minimizes bias in the first-stage imputation, prioritizing low approximation error over variance reduction.
We then construct the grid \(\Lambda_1  \subset \RR\),
which serves as the grid for the regularizer \(\lambda_{1}\).  
We define \(\bm{\Lambda} := \{ \lambda_{0,0}\} \times \{\lambda_{0,1}\} \times \Lambda_{1}\), and apply the RA learner (Algorithm~\labelcref{algorithm: RA learner}) for every \(\bm{\lambda} = (\lambda_{0,0}, \lambda_{0,1}, \lambda_{1}) \in \bm{\Lambda}\), using the split \(\cD_1\).  
This produces a set of candidate CATE estimators, \(\cH = \{\hat{h}_{\bm{\lambda}} \mid \bm{\lambda} \in \bm{\Lambda}\}\).  
We then feed \(\cH\) and the remaining data \(\cD_2, \cD_\cT\) into Algorithm~\labelcref{algorithm: model selection} to obtain our final model.  
The pseudocode is summarized below.

\begin{algorithm}[h]
\caption{\texttt{COKE}: Transfer learning of {\bf C}ATE with {\bf O}verlap-adaptive {\bf KE}rnel ridge regression}
\label{algorithm: main}
\begin{algorithmic}
\Require Dataset $\cD, \cD_{\cT}$, $\lambda_{0,0}, \lambda_{0,1}, \tilde{\lambda}_0, \tilde{\lambda}_1 >0$, Regularizer grid $\Lambda_1 \subset \RR$.
% \State Set grid \(\Lambda_{1} = \Bigl\{ \frac{\xi \log n}{n}, \frac{2 \xi \log n}{n}, \dots, \frac{2^q \xi \log n}{n} \Bigr\}\) where $q = \lceil 2\log n \rceil$, and set $\lambda_{0,0}, \lambda_{0,1} = \frac{\xi\log n}{n}$.
\State Split source data $\cD$ into $\cD_{1}, \cD_{2}$.
\State Define \(\bm{\Lambda} := \{ \lambda_{0,0}\} \times \{\lambda_{0,1}\} \times \Lambda_{1}\).
For each \( \bm{\lambda} \in \bm{\Lambda}\), run Algorithm~\labelcref{algorithm: RA learner} on \(\cD_{1}\) to obtain $\hat{h}_{\bm{\lambda}}$.
\State Set the collection of candidates as $\cH = \{\hat{h}_{\bm{\lambda}} \mid \bm{\lambda} \in \bm{\Lambda}\}$.
\State Run Algorithm~\labelcref{algorithm: model selection} with input \(\cD_2, \cD_\cT, \cH\), yielding the final model \(\hat{h}_{\final}\).
\Ensure $\hat{h}_{\final}$.
\end{algorithmic}
\end{algorithm}

% Algorithm~\labelcref{algorithm: main} outputs the best CATE estimator among the candidates generated by the RA learner. We discuss the adaptivity of the model selection procedure in Section~\labelcref{section: main results}.

\section{Theoretical Results}\label{section: main results}
\noindent
In this section, we present our theoretical results, based on a weak overlap framework. 
%including the final model's MSE bounds. 
%and outline the assumptions that underlie these results. 
%We begin by describing the standard assumptions, then introduce our weak overlap framework.

\subsection{Standard Assumptions}
\noindent
To set the stage, we present three assumptions that are widely used in causal inference and nonparametric regression.
%We will discuss them in detail after listing them below.

\begin{assumption}[Consistency and unconfoundedness]\label{assumption; consistency and unconfoundedness}
An individual $i$ has two potential outcomes \(y_i(1), y_i(0)\) under treatment and control, respectively.
\begin{itemize}
\item Consistency: If an individual $i$ is assigned treatment $a_i \in \{0, 1\}$, we observe \(y_i = y_i(a_i)\). 
\item Unconfoundedness: There are no unobserved confounders, i.e., $( y_i(0), y_i(1) )$ are independent of $a_i$ given $z_i$.    
\end{itemize}
\end{assumption}

\begin{assumption}[Sub-Gaussian noise]\label{assumption; subGaussian noise}
Conditional on \(z_i\) and \(a_i\), the noise variables \(\varepsilon_i\) are sub-Gaussian with proxy \(\sigma > 0\).
For simplicity, we assume \(\sigma\) is bounded by some universal constant.
\end{assumption}

Define the expected second moments of source and target covariates as 
\[
\bSigma_{\cS}  = \EE_{z \sim \cQ_\cS}[\phi(z) \otimes \phi(z)], 
\quad  
\bSigma_\cT = \EE_{z \sim \cQ_\cT} [\phi(z)\otimes \phi(z)].
\]
We now present assumptions on $\bSigma_\cT$ and the kernel $K$.

\begin{assumption}[Polynomial eigenvalue decay]\label{Assumption; eigenvalue decay}\label{assumption; boundedness}
We assume that \(\bSigma_\cT\)'s ordered eigenvalues \(\mu_1 \geq \mu_2 \geq \dots\) satisfy \(\mu_j \lesssim j^{-2\ell}\) for some \(\ell > \frac{1}{2}\).
We set \(\alpha = \frac{2\ell}{1+2\ell}\).
Also, the kernel is bounded as \(\sup_{z \in \Zcal}K(z,z) \leq \xi\) for some universal constant \(\xi>0\).  
%For simplicity, we assume \(\xi\) is bounded by some universal constant.
\end{assumption}

The first two assumptions are widely used in the CATE literature \citep{kunzel2019metalearners,curth2021nonparametric,kennedy2020towards,kennedy2022minimax}. Assumption~\labelcref{Assumption; eigenvalue decay} is also widely used in the KRR literature.
It is well known that the Sobolev space and Besov spaces satisfy this assumption \citep{zhang2023optimality,fischer2020sobolev}.  
In general, for \(H^k(\cZ)\) with \(\cZ \subset \RR^d\) and \(k > \frac{d}{2}\), we have \(\ell = \frac{k}{d}\).
Therefore, both Sobolev kernels and Matérn kernels satisfy this condition.
Moreover, the limiting neural tangent kernel (NTK) of overparameterized neural networks is known to exhibit a polynomially decaying spectrum with \(\ell =\frac{d+1}{2d}\) for the covariate dimension \(d\) \citep{geifman2020similarity,bietti2020deep}.
In the lazy-training regime, this kernel-level view carries over to overparameterized neural-network predictors: Lemma~12 of \citet{li2024eigenvalue} establishes uniform \(L^\infty\) closeness between the trained network and the corresponding limiting-NTK KRR/kernel-regression predictor under suitable width and scaling conditions.
Boundedness of the kernel is also widely assumed \citep{wainwright2019high,fischer2020sobolev,wang2026pseudo}. 
\normalcolor

\subsection{Weak Overlap}\label{subsection: weak overlap}
\noindent
%We now explain \emph{weak overlap}, one of our core assumptions. 
Our setting involves two types of distributional shifts: one between treated covariates and control covariates, and the other between source and target covariates. In this subsection, we introduce the notion of \emph{weak overlap} for both shifts.
%, and then compare it with the commonly used overlap condition.

Existing works on causal inference usually assume \emph{positivity} or \emph{strong overlap}, which requires the propensity score to be bounded away from \(0\) and \(1\).  
We relax this assumption by allowing the propensity score to approach \(0\) or \(1\), and even admit singular cases.  
Below, we present our \emph{weak treatment overlap} assumption.
To facilitate analysis, we define two second-moment operators associated with the treated and control groups:
\[
\Sigmatreated= \EE_{(z,a,y) \sim \cQ^\star_S}[\phi(z) \otimes \phi(z) \one(a =1)], 
\quad  
\Sigmacontrol := \EE_{(z,a,y) \sim \cQ^\star_S}[\phi(z) \otimes \phi(z) \one(a=0)].
\]
%See Section~\labelcref{section: groundwork} for more details.

\begin{assumption}[Weak overlap]\label{assumption; weak treatment overlap}\label{assumption; overlap source target}  
%\red{
There exist parameters \(R,B \geq 1\), possibly depending on the problem instance or on \(n\), such that the following conditions hold:%}
\begin{itemize}
\item Treatment overlap: $\Sigmacontrol  \preceq R(\Sigmatreated + \frac{\xi}{n} \Ib )$ and
$\Sigmatreated \preceq R (\Sigmacontrol + \frac{\xi}{n} \Ib )$

\item Source-target overlap: $\bSigma_{\cT} \preceq B (\bSigma_{\cS} + \frac{\xi}{n} \Ib )$.
\end{itemize}
\end{assumption}

We will show that the treatment overlap assumption can be implied by the commonly used positivity condition, while still allowing for singular propensity scores. 
% To the best of our knowledge, this is the first relaxation of the positivity condition in kernel-based CATE estimation. 
The source-target overlap condition is used in the study of KRR under covariate shift \citep{ma2023optimally,wang2026pseudo}. The parameter \(B\) controls the amount of covariate shift.
% When the source and target have the same distribution, we have $B = 1$.
% We treat \(R,B\) as important quantities and explicitly specify their dependence in subsequent results. 

\begin{remark}
For nonparametric estimation of the CATE function, existing theoretical results rely on the strong positivity assumption \citep{kennedy2020towards,kennedy2022minimax,curth2021nonparametric,gao2020minimax}. 
It is also worth mentioning that the assumption can be relaxed if the goal is to estimate functionals of CATE, such as the ATE \citep{mou2023kernel,ma2022testing}.
\end{remark}

Next, we provide examples illustrating Assumptions~\labelcref{assumption; weak treatment overlap,assumption; overlap source target}.
The following two examples illustrate that our assumption is well satisfied under the classical strong overlap conditions.

\begin{example}[$R$: Bounded propensity score]\label{example; positivity}
If \(\pi(z) \in [\kappa,1-\kappa]\) for some \(\kappa>0\), then the first part of Assumption~\labelcref{assumption; weak treatment overlap} holds with \(R \leq \frac{1}{\kappa}\). 
\end{example}

\begin{example}[$B$: Bounded source-target density ratio]
If the density ratio satisfies $\frac{\mathrm{d} Q_{\mathcal{T}}}{\mathrm{d} Q_{\mathcal{S}}} \leq B'$ for some constant $B' > 0$, then the second part of Assumption~\ref{assumption; overlap source target} holds with $B = B'$.
\end{example}

Next, we show that our weak overlap assumption also holds under a singular propensity score, highlighting its powerful relaxation. 
In the following example, the density ratio between the treated and control covariates is unbounded. 
% This scenario has not been explored in causal inference studies using KRR. 
%We claim to be the first to generalize the concept of overlap in this way.

\begin{example}[$R$: Singular propensity score] \label{example; singular propensity 1}
Consider the setting where $\cZ=[0,1]$ and the propensity score is $\pi(z) = z$. 
Suppose that $\cF$ is the Sobolev space $H^1([0,1])$ and the source covariate density is bounded away from zero and infinity. 
Then, the first part of Assumption~\labelcref{assumption; weak treatment overlap} is satisfied with $R \asymp n^{1/3}$.
\end{example}

% Once again, we emphasize that for \(\pi(z)=z\), the density ratio between the treated and control covariates is unbounded. 
Finally, we present an example of \(B\) under a singular source-target density ratio. 
This example is studied in \cite{tuo2024asymptotic,wang2026pseudo}.

\begin{example}[$B$: Dirac target distribution]\label{example; dirac target}
Consider a scenario where \(\cQ_\cT\) is a Dirac point measure at some point \(z_0\). 
For \(\cF = H^1([0,1])\) and source covariates whose density is bounded above and below, the second part of Assumption~\labelcref{assumption; overlap source target} holds with \(B \asymp n^{\frac{1}{2}}\).
\end{example}

In singular cases (e.g., point mass target distributions or singular propensity scores), \(B\) and \(R\) typically exhibit dependency on \(n\), which is indeed intrinsically linked to the function class. This connection arises because the geometric properties of the function space govern the estimator's ability to extrapolate beyond the observed support (or regions of strong overlap). 
A key distinction of our work is that we do not limit \(B\) or \(R\) to bounded constants, but instead treat them as key parameters in our non-asymptotic analysis.
Proofs for the examples above are presented in Appendix~\labelcref{section: proof weak overlap exampls}.

\subsection{Results Overview}\label{subsection: summary of results}
\noindent
In this section, we give an overview of our main results. 
We study regimes where \(\|h^\star\|_\cF\) is smaller than \(\|f_0^\star\|_\cF\) and \(\|f_1^\star\|_\cF\), i.e.~the contrast function has lower structural complexity than the nuisance. We do not assume any upper bound on \(\|f_0^\star\|_\cF\) or \(\|f_1^\star\|_\cF\). 
% In contrast, existing studies need the nuisance functions to lie in a fixed-radius RKHS ball. 
%{\color{black}
These RKHS norms are non-asymptotic, problem-dependent complexity measures determined by the kernel and response functions. 
%\red{
Recent non-asymptotic studies similarly treat Hilbert/RKHS norms and \(\ell_2\) norms as key complexity parameters rather than universal constants \citep{iwazaki2025improved,lee2024unified}.%}
For example, writing \(z=(z^1,\ldots,z^d)\), in \(H^\beta([0,1]^d)\), a nuisance of the form \(z\mapsto \sin(\sum_{j=1}^d z^j)\) has norm of order \(d^{\beta/2}\). By contrast, if the CATE depends only on an \(s\)-dimensional low-rank or intrinsic representation, e.g. \(h^\star(z)=\sin(\sum_{j=1}^s z^j)\) with \(s<d\), then its norm is of order \(s^{\beta/2}\). Thus, large nuisance norms may reflect ambient dimension, domain geometry, smoothness, or amplitude, while \(\|h^\star\|_\cF\) can remain small because of its simpler structure. %}

To investigate the optimal performance, consider an imaginary scenario where we observe both potential outcomes \(y_i(0)\) and \(y_i(1)\) for all \(i \in [n]\), and there is no covariate shift.
% (i)  We observe both potential outcomes \(y_i(0)\) and \(y_i(1)\) for all \(i \in [n]\); (ii) There is no covariate shift.
%\begin{itemize}
%\item We observe both potential outcomes \(y_i(0)\) and \(y_i(1)\) for all \(i \in [n]\).
%\item There is no covariate shift.
%\end{itemize}
In this case, the problem reduces to classical KRR, whose MSE is of order
\begin{equation}\label{equation: imaginary MSE}
\tilde{\cO}\bigl(n^{-\alpha} \|h^\star\|_{\cF}^{2(1-\alpha)} \bigr).
\end{equation}
See \citet{wainwright2019high,fischer2020sobolev}.
For Sobolev kernels, this rate is also known to be minimax optimal \citep{green2021minimax}.
Our setting is more challenging because missing outcomes prevent direct regression on \(h^\star\), \emph{making it difficult to adapt to \(\|h^\star\|_\cF\)}, while weak treatment and source--target overlaps (Assumption~\labelcref{assumption; weak treatment overlap,assumption; overlap source target}) limit the usable source sample size.
Hence, the \emph{effective} sample size is determined by weak overlap, which influences the learning rate.
We define our effective sample size by accounting for these two overlaps. 
\begin{definition}[Effective sample size]
We define the effective sample size under the two overlaps (Assumption~\labelcref{assumption; weak treatment overlap,assumption; overlap source target}) as 
\[
n_{\operatorname{eff}}:= \frac{n}{BR}.
\]
\end{definition}
The quantity \(n_{\operatorname{eff}}\) plays a central role in the MSE bound.
Later in Section~\labelcref{subsection: lower bound}, we prove that the effective sample size is necessary and that it aligns with the lower bound.

We now present our main result regarding the MSE.
%Let \(M = \max(\|f_0^\star \|_\cF,\|f_1^\star \|_\cF) \). 

\begin{theorem}[MSE bound of final model]\label{theorem; main theorem}
Suppose we run \texttt{COKE} under Assumptions~\labelcref{assumption; boundedness,assumption; consistency and unconfoundedness,assumption; subGaussian noise,assumption; overlap source target,assumption; weak treatment overlap,Assumption; eigenvalue decay} with $\lambda_{0,0} = \lambda_{0,1} = \tilde{\lambda}_0 = \tilde{\lambda}_1 = \frac{\xi \log n}{n}$ and $\Lambda_1 = \{ 2^\ell \frac{\xi \log n}{n} \mid \ell =0,1, \dots, L\}$, where $L = 2\lceil \log n \rceil$.
We further assume that \(n > BR\) and that \(\| h^\star\|_{\cF}\) is bounded by some universal constant. 
%We do not impose any bound on \(\|f_0^\star\|_\cF\) or \(\|f_1^\star\|_\cF\) and 
Let \(M = \max\bigl(\|f_0^\star\|_\cF,\|f_1^\star\|_\cF\bigr) \). 
Then, with probability at least \(1-n^{-10}\), the MSE of our final model satisfies:
\[
\cE_\cT(\hat{h}_{\final}) 
\lesssim n_{\operatorname{eff}}^{-\alpha} \norm{h^\star}^{2(1-\alpha)}_\cF + M^2 \Bigl(\frac{1}{n_{\operatorname{eff}}}+ \frac{R}{n_\cT}\Bigr).
\]
Here, \(\lesssim\) hides absolute constants, \(\sigma,\xi\), and logarithmic factors.
\end{theorem}

The first (leading) term in the bound matches the sharp rate \eqref{equation: imaginary MSE} when replacing \(n\) with the effective sample size \(n_{\operatorname{eff}}\). 
Moreover, this term alone aligns with the lower bound presented in Section~\labelcref{subsection: lower bound}.
The second error term decays at a fast rate ($n_{\operatorname{eff}}^{-1} + R n_{\cT}^{-1}$), diminishing much more rapidly than the nonparametric first term. 
%\red{
More explicitly, the target-sample component of this remainder is dominated by the leading term whenever \(n_\cT \gtrsim M^2R n_{\operatorname{eff}}^\alpha\).%}
Given that unlabeled target samples are typically abundant, in the regime of sufficiently large \(n\) and \(n_\cT\), the bound is dominated by the leading term.
Thus, our result can be interpreted as an oracle inequality with \emph{adaptivity}, where the second term represents a negligible overhead. 
In the regime \( \|h^\star\|_\cF \ll \|f_0^\star \|_\cF, \|f_1^\star\|_\cF \), the influence of nuisance-function complexity (\(M\)) remains mild, ensuring that the rate depends primarily on the structural complexity of the contrast function. 

% We emphasize that even if the source and target share the same response functions ($f_1^{\star}, f_0^{\star}$), the unknown covariate shift necessitates a careful choice of regularizer to ensure optimality. 
% Our setting is doubly challenging as it involves two such shifts. 
% A key contribution of this work is demonstrating that our approach effectively adapts to both unknown shifts. 
% The standard scenario without covariate shift is recovered as a special case by setting $B=1$.

%\red{
When \(\|f_0^\star\|_\cF,\|f_1^\star\|_\cF,\|h^\star\|_\cF \lesssim 1\), this gives the optimal dependence on \(n,B,R\) through \(n_{\operatorname{eff}}=n/(BR)\). Our result also covers regimes in which \(M\) is not constant, showing some robustness to large-\(M\) settings.
% The dependence on \(M\) remains mild because \(M\) enters only through the fast-rate remainder term. 
As discussed above, \(M\) is a problem-instance complexity parameter that may encode smoothness, domain geometry, ambient dimension, localization, and signal amplitude.
This perspective is also related to nuisance-model misspecification: when the true nuisance functions are approximated by elements of a working RKHS, the Hilbert norms of these approximants reflect the finite-sample approximation complexity.%}

Comparison with existing literature further highlights this contribution. While \citet{nie2021quasi} and \citet{foster2023orthogonal} also investigated CATE estimation using KRR, their theoretical guarantees rely on explicit propensity score modeling and require sufficient convergence rates for nuisance estimators; moreover, their performance degrades significantly under weak overlap. 
% In particular, \citet{nie2021quasi} is not optimal in the overlap parameter \(R\), with an \(R^2\) dependence appearing in the leading term, while \citet{foster2023orthogonal} requires \(L^4\)-type control of nuisance estimation errors, which can be challenging under weak overlap.
In addition, even employing separate regression and the method of \cite{wang2026pseudo} yields a bound of $n_{\operatorname{eff}}^{-\alpha}M^{2(1-\alpha)}$. This implies that our method achieves a faster convergence rate when $\max(\|f_0^\star\|_\mathcal{F}^2, \|f_1^\star\|_\mathcal{F}^2) \lesssim n_{\operatorname{eff}}^{1/(2\ell)}$.

%\red{
For Sobolev kernels under nearly-uniform design, covariate shift can also be handled through an $L^\infty$-based analysis~\citep{fischer2020sobolev,tuo2024asymptotic}. For $H^\beta(\cZ)$ with $\cZ \subset \RR^d$, an $L^\infty$-based argument yields the $B$-free rate $(R/n)^{\frac{2\beta-d}{2\beta}}$, while our target-risk bound scales as $(BR/n)^{\frac{2\beta}{d+2\beta}}$. Therefore, our bound is sharper when $B \lesssim (n/R)^{d^2/(4\beta^2)}$, which covers a substantial moderate-smoothness regime: if $d/2<\beta\le d$, then the exponent lies between $1/4$ and $1$. The \(L^\infty\) route can be preferable when \(\beta/d\) is very large, but this requires extreme smoothness, nearly-uniform design, and strong Sobolev-type assumptions.%} 

The proof is deferred to Appendix~\labelcref{section: proof main theorem} and consists of two parts: 
(i) analysis of the RA learner (Section~\ref{sec-MSE-bound}), and 
(ii) analysis of model selection (Section~\labelcref{subsection: model selection}).

\begin{remark}[Results for ATE]\label{remark; ATE}
We may also be interested in the ATE on the target distribution, defined as \(\Delta_\cT^\star := \EE_{z \sim \cQ_\cT}[h^\star(z)]\). 
Our framework readily accommodates this estimand. Specifically, by invoking Algorithm~\ref{algorithm: main} (using the inputs from Theorem~\ref{theorem; main theorem}) and substituting the objective \(L(h)\) in the subroutine (Algorithm~\ref{algorithm: model selection}) with the expression below, we obtain an estimator for the ATE:
\begin{align}\label{equation: model selection ATE}
L(h) =  \bigg|\frac{1}{n_\cT}\sum_{i=1}^{n_\cT} \tilde h(z_{0i}) - \frac{1}{n_\cT} \sum_{i=1}^{n_\cT} h(z_{0i}) \bigg|^2.
\end{align}
With the resulting \(\hat h_{\operatorname{final}}\), the ATE estimator is $\hat\Delta_\cT := \frac{1}{n_\cT} \sum_{i=1}^{n_\cT} \hat h_{\operatorname{final}}(z_{0i})$.
%  \begin{align*}
% \hat\Delta_\cT := \frac{1}{n_\cT} \sum_{i=1}^{n_\cT} \hat h_{\operatorname{final}}(z_{0i})
%  \end{align*}
When \(n_\cT \gtrsim n/B\), the following optimal error bound holds with high probability:
\begin{align*}
|\hat \Delta_\cT -\Delta^\star_\cT| = \tilde\cO(\frac{M}{\sqrt{n_{\operatorname{eff}}}} )
\end{align*}
Importantly, the eigenvalue decay in Assumption~\ref{Assumption; eigenvalue decay} is not required; see Appendix~\ref{section: apdx ATE} for details. 
\end{remark}

\subsection{MSE Bound of RA Learner}\label{sec-MSE-bound}
\noindent
Next, we examine the MSE bound of the estimator from the RA learner (Algorithm~\labelcref{algorithm: RA learner}). 
It requires a triple of regularizers \(\bm{\lambda} = (\lambda_{0,0}, \lambda_{0,1}, \lambda_{1})\) with \(\lambda_{0,0}, \lambda_{0,1}, \lambda_{1} >0\). 
Below, we present a theorem on the MSE bound of the RA learner with regularizers \(\bm{\lambda}\). 
Before that, we define an operator:
\[
\Sbar_\lambda := (\bSigma_{\cS} + \lambda \Ib)^{-\frac{1}{2}} \bSigma_\cT (\bSigma_{\cS} + \lambda \Ib)^{-\frac{1}{2}}, 
\quad \text{for any } \lambda > 0.
\]
Also, for \(\bm{\lambda} = (\lambda_{0,0}, \lambda_{0,1}, \lambda_{1})\), we define
\[
\Rcr(\bm{\lambda}) 
:= 
R\|\Sbar_{\lambda_{1}}\|_{\op}\bigl(\lambda_{0,1}\|f^\star_1\|_{\cF}^2 + \lambda_{0,0}\|f^\star_0\|_{\cF}^2\bigr)
+
\lambda_{1}\|\Sbar_{\lambda_{1}}\|_{\op}\|h^\star\|_{\cF}^2
+
\sig^2 \frac{R\operatorname{Tr}(\Sbar_{\lambda_{1}})}{n}\log n.
\]
The next theorem shows the MSE bound of \(\hat{h}_{\bm{\lambda}}\) from Algorithm~\labelcref{algorithm: RA learner}.

\begin{theorem}[MSE bound of RA learner estimators]\label{theorem; MSE bound RA learner}
Assume we run Algorithm~\labelcref{algorithm: RA learner} with the regularizers \(\bm{\lambda} = (\lambda_{0,0}, \lambda_{0,1}, \lambda_{1})\) satisfying \(\lambda_{0,0},\lambda_{0,1},\lambda_1\geq \xi\log n/n\) and the dataset \(\cD_1\).
Under the same assumptions as in Theorem~\ref{theorem; main theorem}, the MSE of \(\hat{h}_{\bm{\lambda}}\) is bounded by
\[
\cE_\cT(\hat{h}_{\bm{\lambda}}) \lesssim \Rcr(\bm{\lambda})
\]
with probability at least \(1-2n^{-11}\).
Here, \(\lesssim\) hides absolute constants.
\end{theorem}

Its proof is deferred to Appendix~\labelcref{section: proof RA learner}. 
This theorem provides an upper bound on the MSE of $\hat{h}_{\bm{\lambda}}$.
Recall that we generate multiple CATE estimators \(\cH := \{\hat{h}_{\bm{\lambda}} \mid \bm{\lambda}\in \bm{\Lambda}\}\) in Algorithm~\labelcref{algorithm: RA learner}.
The following corollary bounds the best MSE over these candidates.

\begin{corollary}[Optimal MSE bound among candidates]\label{corollary; optimal MSE bound}
Assume that we run \texttt{COKE} under the same setup as in Theorem~\ref{theorem; main theorem}. Then, with probability at least \(1-2n^{-11}\), the following holds:
\[
\inf_{\bm{\lambda} \in \bm{\Lambda}} \cE_\cT(\hat{h}_{\bm{\lambda}}) 
\lesssim  
\Bigl(\frac{BR}{n}\Bigr)^{\alpha}\|h^\star\|_{\cF}^{2(1-\alpha)}(\log n)^{\alpha}
+\frac{BR}{n} M^2\log n.
\]
Here, the symbol \(\lesssim\) hides absolute constants and the parameters \(\sigma\) and \(\xi\).

\end{corollary}

Its proof is deferred to Appendix~\labelcref{section: proof RA learner}. 
We observe that the optimal MSE over the grid of regularizers matches the lower bound, as will be discussed in Theorem~\ref{theorem; lower bound}. A natural question is how to choose the regularizer $\bm{\lambda} \in \bLambda$ that minimizes $\cE_\cT(\hat{h}_{\bm{\lambda}})$. The next section provides an oracle inequality for our model selection procedure.

\subsection{Model Selection and Oracle Inequalities}\label{subsection: model selection}
\noindent
We now present the oracle inequalities for our model selection procedure (Algorithm~\labelcref{algorithm: model selection}), which establish the MSE bound for the final selected model. 
Define a key quantity
\[
\Ocr := \bigl(\xi M^2 + \sigma^2\bigr)\Bigl(\frac{BR}{n} + \frac{R}{n_{\cT}}\Bigr)\log(n n_\cT).
\]
This plays an important role in the analysis of model selection.
We now obtain the oracle inequality for MSE:

% {\color{cyan}[To Seok-Jin: The above in-sample MSE bound plays no role in the main body of this paper. Readers have no idea why it is introduced or how it connects to other results. Move it to the appendix and make edits to this subsection to ensure consistency as you see needed.]}

\begin{proposition}[Oracle inequality for MSE of the selected model]\label{proposition; oracle inequality MSE}
Assume that we run \texttt{COKE} under the same setup as in Theorem~\ref{theorem; main theorem}.
Then, with probability at least \(1 - 3n^{-11}\), the following holds:
\[
\cE_{\cT}(\hat{h}_{\operatorname{final}}) 
\lesssim 
\inf_{\bm{\lambda} \in \bm{\Lambda}} \cE_\cT(\hat{h}_{\bm{\lambda}})  
+ 
\Ocr,
\]
where \(\lesssim\) hides absolute constants.
\end{proposition}

Combining this result with Corollary~\labelcref{corollary; optimal MSE bound}, we obtain Theorem~\labelcref{theorem; main theorem}. Compared to Corollary~\labelcref{corollary; optimal MSE bound}, the above oracle inequality contains the additional term \(\Ocr\), reflecting the cost of model selection. 
However, since \(\Ocr = \tilde{\cO}\bigl(\frac{1}{n} + \frac{1}{n_\cT}\bigr)\), it remains a negligible bound relative to the leading term in Theorem~\labelcref{theorem; main theorem}.
% Hence, our algorithm achieves strong adaptivity, attaining near-optimal performance with model selection, as stated in Theorem~\labelcref{theorem; main theorem}.
The proof is presented in Appendix~\labelcref{section: proofs oracle inequalities}.

\subsection{Lower Bound Results}\label{subsection: lower bound}
\noindent
Finally, we present results related to the lower bound to demonstrate the optimality of our estimator. 
First, we define an \((R,B)\)-bounded instance.

\begin{definition}[\((R,B)\)-bounded instance]
For the distribution of source covariates \(\cQ_\cS\), target covariates \(\cQ_\cT\), and source propensity score \(\pi(\cdot)\), we say they form a \((R,B)\)-bounded instance if \(\frac{1}{R} \leq \pi(\cdot)\leq 1-\frac{1}{R}\) and \(\frac{\de \cQ_\cT}{\de \cQ_\cS} \leq B\).
\end{definition}

We can see that every \((R,B)\)-bounded instance satisfies the weak overlap (Assumption~\labelcref{assumption; overlap source target}). 
Now, we present our result for the lower bound.
{\color{black}
% We state the lower-bound instance for integer order \(k\), corresponding to \(\ell=k\).

\begin{theorem}[Lower bound]\label{theorem; lower bound}
Fix an integer \(k\geq 1\), and set \(\alpha = \frac{2k}{1+2k}\). 
For any \(W>0\), \(R\geq 2\), and \(B\geq 1\) with \(n\geq BR\), there exists a triple \((\cQ_\cS, \cQ_\cT, \pi)\) forming an \((R,B)\)-bounded instance on \(\cZ =[-1,1]\) and a reproducing kernel \(K(\cdot,\cdot)\) whose target covariance operator satisfies Assumption~\labelcref{Assumption; eigenvalue decay} with \(\ell=k\) such that
\[
\inf_{\hat{h}} \sup_{\substack{(f_0^\star, f_1^\star, h^\star): \|h^\star\|_\cF \leq W}} 
\EE\bigl[\cE_\cT(\hat{h})\bigr] 
\gtrsim
\Bigl(\frac{BR}{n}\Bigr)^{\alpha}W^{2(1-\alpha)}.
\]
\end{theorem}
}
The proof is deferred to Appendix~\labelcref{section: lower bound proof}.
This result confirms that the leading term in our upper bound (Theorem~\labelcref{theorem; main theorem}) matches the minimax lower bound. 
{\color{black}
Specifically, for these \(\ell=k\) instances, the optimal rate scales as $n_{\operatorname{eff}}^{-\alpha} \|h^\star\|_\cF^{2(1-\alpha)}$ with \(\alpha = \frac{2k}{1+2k}\). 
}
We analyze this lower bound through its two key components, linking back to the intuition provided in Section~\labelcref{subsection: summary of results}:

\begingroup 
\begin{enumerate}
\item \textbf{Structural complexity ($\|h^\star\|_\cF^{2(1-\alpha)}$):} 
This term corresponds to the \textit{oracle rate} for the regression problem where both potential outcomes are observed, as discussed in Equation~\eqref{equation: imaginary MSE}. 
Importantly, our result shows that in the RKHS framework, one can \textit{always} adapt to the intrinsic complexity of the CATE function $\|h^\star\|_\cF$, regardless of the complexity of nuisance functions.

This contrasts with other settings, such as the H\"{o}lder class setup in \citet{kennedy2022minimax}. There, achieving the oracle rate (when potential outcomes are known) typically requires smoothness conditions on nuisance terms; for example, when $f_a^\star \in C^\alpha$, $\pi \in C^\beta$, and $h^\star \in C^\gamma$, one needs $\alpha + \beta > \gamma \frac{d}{d+2\gamma}$. If these conditions are not met, the minimax lower bound becomes slower. 
% In our setting, by contrast, the lower bound is dictated purely by the CATE complexity and matches the upper bound without such conditions.

\item \textbf{Effective sample size ($n_{\operatorname{eff}}$):} 
The term involving $n_{\operatorname{eff}} = n/(BR)$ highlights the necessity of accounting for overlap. 
While standard CATE literature often treats the degree of overlap ($B, R$) as constant, our lower bound explicitly tracks these dependencies, demonstrating that the minimax rate is fundamentally constrained by the effective sample size determined by the overlap conditions.
%\red{
Thus, Theorem~\labelcref{theorem; lower bound} confirms that the $BR$ dependence in $n_{\operatorname{eff}}=n/(BR)$ is sharp.%}
\end{enumerate}

\endgroup

\section{Simulation Studies}\label{sec:simu}

\subsection{Setups and Benchmarks}
\noindent
We conduct extensive simulation studies to evaluate the finite-sample performance of \texttt{COKE} and compare it with existing approaches. 
% The simulation code and methods used in this study are publicly available on GitHub at \url{https://github.com/hongjiel/COKE}. 
For data generation, we set the source sample size $n$ and the target (unlabeled) sample size $n_{\cT}$ such that $n_{\cT}=n/4$. To generate the covariates $z$, we let $\mathcal{U}^+$ and $\mathcal{U}^-$ respectively denote uniform distributions over $(0,\pi)$ and $(-\pi,0)$. On the source $\cS$, we independently generate $z_{j}\sim \frac{S_B^{1/q}}{S_B^{1/q} + 1}\mathcal{U}^-+ \frac{1}{S_B^{1/q} + 1}\mathcal{U}^+$ for $j=1,\ldots,q$ and $z_{q+1}, \ldots, z_{p} \stackrel{\text{i.i.d.}}{\sim} \text{Uniform}(-\pi,\pi)$. On the target, we independently generate $z_{j}\sim \frac{1}{S_B^{1/q} + 1}\mathcal{U}^-+\frac{S_B^{1/q}}{S_B^{1/q} + 1}\mathcal{U}^+$ for $j=1,\ldots,q$ and $z_{q+1}, \ldots, z_{p} \stackrel{\text{i.i.d.}}{\sim} \text{Uniform}(-\pi,\pi)$. Here, we fix $p = 4$ and $q<p$ is the number of covariates subject to distributional shift between the source and the target. Also, the hyper-parameter $S_B$ controls the degree of covariate shift between the source and target, with a larger $S_B$ resulting in a weaker overlap.

Then we set the propensity score as $\pi(z) = \mathrm{expit}(S_R\sum_{j=1}^4z_j/8)$ and generate $a_i \mid z_i \sim \mathrm{Bernoulli}(\pi(z_i))$, where $S_R$ controls the overlap between the treatment and control groups on the source. We set the outcome models as
\[
f_a^\star(z) = c \cdot q^{-1}\sum_{i=1}^q\left[2 \left( \left| z_i \right| - \frac{\pi}{4} \right) \bold{1}\left( \left| z_i \right| \geq \frac{\pi}{2} \right) + \left| z_i \right| \bold{1}\left( \left| z_i \right| < \frac{\pi}{2} \right) \right]
+ \left(a - \frac{1}{2}\right) \cdot q^{-1}\sum_{i=1}^q\sin z_i,
\]
for $a = 0, 1$, which yields the true CATE function $h^\star(z) = q^{-1}\sum_{i=1}^q\sin z_i$. Note that $f_0^\star(z)$ and $f_1^\star(z)$ are less smooth and more complex than $h^\star(z)$ due to their absolute value terms of $z_i$. Then we generate $y_i \mid a_i, z_i \sim N(f_{a_i}^\star(z_i), 0.25)$. Here, $c$ controls the complexity of the nuisance model $f_a^\star$ relative to the CATE function $h^\star$. We consider simulation settings with varying hyper-parameters including the sample size $n_{\cT}=n/4$, the dimension of covariate shift $q$, the degree of covariate shift $S_B$, the degree of non-overlap between treatment and control $S_R$, and the complexity of the nuisance model $c$. Specifically, we first set $q=1$, $S_B = 10$, $S_R = 2$, $c = 1$, and $n_{\cT}=n/4=\lceil 350\sqrt{S_B}+60S_R+25 \rceil$. Then we vary each single parameter among $S_B$, $S_R$ and $c$ separately with the remaining parameters fixed and $n_{\cT}=n/4$ changing accordingly. We also vary the sample size $n_{\cT}=n/4$ with the others fixed, and vary $S_B$ under the setting $q=2$. In each setting, we evaluate the mean squared error of an estimator $\hat{h}(z)$, i.e.~  $\EE_{z \sim \cQ_\cT}|\hat{h}(z) - h^\star(z)|^2$, averaged over 100 repetitions.

We consider five methods in our studies including \texttt{COKE} and four benchmark methods--separate regression (\texttt{SR}), \textcolor{black}{R-Learner (\texttt{RLearner}) \citep{nie2021quasi},} DR-Learner for CATE (\texttt{DR-CATE}) \citep{kennedy2020towards}, and the ACW estimator tailored for CATE estimation (\texttt{ACW-CATE}), which is motivated by \cite{lee2023improving}. For fair comparison, KRR is used for all regression tasks in these methods. The \texttt{ACW-CATE} method extends the DML approach of \cite{kennedy2020towards} by incorporating a density ratio model to address covariate shift and form the efficient influence function for the CATE on the target sample in our setup. The implementation details of the benchmark methods are provided in Appendix \labelcref{sec:app:bench}.

For KRR, we use the Matérn kernel:
$
K(\boldsymbol z,\boldsymbol w) = \frac{4}{\sqrt{\pi} \rho} e^{-{2\sqrt{2} \|\boldsymbol{z} - \boldsymbol{w}\|_2}/{\rho}} $,
where $\rho$ is the scale parameter set as $5$. For \texttt{COKE}, we set the tuning parameters in Algorithm \labelcref{algorithm: main} as $\lambda_{0,0}=\lambda_{0,1} = \frac{1}{5n}$ and ${\bm \Lambda} = \{\frac{2^k}{5n} : k = 0, 1, \ldots, \lceil\log_2(5n)\rceil\}$. {\color{black}To utilize the data more effectively, we adopt a cross-fitting procedure on the splits; see Appendix \ref{sec:app:cross-fitting} for implementation details.}
Our theoretical results in Section \labelcref{section: main results} on Algorithm \labelcref{algorithm: main} can easily be extended to the cross-fitted version; see Appendix~\labelcref{subsection: proofs for cross fitting}.
To ensure fair comparison, cross-fitting is also used to implement the benchmark methods \texttt{DR-CATE}, \texttt{ACW-CATE}, {and \texttt{R-Learner},} as recommended by existing works \citep{chernozhukov2018double,kennedy2020towards,nie2021quasi}.

\subsection{Results}
\noindent
{As shown in Figure \labelcref{fig:matern}, \texttt{COKE} delivers the best overall performance among the competing methods and shows strong advantages in challenging settings.} When $S_B$ increases, there is more severe covariate shift and weaker overlap between the source and target. In this scenario with $n_{\cT}=n/4=\lceil 350\sqrt{S_B}+60S_R+25 \rceil$, the performance of \texttt{COKE} remains stable and better than other methods \textcolor{black}{except for the nearly no-shift case $S_B=1$, where \texttt{R-Learner} is only marginally better.} For example, when $S_B = 25$, the relative efficiency of \texttt{COKE} compared to \texttt{DR-CATE} is $1.61$, 
\textcolor{black}{and its relative efficiency compared to \texttt{R-Learner} is $1.29$.} Unlike \texttt{COKE}, the mean squared errors of \texttt{DR-CATE}, \texttt{ACW-CATE}, and \texttt{R-Learner} grow significantly with $S_B$ even under an increasing sample size $n$ proportional to $\sqrt{S_B}$. This demonstrates the superior robustness of \texttt{COKE} to the weak overlap issue.%, compared to the benchmark methods.

\begin{figure}[h]%[htbp!]
\centering
\includegraphics[width=16cm]{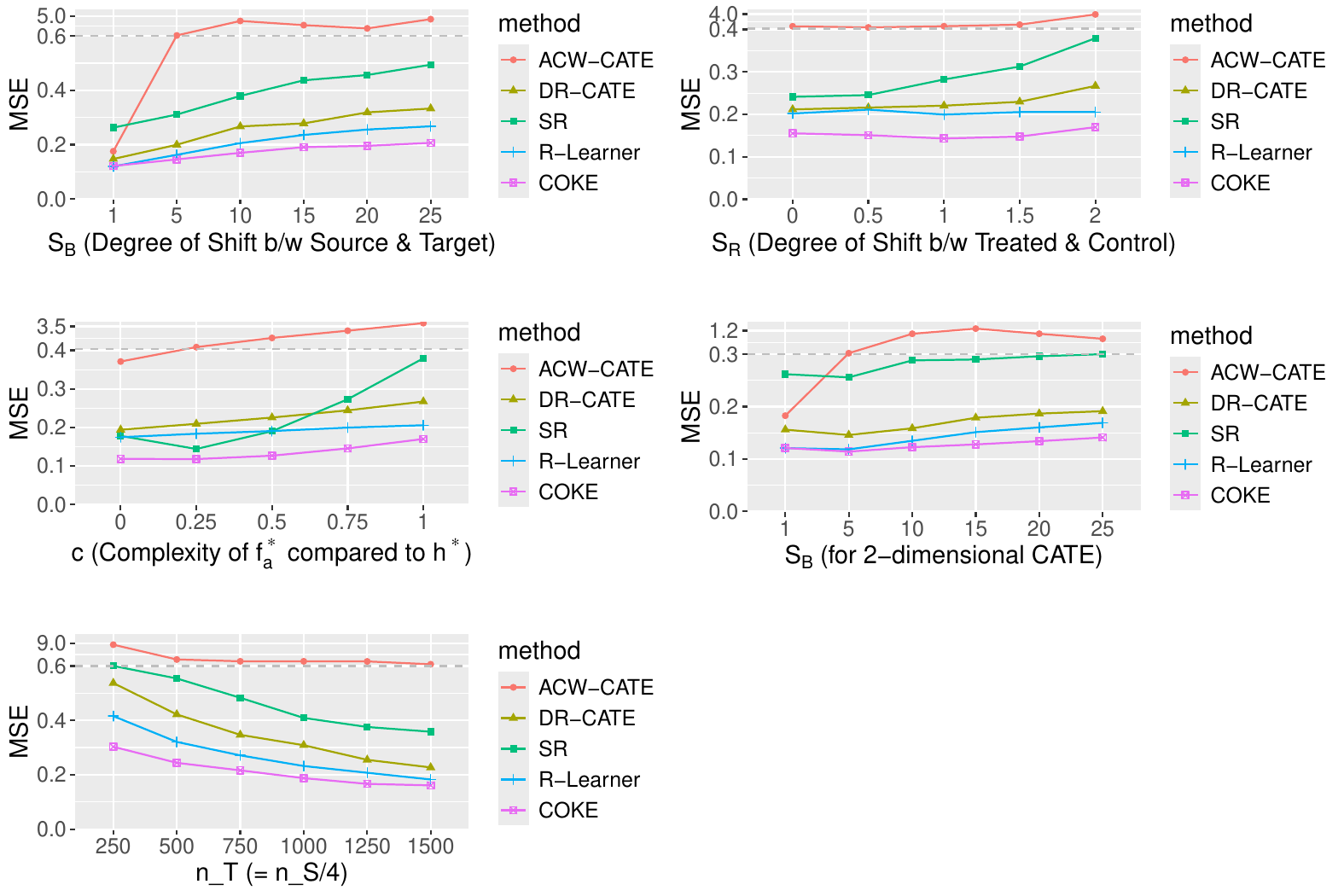}
\caption{Performance of \texttt{COKE}, \texttt{ACW-CATE}, \texttt{DR-CATE}, \texttt{SR}, {and \texttt{R-Learner}} across varying simulation settings. Panels show the average MSE as a function of: (i) $S_B$ (degree of covariate shift between source and target) for $q = 1$, (ii) $S_R$ (degree of shift between treatment and control groups), (iii) $c$ (complexity of outcome models relative to the CATE), (iv) $S_B$ for $q = 2$ (weak overlap on two-dimensional covariates), and (v) $n_{\cT}=n/4$.}
\label{fig:matern}
\end{figure}

As $S_R$ increases, reflecting more severe non-overlap between treated and control groups, \texttt{COKE} also maintains better performances over the benchmarks. \textcolor{black}{For example, when $S_R = 2$, the relative efficiency of \texttt{COKE} compared to \texttt{DR-CATE} is $1.57$, and its relative efficiency compared to \texttt{R-Learner} is $1.21$.} This shows \texttt{COKE}’s effectiveness in handling the weak overlap between the treated and control groups. The inferior performance of DR-CATE and ACW-CATE under weak overlap can be attributed to their reliance on propensity score estimation. In settings with severe covariate shift (large $S_B$) or weak treatment-control overlap (large $S_R$), the estimated propensity scores often approach extreme values near $0$ or $1$. This can lead to large inverse probability weights in the DR estimator and inflate its variance. In contrast, COKE avoids explicit propensity score modeling and instead leverages a two-stage KRR, making it more robust to such instability. This highlights a key advantage of our KRR-based approach in the presence of weak overlap.

As $c$ increases, there arises higher complexity of the outcome models compared to the CATE, \texttt{COKE} shows consistently smaller estimation error compared to the benchmarks across different $c$. \textcolor{black}{For example, when $c = 0.5$, the relative efficiencies of \texttt{COKE} compared to \texttt{DR-CATE}, \texttt{SR}, and \texttt{R-Learner} are all above $1.5$.} Moreover, the improvement of \texttt{COKE} over \texttt{SR} generally becomes more significant as $c$ gets larger, which demonstrates \texttt{COKE}'s better adaptivity to complex outcome regression functions and is consistent with our theoretical results in Section~\labelcref{subsection: summary of results}. We also consider a different setup with $q = 2$, including two covariates subject to weak overlap between the source and the target and making the outcome and CATE models more complex. As shown in the fourth panel of Figure \labelcref{fig:matern}, \texttt{COKE} again displays consistently lower mean squared errors compared to \texttt{ACW-CATE}, \texttt{DR-CATE}, {\texttt{R-Learner}} and \texttt{SR} across various values of $S_B$ when $q=2$. Additionally, we vary the sample sizes of the target and source data while keeping their ratio constant at $n / n_{\cT} = 4$ and present the results in the last panel of Figure \labelcref{fig:matern}. As the sample sizes increase, \texttt{COKE} exhibits a similar rate of risk reduction as other methods, and maintains the smallest estimation error among all methods. Finally, we compare the cross-fitting version of \texttt{COKE} with the original Algorithm \labelcref{algorithm: main} as detailed in Appendix \labelcref{sec:app:simu}. In our setup with $q=1$ and varying $S_B$, the cross-fitting version displays around $13.5$--$15\%$ lower risk than the original data-splitting version.%, which is a noticeable improvement. 

\section{Real-World Example}
\subsection{NHANES Study}\label{sec:real:nhanes}

We apply our method to the National Health and Nutrition Examination Survey (NHANES) data to investigate the heterogeneous effects of dietary habits on health outcomes. The NHANES is a major program of the National Center for Health Statistics (NCHS) designed to assess the health and nutritional status of adults and children in the United States, see \url{https://wwwn.cdc.gov/nchs/nhanes}. Beginning in 1999, NHANES became a continuous program that collects data on a nationally representative sample of several thousands of individuals each year. Since the survey samples distinct subjects in each cycle, the population demographics and environmental factors naturally shift over time, necessitating robust transfer learning methods to transport causal findings from historical cohorts to current populations, especially when the gap of years between them is large.

According to the World Health Organization (WHO), energy intake should be in balance with energy expenditure, and total fat intake should not exceed 30\% of total energy intake to avoid unhealthy weight gain and protect against noncommunicable diseases \citep{who2003diet,hooper2015effects}. We aim to estimate the CATE of excessively high fat intake on blood pressure. {\color{black} We set the outcome $y_i$ as the systolic blood pressure level (averaged over three repeated measures). Furthermore, we define the treatment $a_i$ as a binary indicator for fat intake exceeding 40\% of total energy. This binarization was an intentional methodological choice as the threshold is set substantially above the WHO's recommended upper limit of 30\% to both isolate the impact of clearly excessive fat consumption. This choice also induces a scenario of weak treatment overlap, thereby testing the adaptivity of our method.}

We observe a set of covariates $z_i$ that includes {\em sex}, {\em age}, {\em smoking status}, {\em education level}, and {\em alcohol use}. We consider a transfer learning setup where the source data consists of participants from the 2001–2002 survey cycle ($n=3360$), while the target data consists of participants from the 2015–2016 survey cycle ($n_{\mathcal{T}}=3913$). The 14-year gap between the two surveys, combined with the fact that they contain disjoint sets of subjects, introduces potential distributional shifts in the covariates, making the direct application of source models to the target challenging. {\color{black} To visually assess the magnitude of this shift, we computed the estimated density ratios (fitted using random forest) between the source and target. As illustrated in Figure \ref{fig:densityrationhanes} in the Supplement, the distinct distributions of the log-density ratios confirm the presence of clear covariate shift between the 2001 and 2015 cohorts.}

In this real-world analysis, we treat the treatments and outcomes in the target data as unobserved during training and use them solely for validation. To evaluate the predictive performance of the CATE estimators, we construct an ``empirical gold-standard'' CATE predictor on the full target data (with treatments and outcomes), denoted as $\hat{s}_{0i}$, and calculate its Pearson and Spearman correlation coefficients with the learners under evaluation. Utilizing the validation treatment and outcome in the target data, we implement the \texttt{DR-CATE} learner of \cite{kennedy2020towards} as introduced in Appendix \ref{sec:app:bench} to obtain the empirical gold-standard $\hat{s}_{0i}$, shown to be the optimal choice if one could observe the treatment and outcome in the target data \citep{kennedy2022minimax}. To obtain the nuisance and CATE models when constructing $\hat{s}_{0i}$, we fit generalized additive models. Consistent with Section \labelcref{sec:simu}, we include \texttt{SR}, \texttt{DR-CATE}, \texttt{ACW-CATE} and \texttt{R-Learner} as the benchmark methods for comparison. All methods are implemented with cross-fitting. 

\begin{table}[htb!]
\centering
\begin{tabular}{|c|ccccc|}
\hline Metrics & \texttt{COKE} & \texttt{SR} & \texttt{DR-CATE} & \texttt{ACW-CATE} & \texttt{R-Learner} \\
\hline Spearman Cor with $\hat{s}_{0i}$ & ${\bf 0.56}_{0.012}$ & $0.40_{0.014}$ & $0.44_{0.013}$ & $0.35_{0.014}$  & $0.46_{0.014}$ \\
\hline Pearson Cor with $\hat{s}_{0i}$ & ${\bf 0.55}_{0.010}$  & $0.43_{0.014}$ & $0.45_{0.012}$ & $0.36_{0.014}$  & $0.44_{0.013}$  \\ 
\hline
\end{tabular}
\caption{\label{tab:res:nhanes:crossfit} Spearman and  Pearson correlation coefficients (subscribed with their empirical standard errors) between the empirical gold-standard $\hat{s}_{0i}$ and the CATE predictors obtained by cross-fitting in the NHANES study. %The nuisance and CATE models in $\hat{s}_{0i}$ are obtained using generalized additive models.
}
\end{table}

The results are summarized in Table \ref{tab:res:nhanes:crossfit}. \texttt{COKE} achieves the highest degree of concordance with the empirical gold-standard $\hat{s}_{0i}$ among all methods. Specifically, \texttt{COKE} attains a Spearman correlation of 0.56 and a Pearson correlation of 0.55. This significantly outperforms the benchmark methods. For instance, the separate regression (\texttt{SR}) approach yields correlations of only 0.40 and 0.43, respectively and {\color{black} \texttt{R-Learner} yields $0.46$ and $0.44$, all of which are at least $0.1$ smaller than \texttt{COKE}}. Furthermore, \texttt{COKE} shows a clear advantage over the DML-based approaches (\texttt{DR-CATE} and \texttt{ACW-CATE}), demonstrating its superior adaptability to the distributional shifts present between the 2001 and 2015 NHANES populations.

\subsection{401(k) Eligibility Study}
The impact of the 401(k) program has been extensively studied \citep{abadie2003semiparametric}. Unlike other plans like Individual Retirement Accounts (IRAs), 401(k) eligibility is solely determined by employers. As a result, unobserved individual savings preferences are unlikely to significantly affect eligibility for 401(k) plans. Nonetheless, factors like job choice, income, and age may still confound causal analyses of the 401(k) program. To address this, \citet{abadie2003semiparametric} and \citet{chernozhukov2018double} suggested adjusting for specific covariates related to job selection to treat 401(k) eligibility as exogenous. Examining the ATE of 401(k) eligibility on the overall net financial assets (NFA) is a key question addressed in previous studies such as \cite{abadie2003semiparametric} and \cite{chernozhukov2018double}. However, quantifying the CATE of 401(k) eligibility for individualized policy evaluation is an important yet overlooked problem. 

Motivated by this, we aim to learn the CATE of 401(k) eligibility $a_i$ on the NFA outcome $y_i$ using the data set from the 1991 Survey of Income and Program Participation. We include $7$ adjustment and effect modifying covariates in $z_i$ including {\em age, income, family size, education years, benefit pension status, participation in an IRA plan, and home ownership}. We consider a transfer learning setup with $n=5997$ source samples including all subjects in the original data set with their marital status being {\em married} and $n_{\cT}=3918$ {\em not married} subjects as the target data, with only their $z_i$ used for training and their treatment and outcome information used for validation and evaluation. We examine the severity of covariate shift between the source and target. In specific, we fit the logistic regression to obtain $\widehat{\omega}(z)$ as an estimate of the density ratio of the covariates $z$ between the source and target. Figure \labelcref{fig:densityratio} in the Supplementary Material displays the histograms of $\log_{10}\{\widehat{\omega}(z)\}$ from the source and target data sets, revealing severe covariate shift as evidenced by the minimal intersection between the source and target distributions. This distributional discrepancy significantly reduces the effective sample size of the source data to $399.01$.

\begin{table}[htb!]
\centering
\begin{tabular}{|c|ccccc|}
\hline Metrics & \texttt{COKE} & \texttt{SR} & \texttt{DR-CATE} & \texttt{ACW-CATE} & \texttt{R-Learner} \\
\hline Spearman Cor with $\hat{s}_{0i}$ & ${\bf 0.63}_{0.010}$ & $0.39_{0.014}$ & $0.26_{0.015}$ & $0.38_{0.014}$ & $0.45_{0.014}$ \\
\hline Pearson Cor with $\hat{s}_{0i}$ & ${\bf 0.55}_{0.026}$  & $0.38_{0.032}$ & $0.09_{0.10}$ & $0.32_{0.024}$ & $0.35_{0.034}$ \\ 
\hline
\end{tabular}
\caption{\label{tab:res:401:crossfit} Spearman and  Pearson correlation coefficients (subscribed with their empirical standard errors) between the empirical gold-standard $\hat{s}_{0i}$ and the CATE predictors obtained by cross-fitting in the 401(k) study.} %The nuisance and CATE models in $\hat{s}_{0i}$ are obtained using generalized additive models.}
\end{table}

We include \texttt{SR}, \texttt{DR-CATE}, \texttt{ACW-CATE} and \texttt{R-Learner} as the benchmark methods for comparison. In Table \labelcref{tab:res:401:crossfit}, we present the Spearman and Pearson correlation coefficients between the CATE predictors and the empirical gold-standard $\hat{s}_{0i}$ introduced in Section \ref{sec:real:nhanes}. Again, \texttt{COKE} achieves the highest degree of concordance with $\hat{s}_{0i}$. In specific, \texttt{COKE} attains more than 40\% higher Spearman correlation with $\hat{s}_{0i}$ as well as approximately 45\% higher Pearson correlation compared to all other methods. Finally, we also present the performance metric tables for both studies in Appendix~\labelcref{sec:app:real:results}, where all CATE estimators are constructed with just data-splitting but not cross-fitting. They show similar results as in Tables \labelcref{tab:res:nhanes:crossfit} and \labelcref{tab:res:401:crossfit} in the sense that \texttt{COKE} still achieves the best performance among all methods under comparison. Meanwhile, the cross-fitted version of \texttt{COKE} attains better performance than that without cross-fitting. {\color{black} To provide a more detailed view of the estimation performance, we include joint scatter plots of the COKE estimates versus the empirical gold standard in Figures \ref{fig:scat401k} and \ref{fig:scatnhanes} of the Supplementary Material. They illustrate the strong concordance between our estimator and the gold standard across the full range of treatment effects.}

\section{Discussion}
\noindent
We propose a novel methodology and theoretical framework for CATE estimation under covariate shift.
Our goal is twofold: (i) to adapt to the complexity of the CATE function, as measured by its Hilbert norm, while remaining robust to large nuisance complexity; and (ii) to adapt to the two weak-overlap structures.
To this end, we develop a two-stage KRR methodology with a model selection procedure that achieves these forms of adaptivity.
Here, weak overlap relaxes the traditional overlap assumption and allows the propensity score or source-target density ratio to exhibit singular behavior.
As a future direction, it would be valuable to investigate scenarios where the source and target domains have different sets of covariates or where the response functions are misspecified.
Additionally, extending the proposed procedure to other complex causal models, such as dynamic treatment regimes or policy learning, represents a promising avenue for further research.

\section*{Data Availability Statement}

The real-world data sets that support the findings of this paper are openly available at \url{https://wwwn.cdc.gov/nchs/nhanes/} (NHANES data) and the \texttt{pension} dataset in the R package \texttt{hdm} (401(k) data). All data sets and codes of our study can be found at the anonymous repository: \url{https://github.com/grindelwald-first/COKE-CATE-TRANSFER}.

\if1\anon
{
\section*{Acknowledgement}
%\noindent
Seok-Jin Kim and Kaizheng Wang’s research is supported by NSF grants DMS-2210907 and DMS-2515679, and a startup grant and a Data Science Institute seed grant SF-181 at Columbia University. The authors report there are no competing interests to declare.
} \fi

\begingroup
\setstretch{0.95}
\makeatletter
\let\oldthebibliography\thebibliography
\renewcommand{\thebibliography}[1]{%
\oldthebibliography{#1}%
\vspace{0.5\baselineskip}%
}
\makeatother
\bibliography{ref}
\endgroup

\clearpage
\newpage

\appendix
\setcounter{table}{0}
\setcounter{figure}{0}
\setcounter{equation}{0}

\renewcommand{\thefigure}{S\arabic{figure}}
\renewcommand{\thetable}{S\arabic{table}}
\renewcommand{\theequation}{S\arabic{equation}}

\clearpage
\newpage
\spacingset{1}
\begin{center}
{\Large \bf Supplementary Material for}\\[0.4em]
{\Large \bf
Transfer Learning of CATE with Kernel Ridge Regression}
\end{center}
\vspace{0.5cm}
\etocdepthtag.toc{mtappendix}
\etocsettagdepth{mtchapter}{none}
\etocsettagdepth{mtappendix}{subsection}
\tableofcontents

% \clearpage
% \newpage
\section{Preparations: Linear Model in RKHS and Notations}\label{section: groundwork}
\noindent
This section reformulates the RKHS responses as a linear model in Hilbert space and records notation used throughout the proofs.  
We also express the KRR estimator in a form analogous to linear regression, using the language of RKHS.  

\subsection{Linear Model via RKHS Mapping}
\noindent
Under the treatment regime in Section~\labelcref{section: problem setup}, we reformulate the problem in terms of Hilbertian elements.  
The RKHS structure allows us to construct a linear model through the feature map.  
Recall from Section~\labelcref{section: problem setup} that \( \HH \) denotes the Hilbert space induced by the kernel \( K(\cdot,\cdot) \).  
The two spaces \( \cF \) and \( \HH \) are isomorphic; hence, there exists a bijection that preserves the metric.  
For any element \( \theta \in \HH \), there exists \( f_\theta \in \cF \) such that \( \langle \theta, \phi(x) \rangle_{\HH} = f_\theta(x) \).  
Conversely, any function \( f \in \cF \) can be represented by a Hilbertian element \( \theta(f) \in \HH \) where \( \langle \theta(f), \phi(x) \rangle_{\HH} = f(x) \).  

We denote the Hilbert norm on \( \HH \) by \( \| \cdot \|_{\HH} \).  
Since the two spaces \( \cF \) and \( \HH \) are isomorphic, we have \( \| \theta \|_\HH = \| f_\theta \|_\cF \) and \( \| f \|_\cF = \| \theta(f) \|_\HH \).  
With slight abuse of notation, we write \( \langle x, y \rangle_{\HH} := x^\top y = y^\top x \) for any \( x, y \in \HH \) when the context is clear.  
Similarly, we write \( xy^\top := x \otimes y \) when the context is clear.  

We define the RKHS covariates as
\begin{align*}
x_i := \phi(z_i) \quad \forall i \in [n], \quad x_{0i} := \phi(z_{0i}) \quad \forall i \in [n_\cT].
\end{align*}
By setting $\theta(f_0^\star) = \theta_0^\star$ and $\theta(f_1^\star) = \theta_1^\star$ for some $\theta_0^\star, \theta_1^\star \in \HH$, we can reformulate our RKHS responses as the following linear model:
\begin{align*}
\EE[y_i \mid x_i,a_i=0] = x_i^\top \theta^\star_0, \quad \EE[y_i \mid x_i,a_i=1] = x_i^\top \theta^\star_1.
\end{align*}
Then, \(f_0^\star, f_1^\star \in \cF\) in Section~\labelcref{section: problem setup} correspond to \(\theta_0^\star, \theta_1^\star \in \HH\), respectively, and they have the same Hilbert norm, where $\|f_0^\star \|_\cF = \| \theta_0^\star \|_\HH$ and $\|f_1^\star\|_\cF = \| \theta_1^\star \|_\HH$.

Accordingly, the Hilbertian element of the CATE function is $\theta(h^\star) = \theta_1^\star - \theta_0^\star$, and we define 
\begin{align*}
\eta^\star := \theta_1^\star - \theta_0^\star.
\end{align*}
Our main goal is to estimate $\eta^\star$.
Recall that we defined $(z_i,a_i,y_i) \sim \cQ^\star_\cS$. 
By setting $x_i = \phi(z_i)$, we define the distribution of the source with the RKHS covariates as 
\begin{align*}
(x_i,a_i,y_i) \sim \cP^\star_\cS.
\end{align*}
In addition, we define the distributions of the RKHS covariates for the source and target as $x_i \sim \cP_\cS$ and \(x_{0i} \sim \cP_{\cT}\), respectively.
Then, our main object of interest, CATE, can be formulated as 
\begin{align*}
&\eta^\star = \theta^\star_1 - \theta_0^\star \\
&x^\top \eta^\star :=\EE_{(x,a,y)\sim \cP_\cS^\star} [y \mid x, a=1] - \EE_{(x,a,y)\sim \cP_\cS^\star} [y \mid x, a=0].
\end{align*}
Our goal is to estimate the CATE element, $\eta^\star$, by minimizing the target MSE.
We rewrite the expected second-order moments in terms of RKHS covariates. 
Then, we have $\Sigmatreated= \EE_{(x,a,y) \sim \cP^\star_S}[x \otimes x \one(a=1)]$, $\Sigmacontrol := \EE_{(x,a,y) \sim \cP^\star_S}[x \otimes x \one(a=0)]$, and $\bSigma_{\cS} = \EE_{(x,a,y) \sim \cP^\star_S}[x \otimes x]$.
For target RKHS covariates, we define \( \bSigma_\cT = \EE_{x \sim \cP_\cT} [x \otimes x] \).
Recall that we only observe covariates from the target distribution.
For each dataset \(\cD_j = \{(z_{ji}, a_{ji}, y_{ji})\}_{i=1}^{n_j}\), we define the RKHS covariates as \(x_{ji} := \phi(z_{ji})\) for all \(i \in [n_j]\).

\subsection{Closed Form of KRR Estimator}
\noindent
Recall the KRR setup described in Section~\labelcref{section: preliminaries kernel ridge regression}. 
We define RKHS covariates as $v_i := \phi(u_i)$ for all $i \in [N]$ and define the \emph{design operator} of \(\{v_1, \dots, v_N \}\) as $\Vb: \HH \to \RR^N$, which satisfies for all \(\theta \in \HH\):
\begin{align*}
\Vb \theta = (v_1^\top \theta, v_2^\top \theta, \dots, v_N^\top \theta)^\top.
\end{align*}
Similarly, we define the adjoint of \(\Vb\), denoted as \(\Vb^\top: \RR^N \to \HH \), as the operator such that for all \(\ab =(a_1, \dots, a_N) \in \RR^N\),
\begin{align*}
\Vb^\top \ab= \sum_{i=1}^N a_i v_i \in \HH.
\end{align*}

We define \(\rb =(r_1, \dots, r_N)^\top \).
It is known that the solution of the KRR program \eqref{equation: KRR program} in Section~\labelcref{section: preliminaries kernel ridge regression}, denoted by $\hat{f}$, satisfies
\(\hat{f}(u) = \phi(u)^\top \hat{\theta}\), where 
\begin{align*}
\hat{\theta} =(\Vb^\top \Vb + N\lambda \Ib)^{-1} \Vb^\top \rb.
\end{align*}
This has the same algebraic form as the ridge estimator in Euclidean linear regression, but expressed in terms of Hilbert space operators.

\subsection{Notations for Proofs}\label{subsection: notations}
\noindent
We collect the notation used in the proofs and briefly recap previously defined symbols.

\paragraph*{Notations for second moments and design operators}
{\renewcommand{\arraystretch}{1.15}
\begin{longtable}{@{}p{0.24\textwidth}p{0.72\textwidth}@{}}
\toprule
\textbf{Notation} & \textbf{Meaning} \\
\midrule
\endfirsthead
\toprule
\textbf{Notation} & \textbf{Meaning} \\
\midrule
\endhead
\(\Sigmatreated, \Sigmacontrol, \bSigma_{\cS}\) &
Expected second-order moments:
\(\Sigmatreated= \EE_{(x,a,y) \sim \cP^\star_S}[x \otimes x \one(a=1)]\),
\(\Sigmacontrol := \EE_{(x,a,y) \sim \cP^\star_S}[x \otimes x \one(a=0)]\), and
\(\bSigma_{\cS} = \EE_{(x,a,y) \sim \cP^\star_S}[x \otimes x]\). \\

\(\bSigma_\cT\) &
Target second-order moment:
\(\bSigma_\cT = \EE_{x \sim \cP_\cT} [x \otimes x]\). \\

\(\cD_{j,a}\) &
For \(j = 1,2\) and \(a=0,1\), define
\(\cD_{j,a} :=  \{(z_{ji}, a_{ji}, y_{ji}) \in \cD_j \mid a_{ji} = a\}\). \\

\shortstack[l]{\(\Sighat_{1}, \Sighat_{1,0}, \Sighat_{1,1}\)\\\(\Sighat_{2,0}, \Sighat_{2,1}\)} &
\parbox[t]{\linewidth}{Empirical second-order moments:\par
\(
\begin{aligned}[t]
\Sighat_{1} &:= \frac{1}{n_{1}}\sum_{i=1}^{n_1} x_{1i}x_{1i}^\top, \\
\Sighat_{1,0} &:= \frac{1}{n_{1}}\sum_{i=1}^{n_1} x_{1i}x_{1i}^\top \bm{1}(a_{1i}=0), \\
\Sighat_{1,1} &:= \frac{1}{n_{1}}\sum_{i=1}^{n_1} x_{1i}x_{1i}^\top \bm{1}(a_{1i}=1), \\
\Sighat_{2,0} &:= \frac{1}{n_{2}}\sum_{i=1}^{n_2} x_{2i}x_{2i}^\top \bm{1}(a_{2i}=0), \\
\Sighat_{2,1} &:= \frac{1}{n_{2}}\sum_{i=1}^{n_2} x_{2i}x_{2i}^\top \bm{1}(a_{2i}=1).
\end{aligned}
\)
\par
} \\

\(\Xb_{j,a}, \Xb_{j}\) &
Design operators for RKHS covariates \(\{x_{ji}\bm{1}(a_{ji}=a)\}_{i=1}^{n_j}\) (\(a \in \{0,1\}\), \(j \in \{1,2\}\)) and \(\{x_{ji}\}_{i=1}^{n_j}\).
The definition matches Appendix~\labelcref{section: groundwork}. \\

\(\Xb_\cT\) &
Design operator of \(\{x_{0i}\}_{i=1}^{n_\cT}\) (target covariates). \\

\(M\) &
\(M := \max( \|\theta_0^\star \|_\HH, \|\theta_1^\star \|_\HH) =\max( \|f_0^\star \|_\cF, \|f_1^\star \|_\cF)\). \\

\(\Sbar_\lambda, \Shat_\lambda\) &
For \(\lambda>0\),
\(\Sbar_\lambda := (\bSigma_{\cS} + \lambda \Ib)^{-\frac{1}{2}} \bSigma_{\cT} (\bSigma_{\cS} + \lambda \Ib)^{-\frac{1}{2}}\) and
\(\Shat_\lambda := (\bSigma_{\cS} + \lambda \Ib)^{-\frac{1}{2}} \Sighat_{\cT} (\bSigma_{\cS} + \lambda \Ib)^{-\frac{1}{2}}\). \\
\bottomrule
\end{longtable}
}

\paragraph*{Other key notations}
{\renewcommand{\arraystretch}{1.15}
\begin{longtable}{@{}p{0.24\textwidth}p{0.72\textwidth}@{}}
\toprule
\textbf{Notation} & \textbf{Meaning} \\
\midrule
\endfirsthead
\toprule
\textbf{Notation} & \textbf{Meaning} \\
\midrule
\endhead
\(\langle x,y\rangle_{\HH}\), \(xy^\top\) &
With slight abuse of notation, \(\langle x,y\rangle_{\HH} := x^\top y = y^\top x\) and \(xy^\top := x \otimes y\) for any \(x,y \in \HH\) when the context is clear. \\

\(n_{j,a}\), \(n_j\) &
Sample sizes: \(n_{j,a}\) for each \(\cD_{j,a}\) (\(j \in \{1,2\}\), \(a \in \{0,1\}\)) and \(n_j\) for \(\cD_j\).
Note that \(n_1 = n_{1,0} + n_{1,1}\) and \(n_{2} = n_{2,0} + n_{2,1}\). \\

\(\hat{\theta}_1, \hat{\theta}_0, \hat{\eta}_{\bm{\lambda}}\) &
Hilbertian elements corresponding to \(\hat{f}_1, \hat{f}_0\) in Algorithm~\labelcref{algorithm: RA learner}, and
\(\hat{\eta}_{\bm{\lambda}}\) corresponds to \(\hat{h}_{\bm{\lambda}}\).\\

\(\tilde{\theta}_1, \tilde{\theta}_0, \tilde{\eta}\) &
Hilbertian elements corresponding to \(\tilde{f}_1, \tilde{f}_0\) in Algorithm~\labelcref{algorithm: model selection}, and
\(\tilde{\eta} := \tilde{\theta}_1 - \tilde{\theta}_0\), which corresponds to \(\tilde{h}\).\\

\(\cH\) &
Set of candidate estimators defined in our algorithm. \\

\(\yb_{j,a}, \bm{\varepsilon}_{j,a}, \yb_{j}, \bm{\varepsilon}_{j}\) &
Response and noise vectors associated with \(\{y_{ji}\bm{1}(a_{ji}=a)\}_{i=1}^{n_j}\), \(\{\varepsilon_{ji}\bm{1}(a_{ji}=a)\}_{i=1}^{n_j}\), \(\{y_{ji}\}_{i=1}^{n_j}\), and \(\{\varepsilon_{ji}\}_{i=1}^{n_j}\); i.e., vectorized versions of \(y_i\) and \(\varepsilon_i\).\\

\(\alpha\) &
For \(\ell\)-polynomial decay eigenvalues (Assumption~\labelcref{Assumption; eigenvalue decay}), define \(\alpha = \frac{2\ell}{1+2\ell}\). \\

\(\{\varepsilon_{ji}\}_{i=1}^{n_j}\) &
Noise terms in dataset \(\cD_j\) for \(j=1,2\). \\
\bottomrule
\end{longtable}
}

\subsection{Proof Workflow}
\noindent
In Appendix~\labelcref{section: good events and second moments}, we first define the good event $\event$, which ensures sufficient concentration of the empirical second moments.  
Roughly, under $\event$, each empirical second moment $\Sighat_{j,a}$ for $j = 1, 2$ and $a = 0, 1$ is well concentrated around $\bSigma_{\cS,a}$.  
A rigorous definition appears in Appendix~\labelcref{section: good events and second moments}.  
In the next section, we show that $\PP[\event] \geq 1 - n^{-11}$, and we conduct the remaining analysis under $\event$.  

Our goal is to prove the main result, Theorem~\labelcref{theorem; main theorem}. The proof proceeds as follows:  
\begin{enumerate}  
\item Under the event $\event$, we first bound the MSE of the estimator $\hat{h}_{\bm{\lambda}}$ for a fixed $\bm{\lambda}$.  
This result establishes Theorem~\labelcref{theorem; MSE bound RA learner} (presented in Appendix~\labelcref{section: proof RA learner}).  

\item Using Theorem~\labelcref{theorem; MSE bound RA learner}, we derive Corollary~\labelcref{corollary; optimal MSE bound} (presented in Appendix~\labelcref{section: proof RA learner}).  

\item Under the event $\event$, we prove the oracle inequality for the in-sample MSE in Lemma~\labelcref{lemma; in-sample MSE oracle} (presented in Appendix~\labelcref{section: proofs oracle inequalities}).  

\item Using Lemma~\labelcref{lemma; in-sample MSE oracle} together with Lemma~\ref{lemma; comparability in-sample and population MSE}, we establish the oracle inequality for the MSE in Proposition~\labelcref{proposition; oracle inequality MSE} (presented in Appendix~\labelcref{section: proofs oracle inequalities}).  

\item Finally, by combining Proposition~\labelcref{proposition; oracle inequality MSE}, Theorem~\labelcref{theorem; MSE bound RA learner}, and Corollary~\labelcref{corollary; optimal MSE bound}, we prove the main result, Theorem~\labelcref{theorem; main theorem} (presented in Appendix~\labelcref{section: proof main theorem}).  
\end{enumerate}

\section{Second-moment Concentrations and Good Event $\event$}\label{section: good events and second moments}
\noindent
In this section, we analyze the concentration of second moments and define the good event $\event$. We also present two results used throughout the proofs: Lemma~\labelcref{lemma; moment ratio in E1} and Corollary~\labelcref{corollary; application of Lemma second moment ratio}.

\subsection{Good Event: Sufficient Second Moment Concentrations}
\noindent
We present concentration bounds for the empirical second-order moment operators and define the good event \(\event\). For the definition of empirical second moments, see Appendix~\labelcref{subsection: notations}. We invoke Lemma~\labelcref{lemma; trace class concentration bounded} repeatedly to obtain an absolute constant \(c_0\) such that the following concentration inequalities hold.

\noindent
\underline{Concentration 1:} \quad 
\noindent
With probability at least $1 - \frac{n^{-11}}{6}$, for any $\mu \geq \frac{c_0 \xi \log n}{n}$, the following inequality holds:
\begin{align*}
\frac{1}{2}(\Sighat_{1,1} + \mu \mathbf{I}) \preceq \Sigmatreated + \mu \mathbf{I} \preceq 2(\Sighat_{1,1} + \mu \mathbf{I}).
\end{align*}
Similarly, with the same probability and for the same $\mu$, we have:
\begin{align*}
\frac{1}{2}(\Sighat_{1,0} + \mu \mathbf{I}) \preceq \Sigmacontrol + \mu \mathbf{I} \preceq 2(\Sighat_{1,0} + \mu \mathbf{I}),
\end{align*}
and
\begin{align*}
\frac{1}{2}(\Sighat_{1} + \mu \mathbf{I}) \preceq \bSigma_{\cS} + \mu \mathbf{I} \preceq 2(\Sighat_{1} + \mu \mathbf{I}).
\end{align*}

\noindent\underline{Concentration 2:} \quad 
\noindent
For any $\mu \geq \frac{c_0 \xi \log n}{n}$, with probability at least $1 - \frac{2}{6}n^{-11}$, the following inequalities hold:
\begin{align*}
\frac{1}{2}(\Sighat_{2,1} + \mu \mathbf{I}) \preceq \Sigmatreated + \mu \mathbf{I} \preceq 2(\Sighat_{2,1} + \mu \mathbf{I}),
\end{align*}
\begin{align*}
\frac{1}{2}(\Sighat_{2,0} + \mu \mathbf{I}) \preceq \Sigmacontrol + \mu \mathbf{I} \preceq 2(\Sighat_{2,0} + \mu \mathbf{I}).
\end{align*}

Next, we present similar second-moment concentration bounds for the target data.

\noindent\underline{Concentration 3:} \quad 
\noindent
For any $\mu' \geq c_0\frac{\xi (\log n_\cT + \log n)}{n_\cT}$, with probability at least $1 - \frac{n^{-11}}{6}$, the following inequality holds:
\begin{align*}
\frac{1}{2}(\Sighat_{\cT} + \mu' \mathbf{I}) \preceq \bSigma_{\cT} + \mu' \mathbf{I} \preceq 2(\Sighat_{\cT} + \mu' \mathbf{I}).
\end{align*}

\begin{definition}[Good event]
We define the good event $\event$ as the event in which Concentrations 1, 2, and 3 hold for all $\mu \geq c_0\frac{\xi\log n}{n}$ and $\mu' \geq c_0\frac{\xi (\log n_\cT + \log n)}{n_\cT}$. Then,
\begin{align*}
\PP[\event] \geq 1 - n^{-11}
\end{align*}
by the previous observations.
\end{definition}

\subsection{Second Moment Ratio Bounds under the Event \(\event\)}
\noindent
Next, we record useful properties that hold under the event $\event$. Our regularizers for nuisance estimation are $\lambda_{0,0}, \lambda_{0,1} = \frac{\xi \log n}{n}$, and our analysis is restricted to this range.

\begin{lemma}[Second moment concentrations in \(\event\)]\label{lemma; moment ratio in E1}
Under the event \(\event\), for any \(\lambda \geq \frac{\xi \log n}{n}\), \(j \in \{1,2\}\), and \(a \in \{0,1\}\), the following inequalities hold for some absolute constant \(c_1>0\):
\begin{align*}
\frac{1}{c_1}(\Sighat_{j,a} + \lambda \Ib) \preceq \bSigma_{\cS,a} + \lambda \Ib \preceq c_1 (\Sighat_{j,a} + \lambda \Ib),
\end{align*}
and
\begin{align*}
\frac{1}{c_1}(\Sighat_{1} + \lambda \Ib) \preceq \bSigma_{\cS} + \lambda \Ib \preceq c_1 (\Sighat_{1} + \lambda \Ib).
\end{align*}
In addition, the following properties hold:
\begin{align*}
\frac{1}{c_2 R} (\Sighat_{1,1} + \lambda \Ib) \preceq \Sighat_{1,0} + \lambda \Ib \preceq c_2 R (\Sighat_{1,1} + \lambda \Ib),
\end{align*}
and
\begin{align*}
\frac{1}{c_2 R} (\Sighat_{1,0} + \lambda \Ib) \preceq \Sighat_{1,1} + \lambda \Ib \preceq c_2 R (\Sighat_{1,0} + \lambda \Ib)
\end{align*}
for some absolute constant $c_2>0$.
\end{lemma}

\begin{proof}
Let \(\mu = \frac{c_0 \xi \log n}{n}\). 
First, we prove the first two inequalities. These inequalities hold directly when \(\lambda \geq \mu\). 
For \(\lambda < \mu\), note that \(\frac{\mu}{\lambda} \leq c_0\). For \(j \in \{1,2\}\), under the event \(\event\),
\begin{align*}
\bSigma_{\cS} + \lambda \Ib &\succeq (\bSigma_{\cS} + \mu \Ib) \frac{\lambda}{\mu} \\
&\succeq (\Sighat_{1} + \mu \Ib) \frac{\lambda}{2\mu} \\
&\succeq \frac{1}{2c_0} (\Sighat_{1} + \mu \Ib) \\
&\succeq \frac{1}{2c_0} (\Sighat_{1} + \lambda \Ib).
\end{align*}
Similarly,
\begin{align*}
\Sighat_{1} + \lambda \Ib &\succeq \frac{\lambda}{\mu} (\Sighat_{1} + \mu \Ib) \\
&\succeq \frac{1}{c_0} (\Sighat_{1} + \mu \Ib) \\
&\succeq \frac{1}{2c_0} (\bSigma_{\cS} + \mu \Ib) \\
&\succeq \frac{1}{2c_0} (\bSigma_{\cS} + \lambda \Ib).
\end{align*}
Since \(c_0\) is an absolute constant, we have
\begin{align*}
\frac{1}{c_1}(\Sighat_{1} + \lambda \Ib) \preceq \bSigma_{\cS} + \lambda \Ib \preceq c_1 (\Sighat_{1} + \lambda \Ib)
\end{align*}
for some absolute constant \(c_1 > 0\).

Similarly, for all \(j \in \{1,2\}\) and \(a \in \{0,1\}\),
\begin{align*}
\frac{1}{c_1}(\Sighat_{j,a} + \lambda \Ib) \preceq \bSigma_{\cS, a} + \lambda \Ib \preceq c_1 (\Sighat_{j,a} + \lambda \Ib).
\end{align*}

Next, we prove the third and fourth inequalities. Using the above observations and Assumption~\labelcref{assumption; weak treatment overlap}, we have
\begin{align*}
\Sighat_{1,0} + \lambda \Ib &\preceq c_1 (\Sigmacontrol + \lambda \Ib) \quad \text{(by the first inequality of this lemma)}\\
&\preceq c_1 (R\Sigmatreated + R\frac{\xi}{n} \Ib + \lambda \Ib) \quad \text{(by Assumption~\ref{assumption; weak treatment overlap})}\\
&\preceq 2c_1 R (\Sigmatreated + \lambda \Ib)  \\
&\preceq 2c_1^2 R (\Sighat_{1,1} + \lambda \Ib) \quad (\text{by the first inequality of this lemma}),
\end{align*}
and
\begin{align*}
\Sighat_{1,1} + \lambda \Ib &\preceq c_1 (\Sigmatreated + \lambda \Ib) \quad (\text{by the first inequality of this lemma}) \\
&\preceq c_1 (R\Sigmacontrol + R\frac{\xi}{n} \Ib + \lambda \Ib) \quad \text{(by Assumption~\ref{assumption; weak treatment overlap})} \\
&\preceq 2c_1 R (\Sigmacontrol + \lambda \Ib) \\
&\preceq 2c_1^2 R (\Sighat_{1,0} + \lambda \Ib) \quad (\text{by the first inequality of this lemma}).
\end{align*}
The same reasoning yields the fourth inequality, completing the proof.
\end{proof}

We next state a key lemma for matrix calculations under the good event \(\event\).

\begin{corollary}[Second moment ratio upper bounds]\label{corollary; application of Lemma second moment ratio}
Under the event $\event$, the following inequalities hold for any $\lambda \geq \frac{\xi \log n}{n}$ and some absolute constant \(c>0\):
\begin{align*}
(\Sighat_{1,1} + \lambda \Ib)^{-\frac{1}{2}} \Sighat_{1,0} (\Sighat_{1,1} + \lambda \Ib)^{-\frac{1}{2}} &\preceq c R \Ib, \\
(\Sighat_{1,0} + \lambda \Ib)^{-\frac{1}{2}} \Sighat_{1,1} (\Sighat_{1,0} + \lambda \Ib)^{-\frac{1}{2}} &\preceq c R \Ib, \\
\Sighat_{1,1}^{\frac{1}{2}} (\Sighat_{1,0} + \lambda \Ib)^{-1} \Sighat_{1,1}^{\frac{1}{2}} &\preceq c R \Ib, \\
\Sighat_{1,0}^{\frac{1}{2}} (\Sighat_{1,1} + \lambda \Ib)^{-1} \Sighat_{1,0}^{\frac{1}{2}} &\preceq c R \Ib.
\end{align*}
Additionally, we have:
\begin{align*}
(\widehat{\bSigma}_{1} + \lambda \Ib)^{-\frac{1}{2}} \bSigma_\cT (\widehat{\bSigma}_{1} + \lambda \Ib)^{-\frac{1}{2}} &\preceq c B \Ib, \\
\bSigma_\cT^{\frac{1}{2}} (\widehat{\bSigma}_{1} + \lambda \Ib)^{-1} \bSigma_\cT^{\frac{1}{2}} &\preceq c B \Ib.
\end{align*}
\end{corollary}

\begin{proof}
Consider the first inequality:
\begin{align*}
(\Sighat_{1,1} + \lambda \Ib)^{-\frac{1}{2}} \Sighat_{1,0} (\Sighat_{1,1} + \lambda \Ib)^{-\frac{1}{2}} 
&\preceq (\Sighat_{1,1} + \lambda \Ib)^{-\frac{1}{2}} (\Sighat_{1,0} + \lambda \Ib) (\Sighat_{1,1} + \lambda \Ib)^{-\frac{1}{2}} \\
&\preceq c_2 R (\Sighat_{1,1} + \lambda \Ib)^{-\frac{1}{2}} (\Sighat_{1,1} + \lambda \Ib) (\Sighat_{1,1} + \lambda \Ib)^{-\frac{1}{2}} \\
&\preceq c_2 R \Ib,
\end{align*}
where we used Lemma~\labelcref{lemma; moment ratio in E1} for the second line.
The second inequality follows similarly.

For the third inequality,
\begin{align*}
\left\|\Sighat_{1,1}^{\frac{1}{2}}(\Sighat_{1,0}+\lambda \Ib)^{-1}\Sighat_{1,1}^{\frac{1}{2}}\right\|_{\op}
&=
\left\|(\Sighat_{1,0}+\lambda \Ib)^{-\frac{1}{2}}\Sighat_{1,1}(\Sighat_{1,0}+\lambda \Ib)^{-\frac{1}{2}}\right\|_{\op} \\
&\leq c_2 R,
\end{align*}
where we used the second inequality above. Since the operator on the left is positive semidefinite, this implies the desired Loewner bound.
The fourth inequality follows similarly.

For the fifth inequality,
\begin{align*}
(\widehat{\bSigma}_{1} + \lambda \Ib)^{-\frac{1}{2}} \bSigma_\cT (\widehat{\bSigma}_{1} + \lambda \Ib)^{-\frac{1}{2}}  
&\stackrel{\text{(i)}}{\preceq} (\widehat{\bSigma}_{1} + \lambda \Ib)^{-\frac{1}{2}} B (\bSigma_\cS + \lambda \Ib) (\widehat{\bSigma}_{1} + \lambda \Ib)^{-\frac{1}{2}} \\
&\stackrel{\text{(ii)}}{\preceq} c_1 (\widehat{\bSigma}_{1} + \lambda \Ib)^{-\frac{1}{2}} B (\Sighat_{1} + \lambda \Ib) (\widehat{\bSigma}_{1} + \lambda \Ib)^{-\frac{1}{2}} \\
&= c_1 B \Ib,
\end{align*}
where in step (i) we used Assumption~\labelcref{assumption; overlap source target}, and in step (ii) we applied Lemma~\labelcref{lemma; moment ratio in E1}.
The sixth inequality follows similarly.

\end{proof}

We conclude this section with a useful lemma relating \(\Sigmacontrol, \Sigmatreated,\) and \(\bSigma\) under Assumption~\labelcref{assumption; weak treatment overlap}.
\begin{lemma}\label{lemma; second moment relation under weak overlap}
Under Assumption~\labelcref{assumption; weak treatment overlap}, for any \(\lambda \geq \frac{\xi \log n}{n}\),
\begin{align*}
\Sigmatreated +\lambda \Ib \preceq c R (\Sigmacontrol +\lambda \Ib) \\
\Sigmacontrol +\lambda \Ib \preceq c R (\Sigmatreated +\lambda \Ib).
\end{align*}
and
\begin{align*}
\bSigma_\cS +\lambda \Ib \preceq c R (\Sigmacontrol +\lambda \Ib) \\
\bSigma_\cS +\lambda \Ib \preceq c R (\Sigmatreated +\lambda \Ib).
\end{align*}
These inequalities hold for some absolute constant \(c>0\).
\end{lemma}

\begin{proof}
For the first inequality, observe that
\begin{align*}
\Sigmatreated +\lambda \Ib &\preceq R(\Sigmacontrol+ \frac{\xi}{n} \Ib) + \lambda \Ib \quad \text{(By Assumption~\labelcref{assumption; weak treatment overlap})} \\
&\preceq R(\Sigmacontrol+ \lambda  \Ib) + \lambda \Ib \\
&\preceq 2R(\Sigmacontrol+ \lambda  \Ib).
\end{align*} 
The second inequality follows similarly.

Next, we prove the third inequality. Using the first inequality, we obtain
\begin{align*}
\bSigma_\cS +\lambda \Ib &= \Sigmacontrol + \Sigmatreated + \lambda \Ib \\
&\preceq \Sigmacontrol + 2R(\Sigmacontrol +\lambda \Ib)  \\
&\preceq 3R(\Sigmacontrol +\lambda \Ib).
\end{align*}

The fourth inequality follows similarly.
\end{proof}

\section{Proofs for Theorem~\labelcref{theorem; MSE bound RA learner} and Corollary~\labelcref{corollary; optimal MSE bound}} \label{section: proof RA learner}
\noindent
This section proves Theorem~\labelcref{theorem; MSE bound RA learner} and Corollary~\labelcref{corollary; optimal MSE bound}.
Throughout, we work within the good event \(\event\), which is defined in Appendix~\labelcref{section: good events and second moments}.
As proved in Appendix~\labelcref{section: good events and second moments}, \(\PP[\event] \geq 1 - n^{-11}\).

\subsection{Hilbertian Formulation of RA Learner}\label{subsection: Hilbertian formulation RA learner}
\noindent
We begin by reformulating the RA learner in the language of Hilbertian elements, as established in Appendix~\labelcref{section: groundwork}. 
For the estimator \(\hat{h}_{\bm{\lambda}}\) obtained from the RA learner with regularizers \(\bm{\lambda} = (\lambda_{0,0}, \lambda_{0,1}, \lambda_{1})\), let \(\hat{\eta}_{\bm{\lambda}}\) be its corresponding Hilbert space element.

The RA learner performs nuisance estimation in the first stage.
Let \(\hat{\theta}_0\) and \(\hat{\theta}_1\) be the Hilbertian elements corresponding to \(\hat{f}_0\) and \(\hat{f}_1\), defined as 
\begin{align*}
&\hat{\theta}_1 := (\Xb_{1,1}^\top \Xb_{1,1} + n_{1} \lambda_{0,1}\Ib)^{-1} \Xb_{1,1}^\top \yb_{1,1},\\
&\hat{\theta}_0 := (\Xb_{1,0}^\top \Xb_{1,0} + n_{1} \lambda_{0,0}\Ib)^{-1} \Xb_{1,0}^\top \yb_{1,0}.
\end{align*}
Recall that the pseudo-outcome in Algorithm~\labelcref{algorithm: RA learner} is defined as 
\[
m_{1i} := (y_{1i} - x_{1i}^\top \hat{\theta}_0)\indicator(a_{1i}=1) + (x_{1i}^\top \hat{\theta}_1 - y_{1i}) \indicator(a_{1i}=0).
\]
With a slight abuse of notation, set the target MSE of any CATE estimator \(\hat{\eta}\) as 
\[
\Ecal_{\cT}(\hat{\eta}) := \EE_{x \sim \cP_\cT} \bigl|x^\top(\hat{\eta}-\eta^\star)\bigr|^2 
= \|\hat{\eta} - \eta^\star \|_{\bSigma_\cT}^2.
\]

In RKHS form, Theorem~\labelcref{theorem; MSE bound RA learner} becomes:
\[
\cE_\cT(\hat{\eta}_{\bm{\lambda}}) 
\lesssim R \|\Sbar_{\lambda_{1}} \|_{\op}
\bigl(\lambda_{0,1}\norm{\theta^\star_1}_{\HH}^2 + \lambda_{0,0}\norm{\theta^\star_0}_{\HH}^2\bigr)
+ \lambda_{1} \|\Sbar_{\lambda_{1}} \|_{\op} \|\eta^{\star}\|_{\HH}^2
+\sig^2  \frac{ R\operatorname{Tr}\bigl(\Sbar_{\lambda_{1}}\bigr)}{n}\log n,
\]
which holds for all \(\lambda_{0,0}, \lambda_{0,1}, \lambda_{1} \geq \xi\log n/n\).
The equivalent statement of Corollary~\labelcref{corollary; optimal MSE bound} is:
\[
\min_{\bm{\lambda} \in \bm{\Lambda}} \cE_\cT(\hat{\eta}_{\bm{\lambda}}) 
\lesssim 
\Bigl(\frac{BR}{n}\Bigr)^{\alpha}
\norm{\eta^\star}_{\HH}^{2(1-\alpha)}
(\log n)^\alpha
+ \frac{\xi BR}{n}\max\bigl(\|\theta_{0}^\star\|_{\HH}, \|\theta_{1}^\star\|_{\HH}\bigr)^2 \log n.
\]
We will prove both results in this section and analyze the MSE of \(\hat{\eta}_{\bm{\lambda}}\).

\subsection{Decomposition of MSE}\label{subsection: MSE decomposition}
\noindent
By the definition of the RA learner, we can write the estimator in closed form as
\begin{align*}
\hat{\eta}_{\bm{\lambda}} 
&= (\Sighat_{1} + \lambda_{1} \mathbf{I})^{-1} \frac{1}{n_{1}} \Bigl(\sum_{\cD_{1,1}}x_{1i} \bigl(y_{1i} - x_{1i}^\top \hat{\theta}_0\bigr) +\sum_{\cD_{1,0}} x_{1i}\bigl( x_{1i}^\top \hat{\theta}_1 - y_{1i}\bigr)\Bigr)\\
&=(\Sighat_{1} + \lambda_{1} \mathbf{I})^{-1}  \frac{1}{n_{1}}\Bigl(\sum_{\cD_{1,1}} x_{1i}\bigl(y_{1i} - x_{1i}^\top {\theta}^\star_0 \bigr) + x_{1i}x_{1i}^\top\bigl({\theta}^\star_0-\hat{\theta}_0\bigr)\\
&\quad +\sum_{\cD_{1,0}} x_{1i}\bigl( x_{1i}^\top {\theta}^\star_1 -y_{1i}\bigr) - x_{1i}x_{1i}^\top\bigl({\theta}^\star_1-\hat{\theta}_1\bigr) \Bigr) \\
&= (\Sighat_{1} + \lambda_{1} \mathbf{I})^{-1}  \frac{1}{n_{1}}\Bigl(\sum_{\cD_{1,1}} x_{1i}\bigl(x_{1i}^\top \theta_1^\star + \varepsilon_{1i} - x_{1i}^\top {\theta}^\star_0\bigr) +x_{1i}x_{1i}^\top\bigl({\theta}^\star_0-\hat{\theta}_0\bigr)\\
&\quad+\sum_{\cD_{1,0}} x_{1i}\bigl( x_{1i}^\top {\theta}^\star_1 -x_{1i}^\top \theta_0^\star - \varepsilon_{1i}\bigr)
+x_{1i}x_{1i}^\top\bigl(\hat{\theta}_1-{\theta}^\star_1\bigr)\Bigr) \\
&= (\Sighat_{1} + \lambda_{1} \mathbf{I})^{-1}  \frac{1}{n_{1}}\Bigl(\sum_{i=1}^{n_{1}}x_{1i}x_{1i}^\top\bigl(\theta_1^\star-\theta_0^\star\bigr) 
+
\sum_{\cD_{1,1}} \bigl(x_{1i}\varepsilon_{1i} +x_{1i}x_{1i}^\top\bigl({\theta}^\star_0-\hat{\theta}_0\bigr)\bigr) \\
&\quad +\sum_{\cD_{1,0}}  \bigl(- x_{1i}\varepsilon_{1i}+x_{1i}x_{1i}^\top\bigl(\hat{\theta}_1-{\theta}^\star_1\bigr)\bigr)\Bigr)\\
&= (\Sighat_{1} + \lambda_{1} \mathbf{I})^{-1} \frac{1}{n_{1}}\Bigl(n_{1} \Sighat_{1} \eta^\star + 
\sum_{\cD_{1,1}} \bigl(x_{1i}\varepsilon_{1i} +x_{1i}x_{1i}^\top\bigl({\theta}^\star_0-\hat{\theta}_0\bigr)\bigr) \\
&\quad +\sum_{\cD_{1,0}} \bigl(- x_{1i}\varepsilon_{1i}+x_{1i}x_{1i}^\top\bigl(\hat{\theta}_1-{\theta}^\star_1\bigr)\bigr)\Bigr).
\end{align*}
Then, the difference \(\hat{\eta}_{\bm{\lambda}}-\eta^\star\) is
\begin{align*}
\hat{\eta}_{\bm{\lambda}} -\eta^\star 
&= (\Sighat_{1} + \lambda_{1} \mathbf{I})^{-1} \frac{1}{n_{1}}\Bigl(\sum_{\cD_{1,1}} \bigl(x_{1i}\varepsilon_{1i} +x_{1i}x_{1i}^\top\bigl({\theta}^\star_0-\hat{\theta}_0\bigr)\bigr)\\
&\quad +\sum_{\cD_{1,0}} \bigl(-x_{1i}\varepsilon_{1i}+x_{1i}x_{1i}^\top\bigl(\hat{\theta}_1-{\theta}^\star_1\bigr)\bigr) - n_{1}\lambda_{1} \eta^\star \Bigr) \\
&= (\Sighat_{1} + \lambda_{1} \mathbf{I})^{-1} \Bigl( \frac{1}{n_{1}} \bigl(\sum_{\cD_{1,1}}x_{1i} \varepsilon_{1i}+\sum_{\cD_{1,0}} -x_{1i}\varepsilon_{1i}\bigr) -\lambda_{1} \eta^\star\Bigr)\\
&\quad + (\Sighat_{1} + \lambda_{1} \mathbf{I})^{-1} \Bigl( \frac{1}{n_{1}}\bigl( n_{1}\Sighat_{1,1}\bigl({\theta}^\star_0-\hat{\theta}_0\bigr)
+ n_{1}\Sighat_{1,0} \bigl(\hat{\theta}_1-{\theta}^\star_1\bigr)\bigr)\Bigr) \\
&= (\Sighat_{1} + \lambda_{1} \mathbf{I})^{-1} \Bigl( \frac{1}{n_{1}} \bigl(\sum_{\cD_{1,1}}x_{1i} \varepsilon_{1i}+\sum_{\cD_{1,0}} -x_{1i}\varepsilon_{1i}\bigr) -\lambda_{1} \eta^\star\Bigr)\\
&\quad + (\Sighat_{1} + \lambda_{1} \mathbf{I})^{-1} \Bigl( \Sighat_{1,1}\bigl({\theta}^\star_0-\hat{\theta}_0\bigr)
+ \Sighat_{1,0} \bigl(\hat{\theta}_1-{\theta}^\star_1\bigr)\Bigr).
\end{align*}
To control the MSE, define each term as
\begin{align*}
\Vcr &= \norm{(\Sighat_{1} + \lambda_{1} \mathbf{I})^{-1}\Bigl( \frac{1}{n_{1}} \bigl(\sum_{\cD_{1,1}}x_{1i} \varepsilon_{1i}+\sum_{\cD_{1,0}} -x_{1i}\varepsilon_{1i}\bigr)\Bigr)}_{\bSigma_{\cT}},\\
\Bcr &= \norm{(\Sighat_{1} + \lambda_{1} \mathbf{I})^{-1}\lambda_{1} \eta^\star}_{\bSigma_{\cT}},\\
\Pcr_0 &= \norm{(\Sighat_{1} + \lambda_{1} \mathbf{I})^{-1} \Sighat_{1,1}\bigl(\hat{\theta}_0-{\theta}^\star_0\bigr)}_{\bSigma_{\cT}},\\
\Pcr_1 &= \norm{(\Sighat_{1} + \lambda_{1} \mathbf{I})^{-1} \Sighat_{1,0} \bigl(\hat{\theta}_1-{\theta}^\star_1\bigr)}_{\bSigma_{\cT}}.
\end{align*}
Finally, the MSE can be decomposed as
\[
\norm{\hat{\eta}_{\bm{\lambda}} -\eta^\star}_{\bSigma_{\cT}} \leq \mathscr{V}+\mathscr{B}+ \mathscr{P}_0 +  \mathscr{P}_1.
\]
The following subsections bound each term.

\subsection{Bounding \(\mathscr{V}\) and \(\mathscr{B}\)}
\noindent
We bound \(\Vcr^2 + \Bcr^2\) using techniques from \citet{wang2026pseudo,ma2023optimally} together with \(\lambda_{1} \asymp n_{1}^{-\alpha}\).
We first present an upper bound for \(\Vcr^2 + \Bcr^2\).
We introduce \(\delta_1 >0\) as an auxiliary probability parameter under \(\event\), which can be chosen arbitrarily. 
In the end, we set \(\delta_1 = \frac{n^{-11}}{|\cH|}\).

\begin{lemma}\label{lemma; RA learner bound 1}
Fix any \(\delta_1>0\).
Under the good event \(\event\), with probability at least \(1- \delta_1/3\), we have
\[
\Vcr^2 + \Bcr^2 \lesssim  \lambda_{1}\|\Sbar_{\lambda_{1}}  \|_{\op}\bigl\|\eta^{\star}\bigr\|_{\HH}^2+\sig^2 \frac{ \operatorname{Tr}\bigl(\Sbar_{\lambda_{1}}\bigr) \log (1 / \delta_1)}{n_{1}}.
\]
Here, \(\lesssim\) hides absolute constants.
\end{lemma}

\begin{proof}
Under Assumption~\ref{assumption; subGaussian noise}, the noise variables \(\varepsilon_{1i}(-1)^{a_{1i}+1}\), \(i =1,2,\dots,n\), are mean-zero and sub-Gaussian given \(\{(x_i,a_i) \}_{i=1}^n\). 
Thus, \(\Vcr^2 + \Bcr^2\) can be viewed as the MSE of a general KRR problem, which allows us to apply prior results.
We sketch the argument below. 

For the variance term, applying the Hanson-Wright inequality (Lemma~\labelcref{lemma; quadratic martingale}) yields
\begin{align*}
\Vcr^2 
&\lesssim \sigma^2 \frac{1}{n_{1}^2}\Tr\Bigl( \Xb_1 \bigl(\Sighat_{1} + \lambda_{1} \Ib\bigr)^{-1} \bSigma_\cT \bigl(\Sighat_{1} + \lambda_{1} \Ib\bigr)^{-1} \Xb_1^\top \Bigr) \log\Bigl(\frac{3}{\delta_1}\Bigr)\\
&\lesssim \sigma^2 \frac{1}{n_{1}}\Tr\Bigl( \bigl(\Sighat_{1} + \lambda_{1} \Ib\bigr)^{-1} \bSigma_\cT \bigl(\Sighat_{1} + \lambda_{1} \Ib\bigr)^{-1}\Sighat_{1}\Bigr)\log\Bigl(\frac{1}{\delta_1}\Bigr)\\
&\lesssim \sigma^2 \frac{1}{n_{1}}\Tr\Bigl( \bigl(\Sighat_{1} + \lambda_{1} \Ib\bigr)^{-\frac{1}{2}} \bSigma_\cT \bigl(\Sighat_{1} + \lambda_{1} \Ib\bigr)^{-\frac{1}{2}} 
\bigl(\Sighat_{1} + \lambda_{1} \Ib\bigr)^{-\frac{1}{2}} \Sighat_{1} \bigl(\Sighat_{1} + \lambda_{1} \Ib\bigr)^{-\frac{1}{2}}\Bigr) \log\Bigl(\frac{1}{\delta_1}\Bigr)\\
&\stackrel{\text{(i)}}{\lesssim} \sigma^2 \frac{1}{n_{1}}\Tr\Bigl( \bigl(\Sighat_{1} + \lambda_{1} \Ib\bigr)^{-\frac{1}{2}} \bSigma_\cT \bigl(\Sighat_{1} + \lambda_{1} \Ib\bigr)^{-\frac{1}{2}}\Bigr)\log\Bigl(\frac{1}{\delta_1}\Bigr)\\
&\lesssim \sigma^2 \frac{1}{n_{1}}\Tr\Bigl(\bSigma_\cT \bigl(\Sighat_{1} + \lambda_{1} \Ib\bigr)^{-1}\Bigr)\log\Bigl(\frac{1}{\delta_1}\Bigr)\\
& \stackrel{\text{(ii)}}{\lesssim} \sigma^2  \frac{1}{n_{1}}\Tr\Bigl(\bSigma_\cT \bigl(\bSigma_\cS + \lambda_{1} \Ib\bigr)^{-1}\Bigr)\log\Bigl(\frac{1}{\delta_1}\Bigr)
=\sigma^2\frac{1}{n_{1}}\Tr\bigl(\Sbar_{\lambda_{1}}\bigr)\log\Bigl(\frac{1}{\delta_1}\Bigr)
\end{align*}
with probability at least \(1-\frac{\delta_1}{3}\), where we used 
Lemma~\labelcref{lemma; trace simple inequality} in (i), 
and inequality (ii) follows from 
Lemma~\labelcref{lemma; moment ratio in E1,lemma; matrix inverse inequality}.

For the bias term, we have
\begin{align*}
\Bcr^2 
&\leq
\left\|\bSigma_\cT^{\frac{1}{2}}\bigl(\Sighat_{1}+\lambda_1\Ib\bigr)^{-\frac{1}{2}}\right\|_{\op}^2
\cdot
\left\|\lambda_1\bigl(\Sighat_{1}+\lambda_1\Ib\bigr)^{-\frac{1}{2}}\eta^\star\right\|_{\HH}^2 \\
&\lesssim
\lambda_1
\left\|\bSigma_\cT^{\frac{1}{2}}\bigl(\Sighat_{1}+\lambda_1\Ib\bigr)^{-1}\bSigma_\cT^{\frac{1}{2}}\right\|_{\op}
\|\eta^\star\|_{\HH}^2 \\
&\lesssim
\lambda_1\|\Sbar_{\lambda_1}\|_{\op}\|\eta^\star\|_{\HH}^2,
\end{align*}
where we used Corollary~\labelcref{corollary; application of Lemma second moment ratio} in the third line.
\end{proof}

\subsection{Bounding \(\mathscr{P}_0\) and \(\mathscr{P}_1\)}
\noindent
The terms \(\Pcr_0\) and \(\Pcr_1\) represent propagated errors from the nuisance estimators \(\hat{\theta}_0\) and \(\hat{\theta}_1\).

Note that
\begin{align*}
&\bSigma_{\cT}^{\frac{1}{2}} (\Sighat_{1} + \lambda_{1} \mathbf{I})^{-1} \Sighat_{1,1}\bigl({\theta}^\star_0-\hat{\theta}_0\bigr)
\\
&= \bSigma_{\cT}^{\frac{1}{2}} (\Sighat_{1} + \lambda_{1} \mathbf{I})^{-1} \Sighat_{1,1}\bigl(\Xb_{1,0}^\top \Xb_{1,0}+n_{1}\lambda_{0,0} \Ib\bigr)
^{-1} \bigl(-\Xb_{1,0}^\top \bm{\varepsilon}_{1,0} + n_{1}\lambda_{0,0}\theta_0^\star\bigr) \\
&=  \bSigma_{\cT}^{\frac{1}{2}} (\Sighat_{1} + \lambda_{1} \mathbf{I})^{-1}  \Sighat_{1,1} \bigl(\Sighat_{1,0}+\lambda_{0,0} \mathbf{I}\bigr)
^{-1} \Bigl(-\frac{1}{n_1}\Xb_{1,0}^\top \bm{\varepsilon}_{1,0} + \lambda_{0,0}\theta_0^\star\Bigr),
\end{align*}
we can write
\[
\Pcr_0 \leq \Vcr_0 + \Bcr_0,
\]
where 
\begin{align*}
\Vcr_0 &:= \Bigl\|   (\Sighat_{1} + \lambda_{1} \mathbf{I})^{-1}  \Sighat_{1,1} \bigl(\Sighat_{1,0}+\lambda_{0,0} \mathbf{I}\bigr)
^{-1} \frac{1}{n_1}\Xb_{1,0}^\top \bm{\varepsilon}_{1,0} \Bigr\|_{\bSigma_{\cT}},\\
\Bcr_0 &:= \Bigl\| (\Sighat_{1} + \lambda_{1} \mathbf{I})^{-1}  \Sighat_{1,1} \bigl(\Sighat_{1,0}+\lambda_{0,0} \mathbf{I}\bigr)
^{-1} \lambda_{0,0}\theta_0^\star\Bigr\|_{\bSigma_{\cT}}.
\end{align*}

\begin{lemma}[Propagated bias bound]\label{lemma; bounding B_0}
Under the good event \(\event\),
\[
\mathscr{B}_0^2 \lesssim R  \|\Sbar_{\lambda_{1}} \|_{\op}\lambda_{0,0}\norm{\theta^\star_0}_{\HH}^2.
\]
Here, \(\lesssim\) hides absolute constants.
\end{lemma}

\begin{proof}
By direct calculation,
\begin{align*}
\mathscr{B}_0 
&\leq \Bigl\|\bSigma_\cT^{\frac{1}{2}}\bigl(\Sighat_{1}+ \lambda_{1} \Ib\bigr)^{-\frac{1}{2}} \Bigr\|_{\op} 
\Bigl\|\bigl(\Sighat_{1}+ \lambda_{1} \Ib\bigr)^{-\frac{1}{2}} \Sighat_{1,1}^{\frac{1}{2}}\Bigr\|_{\op} 
\Bigl\|\Sighat_{1,1}^{\frac{1}{2}} \bigl(\Sighat_{1,0}+\lambda_{0,0} \mathbf{I}\bigr)
^{-\frac{1}{2}} \Bigr\|_{\op} 
\sqrt{\lambda_{0,0}}\|\theta_0^\star\|_{\HH}.
\end{align*}
First, by Lemma~\labelcref{lemma; moment ratio in E1} and Lemma~\labelcref{lemma; matrix inverse inequality}, 
\[
\bigl\| \bSigma_{\cT}^{\frac{1}{2}} (\Sighat_{1} + \lambda_{1} \mathbf{I})^{-1}\bSigma_{\cT}^{\frac{1}{2}}   \bigr\|_{\op} 
\lesssim  \bigl\| \bSigma_{\cT}^{\frac{1}{2}} \bigl(\bSigma_\cS + \lambda_{1} \mathbf{I}\bigr)^{-1}\bSigma_{\cT}^{\frac{1}{2}} \bigr\|_{\op} 
=\|\Sbar_{\lambda_{1}} \|_{\op}.
\]
Next, by Lemma~\labelcref{lemma; matrix inverse inequality} and the definition of \(\Sighat_{1,1}\),
\[
\Sighat_{1,1}^{\frac{1}{2}} \bigl(\Sighat_{1} + \lambda_{1} \mathbf{I}\bigr)^{-1}  \Sighat_{1,1}^{\frac{1}{2}} 
\preceq  \Sighat_{1,1}^{\frac{1}{2}}  \bigl(\Sighat_{1,1} + \lambda_{1} \mathbf{I}\bigr)^{-1}  \Sighat_{1,1}^{\frac{1}{2}}  
\preceq \Ib.
\]
Finally, by Corollary~\labelcref{corollary; application of Lemma second moment ratio},
\[
\bigl(\Sighat_{1,0}+\lambda_{0,0} \mathbf{I}\bigr)
^{-\frac{1}{2}}\Sighat_{1,1} \bigl(\Sighat_{1,0}+\lambda_{0,0} \mathbf{I}\bigr)
^{-\frac{1}{2}}  \preceq  cR \Ib
\]
for some absolute constant \(c>0\).
Putting these together yields
\begin{align*}
\mathscr{B}_0 
&\lesssim \bigl\|\Sbar_{\lambda_{1}} \bigr\|_{\op}^{\frac{1}{2}} \times 1 \times \sqrt{R}\sqrt{\lambda_{0,0}}\|\theta_0^\star\|_{\HH} 
\lesssim \lambda_{0,0}^{\frac{1}{2}}\bigl\|\Sbar_{\lambda_{1}} \bigr\|_{\op}^{\frac{1}{2}} \sqrt{R}\|\theta_0^\star\|_{\HH}.
\end{align*}
\end{proof}

\begin{lemma}[Propagated variance bound]\label{lemma; bounding V_0}
Under the good event \(\event\), with probability at least \(1-\frac{\delta_1}{3}\), we have 
\[
\mathscr{V}_0^2 \lesssim \sig^2 \frac{R}{n_1}\Tr\bigl(\Sbar_{\lambda_{1}}\bigr)\log\Bigl(\frac{1}{\delta_1}\Bigr).
\]
Here, \(\lesssim\) hides absolute constants.
\end{lemma}

\begin{proof}
By the Hanson-Wright inequality (Lemma~\labelcref{lemma; quadratic martingale}), with probability \(1-\frac{\delta_1}{3}\):
{\small \begin{align*}
\mathscr{V}_0^2 
&\lesssim 
{\small \frac{\sig^2 \log \bigl(\frac{3}{\delta_1}\bigr)}{n_1^2}\Tr\Bigl(\Xb_{1,0}\bigl(\Sighat_{1,0}+\lambda_{0,0} \mathbf{I}\bigr)^{-1}\Sighat_{1,1} (\Sighat_{1} + \lambda_{1} \mathbf{I}\bigr)^{-1}\bSigma_{\cT}\bigl(\Sighat_{1} + \lambda_{1} \mathbf{I}\bigr)^{-1}\Sighat_{1,1} \bigl(\Sighat_{1,0}+\lambda_{0,0} \mathbf{I}\bigr)^{-1}\Xb_{1,0}^\top \Bigr) }\\
&\lesssim \frac{\sig^2 \log \bigl(\frac{1}{\delta_1}\bigr)}{n_1}\Tr\Bigl(\bigl(\Sighat_{1,0}+\lambda_{0,0} \mathbf{I}\bigr)^{-1}\Sighat_{1,1} (\Sighat_{1} + \lambda_{1} \mathbf{I}\bigr)^{-1}\bSigma_{\cT}\bigl(\Sighat_{1} + \lambda_{1} \mathbf{I}\bigr)^{-1}\Sighat_{1,1} \bigl(\Sighat_{1,0}+\lambda_{0,0} \mathbf{I}\bigr)^{-1} \Sighat_{1,0}\Bigr) \\
&= \frac{\sig^2}{n_1} \Tr\Bigl( (\Sighat_{1} + \lambda_{1} \mathbf{I}\bigr)^{-1}\bSigma_{\cT}\bigl(\Sighat_{1} + \lambda_{1} \mathbf{I}\bigr)^{-1}\Sighat_{1,1} \bigl(\Sighat_{1,0}+\lambda_{0,0} \mathbf{I}\bigr)^{-1} \Sighat_{1,0} \bigl(\Sighat_{1,0}+\lambda_{0,0} \mathbf{I}\bigr)^{-1}\Sighat_{1,1}\Bigr) \log \bigl(\frac{1}{\delta_1}\bigr)\\
&\stackrel{\text{(i)}}{\lesssim} \frac{\sig^2}{n_1}\Tr\Bigl( (\Sighat_{1} + \lambda_{1} \mathbf{I}\bigr)^{-1}\bSigma_{\cT}\bigl(\Sighat_{1} + \lambda_{1} \mathbf{I}\bigr)^{-1}\Sighat_{1,1}\bigl(\Sighat_{1,0}+\lambda_{0,0} \mathbf{I}\bigr)^{-1}\Sighat_{1,1}\Bigr)\log \bigl(\frac{1}{\delta_1}\bigr) \quad (\text{Lemma~\labelcref{lemma; trace simple inequality}})\\
&\stackrel{\text{(ii)}}{\lesssim} \frac{\sig^2}{n_1}\Tr\Bigl( (\Sighat_{1} + \lambda_{1} \mathbf{I}\bigr)^{-1}\bSigma_{\cT}\bigl(\Sighat_{1} + \lambda_{1} \mathbf{I}\bigr)^{-1} R\Sighat_{1,1} \Bigr)\log \bigl(\frac{1}{\delta_1}\bigr) \quad (\text{Corollary~\labelcref{corollary; application of Lemma second moment ratio}})\\
&\lesssim \frac{\sig^2R}{n_1}\Tr\Bigl((\Sighat_{1} + \lambda_{1} \mathbf{I}\bigr)^{-\frac{1}{2}}\bSigma_{\cT}(\Sighat_{1} + \lambda_{1} \mathbf{I}\bigr)^{-\frac{1}{2}}(\Sighat_{1} + \lambda_{1} \mathbf{I}\bigr)^{-\frac{1}{2}}\Sighat_{1,1}(\Sighat_{1} + \lambda_{1} \mathbf{I}\bigr)^{-\frac{1}{2}}\Bigr) \log \bigl(\frac{1}{\delta_1}\bigr)\\
&\stackrel{\text{(iii)}}{\lesssim} \frac{\sig^2R}{n_1}\Tr\Bigl((\Sighat_{1} + \lambda_{1} \mathbf{I}\bigr)^{-\frac{1}{2}}\bSigma_{\cT}(\Sighat_{1} + \lambda_{1} \mathbf{I}\bigr)^{-\frac{1}{2}}\Bigr)\log \bigl(\frac{1}{\delta_1}\bigr) \quad (\because \Sighat_{1} \succeq \Sighat_{1,1}, \text{ and Lemma~\labelcref{lemma; trace simple inequality}}) \\
&\lesssim  \frac{\sig^2R}{n_1}\Tr\Bigl(\bSigma_{\cT}(\Sighat_{1} + \lambda_{1} \mathbf{I}\bigr)^{-1}\Bigr) \log \bigl(\frac{1}{\delta_1}\bigr)\\
&\stackrel{\text{(iv)}}{\lesssim} \frac{\sig^2R}{n_1}\Tr\Bigl(\bSigma_{\cT}\bigl(\bSigma_\cS + \lambda_{1} \mathbf{I}\bigr)^{-1}\Bigr) \log \bigl(\frac{1}{\delta_1}\bigr) \quad (\text{Lemma~\labelcref{lemma; moment ratio in E1,lemma; trace simple inequality}})  \\
&\lesssim  \frac{\sig^2R}{n_1}\Tr\bigl(\Sbar_{\lambda_{1}}\bigr)\log \bigl(\frac{1}{\delta_1}\bigr).
\end{align*} }
In (i) and (iii), we used Lemma~\labelcref{lemma; trace simple inequality}, 
in (ii), we applied Corollary~\labelcref{corollary; application of Lemma second moment ratio} and Lemma~\labelcref{lemma; trace simple inequality}, 
and in (iv), we used Lemma~\labelcref{lemma; moment ratio in E1,lemma; matrix inverse inequality}.
\end{proof}

Combining Lemmas~\labelcref{lemma; bounding B_0} and \labelcref{lemma; bounding V_0} yields:

\begin{corollary}\label{corollary; RA learner bound 2}
Under \(\event\), with probability at least \(1-\frac{\delta_1}{3}\),
\[
\Pcr_0^2 \lesssim R\|\Sbar_{\lambda_{1}} \|_{\op}\lambda_{0,0}\norm{\theta^\star_0}_{\HH}^2 + \sig^2R \frac{1}{n_1}\Tr\bigl(\Sbar_{\lambda_{1}}\bigr)\log\Bigl(\frac{1}{\delta_1}\Bigr).
\]
Here, \(\lesssim\) hides absolute constants.
\end{corollary}

Analogously, we obtain:

\begin{corollary}\label{corollary; RA learner bound 3}
Under \(\event\), with probability at least \(1-\frac{\delta_1}{3}\),
\[
\Pcr_1^2  \lesssim R\|\Sbar_{\lambda_{1}} \|_{\op}\lambda_{0,1}\norm{\theta^\star_1}_{\HH}^2 
+\sig^2 R \frac{1}{n_1}\Tr\bigl(\Sbar_{\lambda_{1}}\bigr)\log \Bigl(\frac{1}{\delta_1}\Bigr).
\]
Here, \(\lesssim\) hides absolute constants.
\end{corollary}

\subsection{Proof of Theorem~\labelcref{theorem; MSE bound RA learner}}
\noindent
Since \(n_1 \asymp n_{2} \asymp n\), combining Lemma~\labelcref{lemma; RA learner bound 1} with Corollaries~\labelcref{corollary; RA learner bound 2} and \ref{corollary; RA learner bound 3} shows that under \(\event\), with probability at least \(1-\delta_1\),
\begin{align*}
\cE_{\cT}(\hat{\eta}_{\bm{\lambda}}) 
&=\bigl\|(\hat{\eta}_{\bm{\lambda}} -\eta^\star)\bigr\|_{\bSigma_{\cT}}^2 
\leq 4\bigl(\Pcr_0^2 + \Pcr_1^2 + \Vcr^2 + \Bcr^2 \bigr)\\
&\lesssim R\|\Sbar_{\lambda_{1}} \|_{\op}\Bigl(\lambda_{0,0}\norm{\theta_0^\star}_\HH^2+\lambda_{0,1}\norm{\theta_1^\star}_\HH^2\Bigr)+\sig^2\frac{R}{n}\Tr\bigl(\Sbar_{\lambda_{1}}\bigr)\log \Bigl(\frac{1}{\delta_1}\Bigr) 
+ \lambda_{1}\|\Sbar_{\lambda_{1}} \|_{\op}\norm{\eta^\star}^2_\HH.
\end{align*}
Since \(\abs{\cH} \leq \log n\), we set \(\delta_1 =\frac{n^{-11}}{|\cH|}\). Then, with probability at least \(1-2n^{-11}\), the following holds for \emph{all} \(\bm{\lambda} \in \bm{\Lambda}\):
\begin{align*}
\cE_{\cT}(\hat{\eta}_{\bm{\lambda}}) 
&\lesssim  \sig^2\frac{R}{n}\Tr\bigl(\Sbar_{\lambda_{1}}\bigr)\log n 
+\lambda_{1} \|\Sbar_{\lambda_{1}} \|_{\op}\norm{\eta^\star}^2_\HH 
+R\|\Sbar_{\lambda_{1}} \|_{\op} M^2(\lambda_{0,0} + \lambda_{0,1}).
\end{align*}
Thus, with probability at least \(1-2n^{-11}\), we have for all estimators \(\hat{\eta}_{\bm{\lambda}} \in \cH\):
\[
\cE_{\cT}(\hat{\eta}_{\bm{\lambda}}) 
\lesssim \sig^2\frac{R}{n}\Tr\bigl(\Sbar_{\lambda_{1}}\bigr)\log n 
+\lambda_{1} \|\Sbar_{\lambda_{1}} \|_{\op}\norm{\eta^\star}^2_\HH 
+R\|\Sbar_{\lambda_{1}} \|_{\op}M^2(\lambda_{0,0} + \lambda_{0,1}),
\]
which completes the proof.

\subsection{Proof of Corollary~\labelcref{corollary; optimal MSE bound}}\label{subsection: proof main corollary}
\noindent
We prove that under the relaxed assumption---that is, when 
\(\|h^\star\|_{\mathcal{F}}\) is bounded by a sufficiently large value---the following inequality holds:
\[
\|h^\star\|^2_{\mathcal{F}} \lesssim R \left(\frac{n}{B \log n}\right)^{\frac{1}{2\ell}} = \widetilde{\Omega}\Bigl(R \cdot n_{\operatorname{eff}}^{\frac{1}{2\ell}}\Bigr).
\]
Under Assumption~\labelcref{assumption; overlap source target}, we have \(\|\Sbar_{\lambda_{1}} \|_{\op} \leq B\).
Additionally, recall that \(\lambda_{0,0}, \lambda_{0,1} = \frac{\xi \log n}{n}\).

Hence, the established MSE bound of Theorem~\ref{theorem; MSE bound RA learner} simplifies to
\begin{align*}
\cE_{\cT}(\hat{\eta}_{\bm{\lambda}})
&\lesssim \sig^2\frac{R}{n}\Tr\bigl(\Sbar_{\lambda_{1}}\bigr)\log n 
+\lambda_{1}B \norm{\eta^\star}^2_\HH 
+\frac{BR\xi\log n}{n}M^2.
\end{align*}
By hiding dependencies on \(\sigma\) and \(\xi\) (which we regard as universal constants), we obtain
\[
\cE_{\cT}(\hat{\eta}_{\bm{\lambda}})
\lesssim \frac{R}{n}\Tr\bigl(\Sbar_{\lambda_{1}}\bigr)\log n 
+  \lambda_{1} B \norm{\eta^\star}^2_\HH 
+\frac{BR\log n}{n}M^2.
\]

We consider two cases:

\paragraph*{Case 1:} \(\frac{R\log n}{n^2}\lesssim\|\eta^\star\|^2_\HH.\)
By applying Lemma~\labelcref{lemma; optimal trade-off lambda} and Corollary~\labelcref{corollary; optimal MSE in grid} with \(h = \frac{R}{n}\log n\), we deduce that the value of \(\lambda^\star\) in Corollary~\ref{corollary; optimal MSE in grid} satisfies
\[
\lambda^\star \asymp \Bigl(\frac{R\log n}{n}\Bigr)^{\alpha}B^{-(1-\alpha)} \|\eta^\star \|_\HH^{-2\alpha}.
\]
We can check that $\lambda^\star$ lies in our grid $\Lambda_1$, i.e., 
\[
\frac{\xi \log n}{n} \leq \lambda^\star \leq \frac{\xi n \log n}{2}.
\]
Hence, by Lemma~\labelcref{lemma; optimal trade-off lambda} and Corollary~\labelcref{corollary; optimal MSE in grid},
\[
\inf_{\bm{\lambda} \in \bm{\Lambda}}  \cE_{\cT}(\hat{\eta}_{\bm{\lambda}}) 
\lesssim \Bigl(\frac{BR}{n}\Bigr)^\alpha\norm{\eta^\star}^{2(1-\alpha)}_\HH (\log n)^\alpha 
+\frac{BR\log n}{n}M^2.
\]

\paragraph*{Case 2:} \(\frac{R\log n}{n^2} \gtrsim\|\eta^\star\|^2_\HH.\)
In this scenario, there exists \(\lambda_{1}' \in [1,2] \cap \Lambda_1 \), and we obtain
\begin{align*}
\inf_{\bm{\lambda} \in \bm{\Lambda}} \cE_{\cT}(\hat{\eta}_{\bm{\lambda}})  
&\lesssim \frac{R}{n}\Tr\bigl(\Sbar_{\lambda_{1}'}\bigr)\log n 
+  \lambda_{1}' B \norm{\eta^\star}^2_\HH 
+\frac{BR\log n}{n}M^2 \\ 
&\lesssim \frac{R}{n}\Tr\bigl(\Sbar_{1}\bigr)\log n 
+  2 B \norm{\eta^\star}^2_\HH 
+\frac{BR\log n}{n}M^2 \\  
&\lesssim \frac{BR\log n}{n}M^2,
\end{align*}
where, in the last inequality, we used \(M\ge 1\), \(n\ge 1\), the Case 2 condition, and
\begin{align*}
\Tr(\Sbar_{1}) \stackrel{\text{(i)}}{\lesssim} \sum_{j=1}^\infty \frac{B\mu_j}{\mu_j + 1} \lesssim B \sum_{j \ge 1} \mu_j \lesssim B,
\end{align*}
where (i) follows from the argument in the proof of Lemma~\ref{lemma; optimal trade-off lambda}.

\section{Proofs for Oracle Inequality (Proposition~\labelcref{proposition; oracle inequality MSE})}\label{section: proofs oracle inequalities}

\subsection{Guideline for Proofs}
\noindent
We first define the in-sample MSE and state the oracle inequality used in this section.
For any estimator \(\hat{h}\), define the in-sample MSE
\[
\cE_\cT^{\inn}(\hat{h}) := \frac{1}{n_\cT} \sum_{i=1}^{n_\cT} |\hat{h}(z_{0i}) - h^\star(z_{0i})|^2,
\]
which is the mean squared error over the target covariates in \(\cD_\cT\). 
We first present the oracle inequality for in-sample MSE:

\begin{lemma}[Oracle inequality for in-sample MSE]\label{lemma; in-sample MSE oracle}
Under Assumptions~\labelcref{assumption; boundedness,assumption; consistency and unconfoundedness,assumption; subGaussian noise,assumption; overlap source target,assumption; weak treatment overlap,Assumption; eigenvalue decay}, with probability at least \(1 - 2n^{-11}\), the following holds:
\[
\cE_{\cT}^{\inn}(\hat{h}_{\operatorname{final}}) 
\lesssim 
\min_{\bm\lambda \in \bm{\Lambda} } \cE_{\cT}^{\inn}(\hat{h}_{\bm{\lambda}}) 
+ 
\Ocr,
\]
where \(\lesssim\) hides absolute constants.
\end{lemma}

Next, we present a lemma showing that the in-sample MSE and the population MSE are comparable.
\begin{lemma}\label{lemma; comparability in-sample and population MSE}
With probability at least \(1-2n^{-11}\), the following inequalities hold simultaneously for all candidate estimators \(\hat h_{\bm\lambda}\in\cH\):
\begin{align*}
\cE^{\operatorname{in}}_{\cT}(\hat h_{\bm\lambda})&\lesssim \cE^{\operatorname{}}_{\cT}(\hat h_{\bm\lambda})+ \Ocr \\
\cE^{\operatorname{}}_{\cT}(\hat h_{\bm\lambda})&\lesssim \cE^{\operatorname{in}}_{\cT}(\hat h_{\bm\lambda})+ \Ocr .\\
\end{align*}
\end{lemma}

Combining these two lemmas yields Proposition~\ref{proposition; oracle inequality MSE}. 
Indeed, the proofs below establish both auxiliary results on the same good event \(\event\), with at most \(n^{-11}\) additional failure probability for each. 
Since \(\PP(\event)\geq 1-n^{-11}\), their intersection has probability at least \(1-3n^{-11}\). 
On this intersection,
\[
\cE_\cT(\hat h_{\operatorname{final}})
\lesssim
\cE_\cT^{\inn}(\hat h_{\operatorname{final}})+\Ocr
\lesssim
\min_{\bm\lambda\in\bm{\Lambda}}\cE_\cT^{\inn}(\hat h_{\bm\lambda})+\Ocr
\lesssim
\min_{\bm\lambda\in\bm{\Lambda}}\cE_\cT(\hat h_{\bm\lambda})+\Ocr,
\]
which proves Proposition~\ref{proposition; oracle inequality MSE}.
Thus, it remains to prove these two lemmas.

To prove Lemma~\ref{lemma; in-sample MSE oracle}, we aim to apply Lemma~\labelcref{lemma; loss model selection}. 
To this end, we first derive norm bounds related to the test outcomes $\{\tilde{h}(z_{0i})\}_{i=1}^{n_{\cT}}$. 

To prove Lemma~\ref{lemma; comparability in-sample and population MSE}, we invoke Lemma~\labelcref{lemma; key lemma in-sample population}. 
For this, we bound several norms of the estimators $\hat{\eta}_{\bm{\lambda}}$, and then apply Lemma~\labelcref{lemma; key lemma in-sample population} to prove the comparability lemma. 

\subsection{Norm Bounds for Test Outcomes}
\noindent
To establish an in-sample MSE oracle inequality, we apply Lemma~\ref{lemma; loss model selection}, which requires several norm bounds. 
Recall that the test outcomes are \(\{\tilde{h}(z_{0i})\}\).
Recall that \(\tilde{\eta}\) denotes the Hilbertian element corresponding to \(\tilde{h}\), and \(\tilde{\theta}_a\) denotes the Hilbertian element corresponding to \(\tilde{f}_a\).

For the test-outcome parameter \(\tilde{\eta} = \tilde{\theta}_1 - \tilde{\theta}_0\), we bound related \(\psi_2\) and Hilbert norms.
We define \(\EE[\tilde{\eta}]\) as the expectation with respect to the noise variables $\varepsilon_1, \dots, \varepsilon_n$.
Recall that we have defined $\Shat_\lambda := (\bSigma_{\cS} + \lambda \Ib)^{-\frac{1}{2}} \Sighat_{\cT} (\bSigma_{\cS} + \lambda \Ib)^{-\frac{1}{2}}$ for any \(\lambda >0\).

Using the split data $\cD_2$, we generate test outcomes and use them for model selection.
The following describes how these test outcomes are constructed in RKHS covariates in Algorithm~\labelcref{algorithm: model selection}.

\begin{enumerate}\label{equation: model selection nuisance estimator}
\item We calculate the nuisance estimator as 
\begin{align*}
&\tilde{\theta}_1 = (\Xb_{2,1}^\top \Xb_{2,1}+n_{2}\tilde{\lambda}_1\Ib)^{-1}\Xb_{2,1}^\top \yb_{2,1}\\
&\tilde{\theta}_0 = (\Xb_{2,0}^\top \Xb_{2,0}+n_{2}\tilde{\lambda}_0\Ib)^{-1}\Xb_{2,0}^\top \yb_{2,0}.
\end{align*}
where $\tilde{\theta}_0, \tilde{\theta}_1$ are the corresponding Hilbertian elements of $\tilde{f}_0$ and $\tilde{f}_1$.
\item We define \(\tilde{\eta}:= \tilde{\theta}_1 - \tilde{\theta}_0\), which corresponds to $\tilde{h}$.
\item Generate the test outcomes \(( x_{0i}^\top \tilde{\eta})_{i=1}^{n_{\cT}} = \Xb_\cT \tilde{\eta}\) and perform model selection.
\end{enumerate}

\begin{lemma}\label{lemma; test outcome norm preparation}
Under the good event $\event$, the test outcomes satisfy 
\begin{align*}
\frac{1}{n_\cT} ||\Xb_\cT (\EE[\tilde{\eta}] -\eta^\star) ||^2_2 &\lesssim 
R\tilde{\lambda}_1 \norm{\theta_1^\star}_{\HH}^2 \norm{\Shat_{\tilde{\lambda}_1}}_{\op} + R  \tilde{\lambda}_0 \norm{\theta_0^\star}_{\HH}^2 \norm{\Shat_{\tilde{\lambda}_0}}_{\op}\\
\|\Xb_\cT (\tilde{\eta} -\EE[\tilde{\eta}]) \|^2_{\psi_2} &\lesssim \sig^2  R\frac{n_{\cT}}{n_{2}}(\norm{\Shat_{\tilde{\lambda}_1}}_{\op} + \norm{\Shat_{\tilde{\lambda}_0}}_{\op})\\
\norm{\Shat_{\tilde{\lambda}_1}}_{\op} &\lesssim \norm{\Sbar_{\tilde{\lambda}_1}}_{\op} + \frac{n_{2} \log (n_\cT n)}{n_{\cT} \log n}\\
\norm{\Shat_{\tilde{\lambda}_0}}_{\op} &\lesssim \norm{\Sbar_{\tilde{\lambda}_0}}_{\op} + \frac{n_{2} \log (n_\cT n)}{n_{\cT} \log n}.
\end{align*}
Here, \(\lesssim\) hides absolute constants.
\end{lemma}

\begin{remark}
The first term is related to $\cE_{\cT}^{\inn}(\EE[\tilde{\eta}])$ in Lemma~\labelcref{lemma; loss model selection}, and the second term is related to the variance term of Lemma~\labelcref{lemma; loss model selection}.
By applying Lemma~\labelcref{lemma; loss model selection}, we first establish the oracle inequality for in-sample MSE.
\end{remark}

\begin{proof}
For any norm $ \| \cdot \|$, we have the decomposition
\begin{align*}
\|\Xb_\cT (\EE[\tilde{\eta}] -\eta^\star) \| &=  \|\Xb_\cT (\EE[\tilde{\theta}_1] -\theta_1^\star) - \Xb_\cT (\EE[\tilde{\theta}_0] -\theta_0^\star)  \|\\
&\leq \|\Xb_\cT (\EE[\tilde{\theta}_1] -\theta_1^\star) \|  +\| \Xb_\cT (\EE[\tilde{\theta}_0] -\theta_0^\star)  \|.
\end{align*}

We first prove the first inequality.
Observe that
\begin{align*}
\|\Xb_\cT (\EE[\tilde{\theta}_1] -\theta_1^\star) \|_{2} &=  \|\Xb_\cT ( \Sighat_{2,1} +\tilde{\lambda}_{1} \Ib)^{-1} \tilde{\lambda}_{1} \theta_1^\star  \|_{2}\\
&\leq \tilde{\lambda}_{1} \|\Xb_\cT (\Sighat_{2,1} + \tilde{\lambda}_{1} \Ib)^{-\frac{1}{2}} \|_{\op} \| ( \Sighat_{2,1} + \tilde{\lambda}_{1} \Ib)^{-\frac{1}{2}} \theta_1^\star  \|_{\HH} \\
&\leq \tilde{\lambda}_{1}^{\frac{1}{2}} \|\theta_1^\star  \|_{\HH} \|\Xb_\cT (\Sighat_{2,1} + \tilde{\lambda}_{1} \Ib)^{-\frac{1}{2}} \|_{\op}.
\end{align*}
Note that
\begin{align*}
\|\Xb_\cT (\Sighat_{2,1}+\tilde{\lambda}_{1} \Ib)^{-\frac{1}{2}} \|^2_{\op}  &= n_{\cT}  \|( \Sighat_{2,1}+\tilde{\lambda}_{1} \Ib)^{-\frac{1}{2}} \Sighat_{\cT} (\Sighat_{2,1}+\tilde{\lambda}_{1} \Ib)^{-\frac{1}{2}} \|_{\op}\\
&\lesssim n_{\cT}  \| \Sighat_{\cT}^{\frac{1}{2}} ( \Sighat_{2,1}+\tilde{\lambda}_{1} \Ib)^{-1}  \Sighat_{\cT}^{\frac{1}{2}} \|_{\op} \\
&\stackrel{\text{(i)}}{\lesssim}n_{\cT}  \| \Sighat_{\cT}^{\frac{1}{2}} (\Sigmatreated+\tilde{\lambda}_{1}\Ib)^{-1}  \Sighat_{\cT}^{\frac{1}{2}} \|_{\op} \\
&\stackrel{\text{(ii)}}{\lesssim} R n_{\cT}  \| \Sighat_{\cT}^{\frac{1}{2}} (\bSigma_\cS+\tilde{\lambda}_{1}\Ib)^{-1}  \Sighat_{\cT}^{\frac{1}{2}} \|_{\op} \\
&= Rn_{\cT}  \|{\Shat}_{\tilde{\lambda}_1} \|_{\op} 
\end{align*}
where (i) follows from Lemma~\labelcref{lemma; moment ratio in E1}, and (ii) follows from Lemma~\labelcref{lemma; second moment relation under weak overlap}.
Thus,
\begin{align*}
\frac{1}{\sqrt{n_{\cT}}}\|\Xb_\cT (\EE[\tilde{\theta}_1] -\theta_1^\star) \|_{2} 
&\lesssim (R \tilde{\lambda}_{1} \|{\Shat}_{\tilde{\lambda}_1} \|_{\op} )^{\frac{1}{2}}\|\theta_1^\star  \|_{\HH} .
\end{align*}
We can similarly obtain a bound for 
\(
\frac{1}{\sqrt{n_{\cT}}}\|\Xb_\cT (\EE[\tilde{\theta}_0] - \theta_0^\star)\|_{2},
\)
and these observations prove the first inequality.

To prove the second inequality, note that
{\small \begin{align*}
&\quad \|\Xb_\cT(\tilde{\eta} -\EE[\tilde{\eta}]) \|^2_{\psi_2} \\
&\lesssim \sig^2 \|\Xb_\cT (\Sighat_{2,1} + {\tilde{\lambda}}_{1} \Ib )^{-1} \frac{1}{n_{2} }\Xb_{2,1}  \|^2_{\op} + \sig^2 \|\Xb_\cT (\Sighat_{2,0} + {\tilde{\lambda}}_{0} \Ib )^{-1} \frac{1}{n_{2} }\Xb_{2,0}  \|^2_{\op}\\
&=  \sig^2 \frac{n_{\cT}}{n_{2}}\| \frac{1}{n_{\cT}} \Xb_\cT(\Sighat_{2,1} + {\tilde{\lambda}}_{1} \Ib  )^{-1} \Sighat_{2,1}(\Sighat_{2,1} + {\tilde{\lambda}}_{1} \Ib  )^{-1} \Xb_\cT^\top \|_{\op} \\
&\quad+ \sig^2\frac{n_{\cT}}{n_{2}}\| \frac{1}{n_{\cT}} \Xb_\cT(\Sighat_{2,0} + {\tilde{\lambda}}_{0} \Ib  )^{-1} \Sighat_{2,0}(\Sighat_{2,0} + {\tilde{\lambda}}_{0} \Ib  )^{-1} \Xb_\cT^\top \|_{\op}\\
&\leq \sig^2 \frac{n_{\cT}}{n_{2}}\| \frac{1}{n_{\cT}} \Xb_\cT(\Sighat_{2,1} + {\tilde{\lambda}}_{1} \Ib  )^{-1}  \Xb_\cT^\top  \|_{\op} +\sig^2\frac{n_{\cT}}{n_{2}}\| \frac{1}{n_{\cT}} \Xb_\cT(\Sighat_{2,0} + {\tilde{\lambda}}_{0} \Ib  )^{-1}  \Xb_\cT^\top  \|_{\op}\\
&\stackrel{\text{(i)}}{\lesssim} \sig^2 \frac{n_{\cT}}{n_{2}} R\|\frac{1}{n_{\cT}}  \Xb_\cT(\bSigma_\cS + {\tilde{\lambda}}_{1} \Ib  )^{-1}  \Xb_\cT^\top  \|_{\op}+\sig^2  \frac{n_{\cT}}{n_{2}} R\|\frac{1}{n_{\cT}}  \Xb_\cT(\bSigma_\cS + {\tilde{\lambda}}_{0} \Ib  )^{-1}  \Xb_\cT^\top  \|_{\op} \quad (\text{By Lemma~\labelcref{lemma; moment ratio in E1,lemma; second moment relation under weak overlap}}) \\
&\lesssim \sig^2\frac{n_{\cT}}{n_{2}} R \|\frac{1}{n_{\cT}} (\bSigma_\cS + {\tilde{\lambda}}_{1} \Ib  )^{-\frac{1}{2}}  \Xb_\cT^\top \Xb_\cT (\bSigma_\cS + {\tilde{\lambda}}_{1} \Ib  )^{-\frac{1}{2}} \|_{\op} \\
&\quad +\sig^2 \frac{n_{\cT}}{n_{2}} R\|\frac{1}{n_{\cT}} (\bSigma_\cS + {\tilde{\lambda}}_{0} \Ib  )^{-\frac{1}{2}}  \Xb_\cT^\top \Xb_\cT(\bSigma_\cS + {\tilde{\lambda}}_{0} \Ib  )^{-\frac{1}{2}} \|_{\op}\\
&= \sig^2\frac{n_{\cT}}{n_{2}} R  \|{\Shat}_{\tilde{\lambda}_1}  \|_{\op}+\sig^2\frac{n_{\cT}}{n_{2}} R  \|{\Shat}_{\tilde{\lambda}_0}  \|_{\op} 
\end{align*}}
where inequality (i) follows from Lemma~\labelcref{lemma; moment ratio in E1,lemma; second moment relation under weak overlap}.

Next, we prove the third inequality. 
Note that under the good event \(\event\),
\begin{align*}
\norm{\Shat_{\tilde{\lambda}_1}}_{\op} &= \norm{ (\bSigma_\cS+ \tilde{\lambda}_1 \Ib)^{-\frac{1}{2}} {\Sighat}_\cT (\bSigma_\cS+ \tilde{\lambda}_1 \Ib)^{-\frac{1}{2}}  }_{\op} \\
&\stackrel{\text{(i)}}{\lesssim }\norm{ (\bSigma_\cS+ \tilde{\lambda}_1 \Ib)^{-\frac{1}{2}} (\bSigma_\cT +\frac{\xi}{n_{\cT}}\log(n n_\cT)) (\bSigma_\cS+ \tilde{\lambda}_1 \Ib)^{-\frac{1}{2}}  }_{\op} \\
&\lesssim \norm{\Sbar_{\tilde{\lambda}_1}}_{\op} + \|(\bSigma_\cS+ \tilde{\lambda}_1 \Ib)^{-\frac{1}{2}} \times \left(\frac{\xi}{n_{\cT}}\log(n n_\cT)\right)\Ib  \times (\bSigma_\cS+ \tilde{\lambda}_1 \Ib)^{-\frac{1}{2}}\|_{\op} \\
&\lesssim  \norm{\Sbar_{\tilde{\lambda}_1}}_{\op} + \frac{n_{2}}{n_{\cT}} \frac{\log (n_\cT n)}{\log n}.
\end{align*}
where (i) holds by the definition of $\event$.
Therefore, under \(\event\),
\begin{align*}
&\norm{\Shat_{\tilde{\lambda}_1}}_{\op} \lesssim \norm{\Sbar_{\tilde{\lambda}_1}}_{\op} + \frac{n_{2}}{n_{\cT}} \frac{\log (n_\cT n)}{\log n} 
\end{align*}
and the fourth inequality can be proved in the same way.

\end{proof}

\begin{corollary}[Norm bounds for test outcomes]\label{corollary; test label norm bounds}
Under the event $\event$, the following holds:
\begin{align*}
\|\Xb_\cT (\tilde{\eta}-\EE[\tilde{\eta}]) \|^2_{\psi_2} &\lesssim \sig^2 \frac{n_{\cT}}{n_{2}} BR +\sig^2R \frac{\log (nn_\cT)}{\log n} \\
\|\frac{1}{\sqrt{n_{\cT}}}\Xb_\cT (\EE[\tilde{\eta}] -\eta^\star) \|^2_{2}  &\lesssim \xi RM^2 \log n \Bigl(\frac{B}{n} +\frac{1}{n_{\cT}}\frac{\log (nn_\cT)}{\log n}\Bigr).
\end{align*}
Here, \(\lesssim\) hides absolute constants.
\end{corollary}

\begin{proof}
By Lemma~\labelcref{lemma; test outcome norm preparation}, we have
\begin{align*}
\|\Xb_\cT (\tilde{\eta}-\EE[\tilde{\eta}]) \|^2_{\psi_2} & \lesssim \sig^2  R\frac{n_{\cT}}{n_{2}}(\norm{\Shat_{\tilde{\lambda}_1}}_{\op} + \norm{\Shat_{\tilde{\lambda}_0}}_{\op}) \\
&\lesssim    \sig^2  \frac{n_{\cT}}{n_{2}} R \|\Sbar_{\tilde{\lambda}_1}  \|_{\op}+\sig^2  \frac{n_{\cT}}{n_{2}} R \|\Sbar_{\tilde{\lambda}_0}  \|_{\op} + \sig^2 R  \frac{\log (nn_\cT)}{\log n} \\
&\lesssim \sig^2 \frac{n_{\cT}}{n_{2}} BR +\sig^2R \frac{\log (nn_\cT)}{\log n} \\
&= \sig^2 R  \Bigl(\frac{n_\cT}{n}B + \frac{\log (nn_\cT)}{\log n}\Bigr).
\end{align*}

For the second inequality, by Lemma~\labelcref{lemma; test outcome norm preparation}, we have
\begin{align*}
&\|\frac{1}{\sqrt{n_{\cT}}}\Xb_\cT (\EE[\tilde{\eta}] -\eta^\star) \|^2_{2} \\
&\lesssim  R  \tilde{\lambda}_1 \norm{\theta_1^\star}_{\HH}^2 \norm{\Shat_{\tilde{\lambda}_1}}_{\op} +  R  \tilde{\lambda}_0 \norm{\theta_0^\star}_{\HH}^2 \norm{\Shat_{\tilde{\lambda}_0}}_{\op}\\
&\lesssim  R \tilde{\lambda}_1 \norm{\theta_1^\star}_{\HH}^2 \norm{\Sbar_{\tilde{\lambda}_1}}_{\op} 
+  R  \tilde{\lambda}_0 \norm{\theta_0^\star}_{\HH}^2 \norm{\Sbar_{\tilde{\lambda}_0}}_{\op} 
+   R \tilde{\lambda}_1 \norm{\theta_1^\star}_{\HH}^2\frac{n_{2}}{n_{\cT}} \frac{\log (nn_\cT)}{\log n}+  R \tilde{\lambda}_0 \norm{\theta_0^\star}_{\HH}^2 \frac{n_{2}}{n_{\cT}}\frac{\log (nn_\cT)}{\log n}\\
&\leq BR  \tilde{\lambda}_1 \norm{\theta_1^\star}_{\HH}^2  + BR \tilde{\lambda}_0 \norm{\theta_0^\star}_{\HH}^2  +  R \tilde{\lambda}_1n_{2} \norm{\theta_1^\star}_{\HH}^2\frac{1}{n_{\cT}}\frac{\log (nn_\cT)}{\log n}+  R \tilde{\lambda}_0 n_{2} \norm{\theta_0^\star}_{\HH}^2 \frac{1}{n_{\cT}}\frac{\log (nn_\cT)}{\log n} \\
&\lesssim BR \frac{\xi\log n }{n} \max\bigl(\norm{\theta_0^\star}_{\HH},\norm{\theta_1^\star}_{\HH}\bigr)^2 +R \max\bigl(\norm{\theta_0^\star}_{\HH},\norm{\theta_1^\star}_{\HH}\bigr)^2\frac{\xi \log n}{n_{\cT}}\frac{\log (nn_\cT)}{\log n} \\
&\lesssim \xi RM^2  \Bigl(\frac{B}{n}\log n +\frac{\log (nn_\cT)}{n_{\cT}}\Bigr).
\end{align*}    
\end{proof}

\subsection{Proof of Lemma~\labelcref{lemma; in-sample MSE oracle}}

\begin{proof}
We aim to apply Lemma~\labelcref{lemma; loss model selection} to $\tilde{h}$, under the event $\event$ and given $\cD_1$.
Under the good event \(\event\), we first bound $ \cE_\cT^{\inn}(\EE[\tilde{\eta}])$.
By using Corollary~\labelcref{corollary; test label norm bounds}, we have
\begin{align*}
\cE_\cT^{\inn}(\EE[\tilde{\eta}])= \|\frac{1}{\sqrt{n_{\cT}}}\Xb_\cT (\EE[\tilde{\eta}] -\eta^\star)\|_{2}^2 \lesssim\xi RM^2  \Bigl(\frac{B}{n}\log n +\frac{\log (nn_\cT)}{n_{\cT}}\Bigr).
\end{align*}

Next, we bound \(V^2\) of Lemma~\labelcref{lemma; loss model selection} under the event \(\event\).
Using Corollary~\labelcref{corollary; test label norm bounds}, we obtain
\begin{align*}
V^2 := \|\Xb_\cT (\tilde{\eta}- \EE[\tilde{\eta}]) \|_{\psi_2}^2 &\lesssim 
\sig^2 R \Bigl(\frac{n_\cT}{n}B + \frac{\log (nn_\cT)}{\log n}\Bigr).
\end{align*}

Since \(\tilde{\lambda}_1, \tilde{\lambda}_0 \asymp \frac{\xi \log n}{n}\), applying Lemma~\labelcref{lemma; loss model selection} gives the desired in-sample MSE oracle inequality:
\begin{align*}
&\cE_{\cT}^{\inn}(\hat{\eta}_{\operatorname{final}}) \\
&\leq \min_{\hat{\eta}_{\bm{\lambda}} \in \cH} \cE_{\cT}^{\inn}(\hat{\eta}_{\bm{\lambda}})
+\xi RM^2  \Bigl(\frac{B}{n}\log n +\frac{\log (nn_\cT)}{n_{\cT}}\Bigr)
+ \sig^2 (\log n)  R  \Bigl(\frac{1}{n}B + \frac{1}{n_\cT}\frac{\log (nn_\cT)}{\log n} \Bigr)
\\
&\leq \min_{\hat{\eta}_{\bm{\lambda}} \in \cH} \cE_{\cT}^{\inn}(\hat{\eta}_{\bm{\lambda}})
+\xi RM^2  \Bigl(\frac{B}{n}\log n +\frac{\log (nn_\cT)}{n_{\cT}}\Bigr)
+ \sig^2  R  \Bigl(\frac{B}{n}\log n + \frac{\log (nn_\cT)}{n_\cT}  \Bigr)
\\
&\leq \min_{\hat{\eta}_{\bm{\lambda}} \in \cH} \cE_{\cT}^{\inn}(\hat{\eta}_{\bm{\lambda}})
+R(\xi M^2 + \sig^2 ) \Bigl(\frac{B}{n}\log n +\frac{\log (nn_\cT)}{n_{\cT}}\Bigr).
\end{align*} 
This holds with probability $1-n^{-11}$ under the event $\event$.
Since $R(\xi M^2 + \sig^2 ) \Bigl(\frac{B}{n}\log n +\frac{\log (nn_\cT)}{n_{\cT}}\Bigr) \lesssim \Ocr$, the desired result follows.
\end{proof}

\subsection{Norm Bounds for RA Learner Estimator}\label{subsection: norm bounds RA learner estimator}
\noindent
Next, we prepare to prove Proposition~\labelcref{proposition; oracle inequality MSE}.
Pick any fixed $\hat{\eta}_{\bm{\lambda}} \in \cH$. 
To apply Lemma~\labelcref{lemma; key lemma in-sample population}, we establish several norm bounds for that estimator.
Recall that, in Appendix~\labelcref{section: proof RA learner}, we proved that the estimation error can be decomposed as
\begin{align*}
{\hat{\eta}_{\bm{\lambda}} -\eta^\star } &= (\Sighat_{1} + \lambda_{1} \mathbf{I})^{-1} \Bigl( \frac{1}{n_{1}} \bigl(\sum_{a_{1i}=1}x_{1i} \varepsilon_{1i}+\sum_{a_{1i}=0} -x_{1i}\varepsilon_{1i}\bigr) -\lambda_{1} \eta^\star\Bigr)\\
&\quad + (\Sighat_{1} + \lambda_{1} \mathbf{I})^{-1} \Bigl( \Sighat_{1,1}({\theta}^\star_0-\hat{\theta}_0)
+ \Sighat_{1,0} (\hat{\theta}_1-{\theta}^\star_1)  \Bigr).
\end{align*}

We define \(\bar{\eta}_{\bm{\lambda}} = \EE[\hat{\eta}_{\bm{\lambda}}],\bar{\theta}_1 = \EE[\hat{\theta}_1],\bar{\theta}_0 = \EE[\hat{\theta}_0] \) where the expectations are taken in the noise variables \(\{\varepsilon_i\}_{i=1}^n\).

\begin{lemma}\label{lemma; bias H bound}
The following holds under the event $\event$:
\begin{align*}
\norm{\bar{\eta}_{\bm{\lambda}} - \eta^\star}_{\HH} \lesssim \sqrt{R} \| \theta_0^\star\|_{\HH}+  \sqrt{R} \| \theta_1^\star\|_{\HH} + \|\eta^\star \|_{\HH}.
\end{align*}
Here, \(\lesssim\) hides absolute constants.
\end{lemma}

\begin{proof}
By the error decomposition established in Appendix~\labelcref{section: proof RA learner}, we can write
\begin{align*}
\norm{\bar{\eta}_{\bm{\lambda}} - \eta^\star}_{\HH} &\leq \|(\Sighat_{1} + \lambda_{1} \mathbf{I})^{-1} \Sighat_{1,1}(\bar{\theta}_0 -\theta_0^\star) \|_{\HH} + \| (\Sighat_{1} + \lambda_{1} \mathbf{I})^{-1} \Sighat_{1,0}(\bar{\theta}_1 -\theta_1^\star) \|_{\HH} \\
&\quad+ \|(\Sighat_{1} + \lambda_{1} \mathbf{I})^{-1}\lambda_{1} \eta^\star  \|_{\HH}.
\end{align*}
In the proof, we only bound the term $I_1 := \|(\Sighat_{1} + \lambda_{1} \mathbf{I})^{-1} \Sighat_{1,1}(\bar{\theta}_0 -\theta_0^\star) \|_{\HH}$; the other term $ \| (\Sighat_{1} + \lambda_{1} \mathbf{I})^{-1} \Sighat_{1,0}(\bar{\theta}_1 -\theta_1^\star) \|_{\HH} $ can be bounded similarly.

Observe
\begin{align*}
I_1 &\leq \| (\Sighat_{1} + \lambda_{1} \mathbf{I})^{-1}  \Sighat_{1,1} (\Sighat_{1,0}+\lambda_{0,0} \mathbf{I})^{-1} \lambda_{0,0}\theta_0^\star\|_{\HH} \\
&\leq \lambda_{0,0} \| (\Sighat_{1} + \lambda_{1} \mathbf{I})^{-1}  \Sighat_{1,1} (\Sighat_{1,0}+\lambda_{0,0} \mathbf{I})^{-\frac{1}{2}}\|_{\op} \|(\Sighat_{1,0}+\lambda_{0,0} \mathbf{I})^{-\frac{1}{2}}\theta_0^\star \|_{\HH}.
\end{align*}
We first examine
\begin{align*}
&\| (\Sighat_{1} + \lambda_{1} \mathbf{I})^{-1}  \Sighat_{1,1} (\Sighat_{1,0}+\lambda_{0,0} \mathbf{I})^{-\frac{1}{2}}\|^2_{\op}  \\
&= \| (\Sighat_{1} + \lambda_{1} \mathbf{I})^{-1}  \Sighat_{1,1} (\Sighat_{1,0}+\lambda_{0,0} \mathbf{I})^{-1} \Sighat_{1,1}(\Sighat_{1} + \lambda_{1} \mathbf{I})^{-1}\|_{\op} \\
&\stackrel{\text{(i)}}{\lesssim} R \| (\Sighat_{1} + \lambda_{1} \mathbf{I})^{-1}  \Sighat_{1,1}(\Sighat_{1} + \lambda_{1} \mathbf{I})^{-1}\|_{\op} \\
&\stackrel{\text{(ii)}}{\leq} R  \| (\Sighat_{1} + \lambda_{1} \mathbf{I})^{-1}  \Sighat_{1}(\Sighat_{1} + \lambda_{1} \mathbf{I})^{-1}\|_{\op} \\
&\lesssim R  \| (\Sighat_{1} + \lambda_{1} \mathbf{I})^{-1}\|_{\op} \\
&\lesssim R \frac{1}{\lambda_{1}},
\end{align*}
where we applied Corollary~\labelcref{corollary; application of Lemma second moment ratio} for step (i), and step (ii) holds by $\Sighat_{1,1} \preceq \Sighat_{1}$.
Hence,
\begin{align*}
I_1 \leq  \sqrt{\lambda_{0,0}} \sqrt{\frac{R}{\lambda_{1}}} \|\theta_0^\star \|_{\HH} \leq \sqrt{R}\|\theta_0^\star \|_{\HH},
\end{align*}
since our algorithm forces \(\lambda_{1} \geq \lambda_{0,0}, \lambda_{0,1}\).

For the term \(\|(\Sighat_{1} + \lambda_{1} \mathbf{I})^{-1}\lambda_{1} \eta^\star  \|_{\HH}\), we can bound it using 
\begin{align*}
\|(\Sighat_{1} + \lambda_{1} \mathbf{I})^{-1}\lambda_{1}\|_{\op} &\leq 1
\end{align*}
and thus \(\|(\Sighat_{1} + \lambda_{1} \mathbf{I})^{-1}\lambda_{1} \eta^\star  \|_{\HH} \leq \|\eta^\star \|_{\HH}\).
\end{proof}

We next bound the \(\psi_2\) norm $\norm{\hat{\eta}_{\bm{\lambda}}-\bar{\eta}_{\bm{\lambda}}}_{\psi_2}$.

\begin{lemma}\label{lemma; variance psi2 bound}
The following holds under the event $\event$:
\begin{align*}
\norm{\hat{\eta}_{\bm{\lambda}}-\bar{\eta}_{\bm{\lambda}}}_{\psi_2} \lesssim \sigma\frac{ \sqrt{R}}{\sqrt{n_1 \lambda_{1}}} + \sigma\frac{1}{\sqrt{n_{1}\lambda_{1}}} \lesssim \sigma \frac{\sqrt{R}}{ \sqrt{\xi \log n}}.
\end{align*}   
Here, \(\lesssim\) hides absolute constants.
\end{lemma}

\begin{proof}
By the error decomposition established in Appendix~\labelcref{section: proof RA learner}, we can write
\begin{align*}
\norm{ \hat{\eta}_{\bm{\lambda}} -\bar{\eta}_{\bm{\lambda}} }_{\psi_2} &= \|  (\Sighat_{1} + \lambda_{1} \mathbf{I})^{-1} \Sighat_{1,1}(\hat{\theta}_0 -\bar{\theta}_0) \|_{\psi_2} + \| (\Sighat_{1} + \lambda_{1} \mathbf{I})^{-1} \Sighat_{1,0}(\hat{\theta}_1 -\bar{\theta}_1) \|_{\psi_2} \\
&\quad+ \| (\Sighat_{1} + \lambda_{1} \mathbf{I})^{-1} \frac{1}{n_{1}}\bigl(  \sum_{a_{1i}=1}x_{1i}\varepsilon_{1i}+\sum_{a_{1i}=0} (- x_{1i}\varepsilon_{1i}) \bigr) \|_{\psi_2}.
\end{align*}
In the proof, we only bound $I = \|  (\Sighat_{1} + \lambda_{1} \mathbf{I})^{-1} \Sighat_{1,1}(\hat{\theta}_0 -\bar{\theta}_0) \|_{\psi_2}$; the other term $\| (\Sighat_{1} + \lambda_{1} \mathbf{I})^{-1} \Sighat_{1,0}(\hat{\theta}_1 -\bar{\theta}_1) \|_{\psi_2}$ can be bounded similarly.

Observe that 
\begin{align*}
I&= \Bigl\|   (\Sighat_{1} + \lambda_{1} \mathbf{I})^{-1}  \Sighat_{1,1} (\Sighat_{1,0}+\lambda_{0,0} \mathbf{I})
^{-1} \frac{1}{n_1} \Xb_{1,0}^\top \bm{\varepsilon}_{1,0}\Bigr\|_{\psi_2}.   
\end{align*}
We view this as an operator acting on $\bm{\varepsilon}_{1,0}$ and bound the operator norm of the Hilbertian operator $(\Sighat_{1} + \lambda_{1} \mathbf{I})^{-1}  \Sighat_{1,1} (\Sighat_{1,0}+\lambda_{0,0} \mathbf{I})^{-1} \frac{1}{n_1} \Xb_{1,0}^\top$. 
Define its operator norm as $A$, 
\[
A := \Bigl\|(\Sighat_{1} + \lambda_{1} \mathbf{I})^{-1}  \Sighat_{1,1} (\Sighat_{1,0}+\lambda_{0,0} \mathbf{I})^{-1} \frac{1}{n_1} \Xb_{1,0}^\top \Bigr\|_{\op}.
\]
Observe that 
\begin{align*}
A^2 & \leq \frac{1}{n_1} \norm{(\Sighat_{1} + \lambda_{1} \mathbf{I})^{-1}  \Sighat_{1,1} (\Sighat_{1,0}+\lambda_{0,0} \mathbf{I})^{-1} \Sighat_{1,0}(\Sighat_{1,0}+\lambda_{0,0} \mathbf{I})^{-1} \Sighat_{1,1}  (\Sighat_{1} + \lambda_{1} \mathbf{I})^{-1}  }_{\op}  \\
&\lesssim \frac{1}{n_1} \norm{(\Sighat_{1} + \lambda_{1} \mathbf{I})^{-1}  \Sighat_{1,1} (\Sighat_{1,0}+\lambda_{0,0} \mathbf{I})^{-1}  \Sighat_{1,1} (\Sighat_{1} + \lambda_{1} \mathbf{I})^{-1} }_{\op} \\
&\stackrel{\text{(i)}}{\lesssim } R  \frac{1}{n_1} \norm{(\Sighat_{1} + \lambda_{1} \mathbf{I})^{-1}  \Sighat_{1,1}   (\Sighat_{1} + \lambda_{1} \mathbf{I})^{-1}  }_{\op}\\
&\stackrel{\text{(ii)}}{\lesssim } R  \frac{1}{n_1} \norm{(\Sighat_{1} + \lambda_{1} \mathbf{I})^{-1}  \Sighat_{1}   (\Sighat_{1} + \lambda_{1} \mathbf{I})^{-1}  }_{\op} \\
&\lesssim  R  \frac{1}{n_1} \norm{(\Sighat_{1} + \lambda_{1} \mathbf{I})^{-1}   }_{\op}\\
&\lesssim R  \frac{1}{n_1 \lambda_{1}},
\end{align*}
where in the third line (i), we used Corollary~\labelcref{corollary; application of Lemma second moment ratio}, and step (ii) holds by the relation $\Sighat_{1,1} \preceq \Sighat_{1}$.
Thus,
\begin{align*}
A \lesssim \frac{ \sqrt{R}}{\sqrt{n_1 \lambda_{1}}},
\end{align*}    
and hence
\begin{align*}
I \lesssim \sig \frac{ \sqrt{R}}{\sqrt{n_1 \lambda_{1}}}.
\end{align*}

For the other term \(\Bigl\| (\Sighat_{1} + \lambda_{1} \mathbf{I})^{-1} \frac{1}{n_{1}}\Bigl(  \sum_{a_{1i}=1}x_{1i}\varepsilon_{1i}+\sum_{a_{1i}=0} (- x_{1i}\varepsilon_{1i}) \Bigr) \Bigr\|_{\psi_2}\), we can interpret it as 
\begin{align*}
(\Sighat_{1} + \lambda_{1} \mathbf{I})^{-1} \frac{1}{n_{1}} \Xb_1^\top \bm{\varepsilon}'_1 ,
\end{align*}
where \(\bm{\varepsilon}_1' \) is a vectorized form of \(\varepsilon_{1i}(-1)^{a_{1i}+1}\).
Under Assumption~\labelcref{assumption; consistency and unconfoundedness}, \(\varepsilon_{1i}(-1)^{a_{1i}+1}\) is a mean-zero and sub-Gaussian noise given \(\{(x_i,a_i)\}_{i=1}^n\).
The same argument bounds the operator norm of 
\begin{align*}
\Bigl\| (\Sighat_{1} + \lambda_{1} \mathbf{I})^{-1} \frac{1}{n_{1}} \Xb_1^\top  \Bigr\|^2_{\op} &\leq \frac{1}{n_{1} } \bigl\|(\Sighat_{1} + \lambda_{1} \mathbf{I})^{-1}\Sighat_{1}(\Sighat_{1} + \lambda_{1} \mathbf{I})^{-1} \bigr\|_{\op} \\
&\lesssim \frac{1}{n_{1}\lambda_{1}},
\end{align*}
and this completes the proof.
\end{proof}

\subsection{Proof of Lemma~\ref{lemma; comparability in-sample and population MSE}}
\noindent
For each candidate estimator \(\hat{\eta}_{\bm{\lambda}}\in \cH\), define the random function
\[
f_{\bm{\lambda}}(z)
:=\hat{h}_{\bm{\lambda}}(z)-h^\star(z)
= \phi(z)^\top(\hat{\eta}_{\bm{\lambda}}-\eta^\star).
\]
The candidate estimators are trained only on the source sample, and hence \(f_{\bm{\lambda}}\) is independent of the target covariates \(\{z_{0i}\}_{i=1}^{n_\cT}\).
Moreover,
\[
\|f_{\bm{\lambda}}\|_{n_\cT}^2
=\cE_\cT^{\inn}(\hat{\eta}_{\bm{\lambda}}),
\qquad
\|f_{\bm{\lambda}}\|_{L^2}^2
=\cE_\cT(\hat{\eta}_{\bm{\lambda}}).
\]

We verify the conditions of Lemma~\labelcref{lemma; key lemma in-sample population}.
Let
\[
\varepsilon_0=(n n_\cT |\cH|)^{-12},
\qquad
\delta_0=\frac{n^{-11}}{2|\cH|}.
\]
By Lemmas~\labelcref{lemma; bias H bound,lemma; variance psi2 bound}, for every \(\hat{\eta}_{\bm{\lambda}}\in\cH\),
\[
\|\bar{\eta}_{\bm{\lambda}}-\eta^\star\|_{\HH}\lesssim \sqrt{R}M,
\qquad
\|\hat{\eta}_{\bm{\lambda}}-\bar{\eta}_{\bm{\lambda}}\|_{\psi_2}
\lesssim \frac{\sigma\sqrt{R}}{\sqrt{\xi\log n}},
\]
where we use \(\|\eta^\star\|_\HH\leq 2M\).
Since \(\sup_z\|\phi(z)\|_\HH\leq \sqrt{\xi}\), the preceding display implies
\[
\PP\bigl(|f_{\bm{\lambda}}(z_0)|>r\bigr)\leq \varepsilon_0,
\qquad
r \asymp \sqrt{R}\left(M\sqrt{\xi}
+\sigma\sqrt{\frac{\log(1/\varepsilon_0)}{\log n}}\right),
\]
where \(z_0\sim \cP_\cT\). 
The same bounds also give
\[
\EE |f_{\bm{\lambda}}(z_0)|^4 \leq U^4,
\qquad
U^2\asymp R(\xi M^2+\sigma^2).
\]
Applying Lemma~\labelcref{lemma; key lemma in-sample population} to each \(f_{\bm{\lambda}}\) and taking a union bound over \(\cH\), with probability at least \(1-n^{-11}\), for all \(\hat{\eta}_{\bm{\lambda}}\in\cH\),
\[
\left|
\sqrt{\cE_\cT^{\inn}(\hat{\eta}_{\bm{\lambda}})}
-\sqrt{\cE_\cT(\hat{\eta}_{\bm{\lambda}})}
\right|
\lesssim
r\sqrt{\frac{\log(2/\delta_0)}{n_\cT}}
+U\varepsilon_0^{1/4}.
\]
Using \((a+b)^2\leq 2a^2+2b^2\) in both directions, this implies
\[
\cE_\cT^{\inn}(\hat{\eta}_{\bm{\lambda}})
\lesssim
\cE_\cT(\hat{\eta}_{\bm{\lambda}})+\Ocr',
\qquad
\cE_\cT(\hat{\eta}_{\bm{\lambda}})
\lesssim
\cE_\cT^{\inn}(\hat{\eta}_{\bm{\lambda}})+\Ocr',
\]
where
\[
\Ocr'
:=
\frac{r^2\log(2/\delta_0)}{n_\cT}
+U^2\varepsilon_0^{1/2}.
\]
Since \(|\cH|\leq \log n\), we have
\[
\Ocr'
\lesssim
R(\xi M^2+\sigma^2)
\left(\frac{\log(nn_\cT)}{n_\cT}+\frac{1}{n^5}\right)
\lesssim \Ocr.
\]
Therefore, with probability at least \(1-n^{-11}\), for all \(\hat{\eta}_{\bm{\lambda}} \in \cH\),
\begin{align*}
&\cE_{\cT}^{\inn}(\hat{\eta}_{\bm{\lambda}}) \lesssim \cE_\cT(\hat{\eta}_{\bm{\lambda}}) +\Ocr',\\
&\cE_{\cT}(\hat{\eta}_{\bm{\lambda}}) \lesssim  \cE_{\cT}^{\inn}(\hat{\eta}_{\bm{\lambda}}) +\Ocr'.
\end{align*}
In conclusion, for any estimator $\hat{\eta}_{\bm{\lambda}}$, we have 
\begin{align*}
\cE_\cT^{\inn}(\hat{\eta}_{\bm{\lambda}}) \lesssim \cE_\cT(\hat{\eta}_{\bm{\lambda}}) +\Ocr
\end{align*}
and 
\begin{align*}
\cE_\cT(\hat{\eta}_{\bm{\lambda}}) \lesssim  \cE_\cT^{\inn}(\hat{\eta}_{\bm{\lambda}}) +\Ocr  
\end{align*}
with probability $1-n^{-11}$ under the event $\event$.
% By combining this with Lemma~\labelcref{lemma; in-sample MSE oracle}, we get the desired result directly.

\hfill \BlackBox

\section{Proof of Theorem~\labelcref{theorem; main theorem}}\label{section: proof main theorem}
\noindent
By Corollary~\labelcref{corollary; optimal MSE bound}, 
we have
\[
\min_{\bm\lambda \in {\bm{\Lambda}}} \cE_{\cT}(\hat{h}_{\bm{\lambda}})
\lesssim
\Bigl(\frac{BR}{n}\Bigr)^\alpha \|h^\star\|_\HH^{2(1-\alpha)} (\log n)^\alpha
+ \frac{BR}{n}M^2 \log n,
\]
which holds with probability at least \(1 - 2n^{-11}\).

Then, combining this with 
Proposition~\labelcref{proposition; oracle inequality MSE} and Lemma~\ref{lemma; in-sample MSE oracle}, we obtain
\[
\cE_\cT(\hat{h}_{\operatorname{final}}) 
\lesssim  
\Bigl(\frac{BR}{n}\Bigr)^\alpha \|h^\star\|_\HH^{2(1-\alpha)} (\log n)^\alpha
+ \frac{BR}{n}M^2 \log n 
+ \Ocr,
\]
with probability at least \(1 - 7n^{-11}\).
Finally, for \(n > 7\), Theorem~\labelcref{theorem; main theorem} follows.

\hfill \BlackBox

\section{Proof of Theorem~\labelcref{theorem; lower bound} (Lower Bound)}\label{section: lower bound proof}

{\color{black}
\begin{proof}
Let
\[
\cZ_+ = [0,1],
\qquad 
\cZ_- = [-1,0),
\qquad
\cZ = [-1,1].
\]
Let \(K_+(\cdot,\cdot)\) and \(K_-(\cdot,\cdot)\) be the reproducing kernels of \(H^k(\cZ_+)\) and \(H^k(\cZ_-)\), respectively, and define
\[
K(z,w) = \begin{cases}
K_+(z,w)&\quad \text{when \(z,w \in \cZ_+\)}, \\
K_-(z,w)&\quad \text{when \(z,w \in \cZ_-\)}, \\
0&\quad \text{otherwise}.
\end{cases}
\]
This kernel is symmetric and positive semidefinite.
Its RKHS \(\cF\) is the orthogonal direct sum of the two \(H^k\) components:
\begin{align}\label{equation: kernel for lower bound}
\cF
= \{&f = f_+(z)\one(z\in \cZ_+) + f_-(z)\one(z\in \cZ_-) 
\text{ where } f_+ \in H^k(\cZ_+),\ f_- \in H^k(\cZ_-)\},
\nonumber\\
\|f\|_\cF^2
&= \|f_+\|_{H^k(\cZ_+)}^2 + \|f_-\|_{H^k(\cZ_-)}^2 .
\end{align}
Under \(\cQ_\cT=\operatorname{Unif}(\cZ_+)\), the target covariance operator is supported on the one-dimensional block \(\cZ_+\), and it satisfies Assumption~\labelcref{Assumption; eigenvalue decay} with \(\ell=k\) by the standard eigenvalue decay of Sobolev kernels \citep{fischer2020sobolev}.

Now set
\[
\pi(z) \equiv \frac{1}{R},\qquad
\cQ_\cS = \frac{1}{B}\operatorname{Unif}(\cZ_+)
+ \Bigl(1-\frac{1}{B}\Bigr)\operatorname{Unif}(\cZ_-),
\qquad
\cQ_\cT = \operatorname{Unif}(\cZ_+).
\]
Since \(R\geq 2\) and \(B\geq 1\), this forms an \((R,B)\)-bounded instance: \(\pi(z)=1/R\leq 1-1/R\), and \(\de\cQ_\cT/\de\cQ_\cS=B\) on \(\cZ_+\) and \(0\) on \(\cZ_-\).

Define
\[
\cF_+(w)=\{g\in H^k(\cZ_+):\|g\|_{H^k(\cZ_+)}\leq w\},
\qquad
\cF_-(w)=\{g\in H^k(\cZ_-):\|g\|_{H^k(\cZ_-)}\leq w\}.
\]
We consider the subspace
\[
\cF_{\operatorname{sub}}(W)
=\Bigl\{h=h_+(z)\one(z\in \cZ_+) + h_-(z)\one(z\in \cZ_-):
h_+\in\cF_+(W/\sqrt{2}),\ h_-\in\cF_-(W/\sqrt{2})\Bigr\}.
\]
By \eqref{equation: kernel for lower bound}, \(\cF_{\operatorname{sub}}(W)\subseteq\{h:\|h\|_\cF\leq W\}\).
We further restrict to the submodel in which \(f_0^\star\equiv 0\) is known to the learner and \(f_1^\star=h^\star\), where
\[
h^\star\in \cF_{\operatorname{sub}}(W).
\]
This submodel is contained in \(\{h^\star:\|h^\star\|_\cF\leq W\}\), and revealing \(f_0^\star\) can only make the estimation problem easier. 
Therefore, any lower bound for this submodel is also a lower bound for the original problem.

In this submodel, the only observations that contain information about the target-side component \(h_+^\star\) are the treated samples with covariates in \(\cZ_+\). 
Let
\[
N_+ = \bigl|\{i:z_i\in \cZ_+,\ a_i=1\}\bigr|.
\]
Then \(N_+\sim \operatorname{Binomial}(n,1/(BR))\), and conditional on \(N_+\), the informative covariates are independent draws from \(\operatorname{Unif}(\cZ_+)\). 
By Markov's inequality, \(\PP(N_+\leq 2n/(BR))\geq 1/2\). 
On this event, the experiment is no more informative than the corresponding one-dimensional nonparametric regression experiment over the \(H^k([0,1])\) ball with radius \(W/\sqrt{2}\) and \(2n/(BR)\) samples.
Applying the standard minimax lower bound for this integer-order class \citep[Theorem~2.10]{tsybakov2009introduction} yields
\begin{align*}
\inf_{\hat{h}} \sup_{\substack{(f_0^\star, f_1^\star, h^\star):\\ \|h^\star\|_\cF \leq W}}
\EE[\cE_\cT(\hat{h})]
&\gtrsim
\Bigl(\frac{2n}{BR}\Bigr)^{-\frac{2k}{1+2k}}
W^{\frac{2}{1+2k}} \\
&\gtrsim
\Bigl(\frac{BR}{n}\Bigr)^{\alpha}W^{2(1-\alpha)},
\end{align*}
where \(\alpha=\frac{2k}{1+2k}\).

\end{proof}
}

\section{Proofs for Examples in Section~\labelcref{subsection: weak overlap}}
\label{section: proof weak overlap exampls}

\subsection{Proof of Example~\labelcref{example; positivity}}
\noindent
\begin{proof}
Let the probability distribution function of the source covariates \(z_i\) be \(F_{\cQ_\cS}(\cdot)\).
By the definition of \(\Sigmatreated\) and \(\Sigmacontrol\), we have
\begin{align}\label{equation: integral form of moments}
\Sigmatreated 
&= \int_{\cZ} \bigl(\phi(z) \otimes \phi(z)\bigr) \pi(z)\mathrm{d}F_{\cQ_\cS}(z), 
\quad 
\Sigmacontrol  
= \int_{\cZ} \bigl(\phi(z) \otimes \phi(z)\bigr) \bigl(1-\pi(z)\bigr)\mathrm{d}F_{\cQ_\cS}(z).
\end{align}
Hence, when \(\kappa \leq \pi(z) \leq 1-\kappa\), the following holds:
\begin{align*}
\Sigmatreated \preceq \frac{1}{\kappa} \Sigmacontrol 
\quad\text{and}\quad    
\Sigmacontrol  \preceq \frac{1}{\kappa} \Sigmatreated.
\end{align*}
\end{proof}

\subsection{Proof of Example~\labelcref{example; singular propensity 1}}
\noindent
\begin{proof}
Set \(\xi \geq 1\). 
We prove this for the general \(H^k([0,1])\), where \(k \in \NN\).
It remains to find \(R\) such that, for all \(g \in H^k([0,1])\),
\begin{align*}
\theta(g)^\top \Sigmatreated \theta(g) &\leq R \theta(g)^\top \Sigmacontrol \theta(g) 
+ \frac{R}{n}\|g\|_{H^k([0,1])}^2, 
\\
\theta(g)^\top \Sigmacontrol \theta(g) &\leq R \theta(g)^\top \Sigmatreated \theta(g) 
+ \frac{R}{n}\|g\|_{H^k([0,1])}^2,
\end{align*}
where \(\theta(g)\) is the Hilbertian element of \(g\), as defined in Appendix~\labelcref{section: good events and second moments}.
First, we prove the case when the source distribution is uniform on \([0,1]\).
When the density is bounded above and below by constants, the argument extends directly to more general cases.

These conditions are equivalent to 
\begin{align*}
&\int_0^1 (1-z)g(z)^2 \de z 
\leq R \int_0^1 zg(z)^2 \de z 
+ \frac{R}{n}\|g\|_{H^k([0,1])}^2,  
\\
&\int_0^1 zg(z)^2 \de z 
\leq R \int_0^1 (1-z)g(z)^2 \de z  
+ \frac{R}{n}\|g\|_{H^k([0,1])}^2.
\end{align*}
The second inequality is equivalent to
\begin{align*}
&\int_0^1  z g(z)^2 \de z 
\leq R \int_0^1 (1-z)g(z)^2 \de z  
+ \frac{R}{n}\|g\|_{H^k}^2 
\\
&\Leftrightarrow   
\int_0^1  \bigl(1-u\bigr) g(1-u)^2 \de u 
\leq R \int_0^1 ug(1-u)^2 \de u  
+ \frac{R}{n}\|\tilde{g}\|_{H^k}^2  \quad \bigl(\text{by setting } g(1-u) = \tilde{g}(u)\bigr).
\\
&\Leftrightarrow  
\int_0^1  \bigl(1-u\bigr)\tilde{g}(u)^2 \de u 
\leq R \int_0^1 u\tilde{g}(u)^2 \de u  
+ \frac{R}{n}\|\tilde{g}\|_{H^k}^2.
\end{align*}
Hence, it suffices to find \(R\) that satisfies, for all \(g\) with \(\|g\|_{H^k}=1\),
\begin{align*}
\int_0^1 g(z)^2 \de z 
\leq (R+1)\int_0^1 zg(z)^2 \de z  
+ \frac{R}{n}.
\end{align*}

By Claim~\labelcref{claim; ex 3} below, the inequalities hold for all \(R\) satisfying
\begin{align*}
R +1 \geq \frac{c}{r}, 
\quad  
\frac{R}{n} \geq cr^{2k}
\end{align*}
for some \(0< r <\frac{1}{2}\) and some constant \(c>0\).
Choosing \(r \asymp n^{-\frac{1}{2k+1}}\) shows that Assumption~\labelcref{assumption; weak treatment overlap} holds with \(R \asymp n^{\frac{1}{2k+1}}\).    
\end{proof}

\begin{claim}\label{claim; ex 3}
For any \(g\) with \(\|g\|_{H^k([0,1])} =1\) and any \(0<r < \frac{1}{2}\),
\begin{align*}
\frac{c}{r} \int_0^1 zg^2(z) \de z  
+  cr^{2k}
\geq \int_0^1 g^2(z)\de z
\end{align*}
holds for some constant \(c>0\) that depends only on \(k\).
\end{claim}

\begin{proof}
Let \(c_1\) be the constant from Lemma~\labelcref{lemma; polynomial integration} and set \(c = c_1 +1\). 
Then,
\begin{align*}
\frac{c}{r} \int_0^1 zg^2(z) \de z  +  cr^{2k} 
&\geq \Bigl(c_1  \frac{1}{r}+\frac{1}{r}\Bigr) 
\Bigl(\int_0^r zg^2(z)\de z + \int_r^1   zg^2(z)\de z \Bigr) 
+  c_1r^{2k}
\\
&\geq \int_r^1 g^2(z)\de z  +c_1 \int_r^1 g^2(z)\de z 
+  c_1r^{2k}  
\\
&\geq \int_0^1  g^2(z)\de z 
\quad \text{(by Lemma~\labelcref{lemma; polynomial integration})}.
\end{align*}
\end{proof}

\begin{lemma}\label{lemma; polynomial integration}
For any \(0<r<\frac{1}{2}\) and \(g \in H^k([0,1])\) with \(\|g\|_{H^k([0,1])}=1\), we have 
\begin{align*}
\int_{0}^r g(z)^2 \de z \leq  c_1 \int_r^1 g(z)^2 \de z +  c_1r^{2k}
\end{align*}
for some constant \(c_1\) that depends only on \(k\).
\end{lemma}

\begin{proof}
By Sobolev embedding and Morrey's theorem \citep{evans2022partial}, there exist a polynomial \(p\) of degree \(k-1\) and a function \(\varepsilon\) such that for all \(0 < t < 2r\),
\begin{align*}
g(t) = p(t) + \varepsilon(t),
\end{align*}
where \(|\varepsilon(t)| \leq c_2 |t-r|^{k-\frac{1}{2}}\) for some constant \(c_2>0\) that depends only on \(k\). 
(Indeed, it is \(\frac{1}{(k-1)!}\) times a constant from Morrey's theorem.)

\paragraph*{Goal.} 
We want to show
\begin{align*}
\|g \|_{L^2([0,r])}^2 
\leq c_1 \|g \|_{L^2([r,2r])}^2 
+c_1r^{2k}.
\end{align*} 
First, observe that
\begin{align}
\| g\|^2_{L^2([0,r])} =\| p + \varepsilon\|^2_{L^2([0,r])}  
&\leq 2\| p\|_{L_2([0,r])}^2 + 2\|  \varepsilon\|^2_{L^2([0,r])} 
\nonumber \\
&\leq 2\| p\|_{L_2([0,r])}^2 +2 rc_2^2 r^{2k-1}. 
\label{equation: singular propensity 1}
\end{align}
Next, observe that
\begin{align*}
\| g\|^2_{L^2([r,2r])} 
=\|p+\varepsilon \|^2_{L^2([r,2r])} 
&= \|p\|^2_{L^2([r,2r])} 
+ \|\varepsilon\|^2_{L^2([r,2r])} 
+2\langle p, \varepsilon \rangle_{L^2([r,2r])}
\\
&\geq \|p\|^2_{L^2([r,2r])} 
+ \|\varepsilon\|_{L^2([r,2r])}^2  
- 2\|p\|_{L^2([r,2r])}\|\varepsilon\|_{L^2([r,2r])}.
\end{align*}
Hence,
\begin{align*}
&\|g \|^2_{L^2([r,2r])} 
+ \|\varepsilon \|^2_{L^2([r,2r])} 
\\
&\geq \frac{1}{2}\|p\|^2_{L^2([r,2r])}  
+ \frac{1}{2}\|p\|^2_{L^2([r,2r])} 
+2\|\varepsilon \|^2_{L^2([r,2r])} 
- 2\|p\|_{L^2([r,2r])}\|\varepsilon\|_{L^2([r,2r])}
\\
&\geq \frac{1}{2}\|p\|^2_{L^2([r,2r])}.
\end{align*}
Set \(\tilde{c} =4\max\bigl(c_2^2, c_3^k\bigr)\), where \(c_3\) is the constant from Lemma~\labelcref{lemma; polynomial integral ratio}.
Combining this with Lemma~\labelcref{lemma; polynomial integral ratio}, we obtain
\begin{align*}
&\tilde{c} 
\bigl(\|g \|^2_{L^2([r,2r])} 
+ \|\varepsilon \|^2_{L^2([r,2r])}\bigr) 
+ \tilde{c}r^{2k} 
\\
&\geq 4c_3^k 
\bigl(\|g \|^2_{L^2([r,2r])} 
+ \|\varepsilon \|^2_{L^2([r,2r])}\bigr) 
+ 4c_2^2r^{2k} 
\\
&\geq 2c_3^k \|p\|^2_{L^2([r,2r])}
+ 2rc_2^2r^{2k-1}
\\
&\stackrel{\text{(i)}}{\geq} 2\|p\|^2_{L^2([0,r])} 
+2rc_2^2r^{2k-1} 
\\
&\stackrel{\text{(ii)}}{\geq}  \| p + \varepsilon\|^2_{L^2([0,r])},
\end{align*}
where (i) follows from Lemma~\labelcref{lemma; polynomial integral ratio} and (ii) follows from \eqref{equation: singular propensity 1}.
Thus,
\begin{align*}
\tilde{c}  \|g\|^2_{L^2([r,2r])}  
+ (\tilde{c} c_2^2 + \tilde{c})r^{2k} 
\geq \| g\|^2_{L^2([0,r])}.
\end{align*}
Setting \(c_1 := \tilde{c}c_2^2 + \tilde{c}\) completes the proof.
\end{proof}

\begin{lemma}\label{lemma; polynomial integral ratio}
For any polynomial \(p(x)\) of degree \(\beta\) and any \(0<r<1\), we have
\begin{align*}
\int_{-r}^0 p(x)^2 \de x 
\leq c_3^\beta \int_0^r p(x)^2 \de x
\end{align*}
for some absolute constant \(c_3>0\).
\end{lemma}

\begin{proof}
Let \(p(x)= a_0 + a_1 x + \dots + a_\beta x^\beta\).
Denote the maximum and minimum eigenvalues of 
\(\EE_{x \sim \operatorname{Unif}(0,1)}\bigl[(1,x,\dots,x^\beta)(1,x,\dots,x^\beta)^\top\bigr]\) 
by \(M_\beta\) and \(m_\beta\), respectively. Then,
\begin{align*}
\int_0^r p(x)^2 \de x 
&= r \int_0^1 p(rx)^2 \de x 
\geq  rm_\beta 
\|(a_0, ra_1, \dots, r^\beta a_\beta) \|_2^2.
\end{align*}
Next,
\begin{align*}
\int_{-r}^0 p(x)^2 \de x 
&= r \int_{-1}^0 p(rx)^2 \de x 
= r \int_0^1 p(-rx)^2 \de x 
\\
&\leq  rM_\beta 
\bigl\|(a_0, -ra_1, r^2 a_2, \dots, r^\beta a_\beta(-1)^\beta)\bigr\|_2^2 
\\
&= rM_\beta 
\|(a_0, ra_1, r^2 a_2, \dots, r^\beta a_\beta)\|_2^2.
\end{align*}
By Lemma~\labelcref{lemma; key constant LPR}, \(\frac{M_\beta}{m_\beta} \leq c_3^\beta\) for some absolute constant \(c_3>0\). 
Hence the result follows immediately.
\end{proof}

% \newpage
% \clearpage
\section{Detailed Results for ATE (Remark~\ref{remark; ATE})}\label{section: apdx ATE}

\textbf{Algorithm. } Identical to Algorithm~\ref{algorithm: model selection}, except we replace \(L(h)\) with \cref{equation: model selection ATE}.

\textbf{Assumptions.}
We require the same set of assumptions as in Theorem~\ref{theorem; main theorem}, with one relaxation: the bounded-kernel condition in Assumption~\ref{Assumption; eigenvalue decay} is still needed, but the polynomial eigenvalue decay requirement is not.

\textbf{Results. } 
We now state our main theorem for the ATE and then present its proof.

\begin{theorem}[Upper rates for ATE]\label{theorem; ATE}
By invoking Algorithm~\ref{algorithm: main} (using the inputs from Theorem~\ref{theorem; main theorem}) and replacing the objective \(L(h)\) in the subroutine (Algorithm~\ref{algorithm: model selection}) with the expression below, we obtain an estimator for the ATE:
\begin{align}
L(h) =  \bigg|\frac{1}{n_\cT}\sum_{i=1}^{n_\cT} \tilde h(z_{0i}) - \frac{1}{n_\cT} \sum_{i=1}^{n_\cT} h(z_{0i}) \bigg|^2.
\end{align}
With the resulting \(\hat h_{\operatorname{final}}\), the ATE estimator is
\begin{align*}
\hat\Delta_\cT := \frac{1}{n_\cT} \sum_{i=1}^{n_\cT} \hat h_{\operatorname{final}}(z_{0i}).
\end{align*}
When \(n_\cT \gtrsim n/B\), under Assumptions~\labelcref{assumption; boundedness,assumption; consistency and unconfoundedness,assumption; overlap source target,assumption; subGaussian noise,assumption; weak treatment overlap,Assumption; eigenvalue decay} (except that polynomial eigenvalue decay is not required), the final estimator satisfies, with probability at least \(1-n^{-10}\),
\begin{align*}
|\hat \Delta_\cT -\Delta^\star_\cT| \lesssim \frac{M}{\sqrt{n_{\operatorname{eff}}}}
\end{align*}
where \(\lesssim\) hides logarithmic factors and universal constants.
Moreover, this rate is minimax-optimal up to logarithmic factors and absolute constants.
\end{theorem}

Notably, the eigenvalue decay requirement in Assumption~\ref{Assumption; eigenvalue decay} is not needed. In other words, as long as \(\cF\) is an RKHS with a bounded kernel, we achieve \(\sqrt{n}\)-consistency.

\begin{proof}
Define 
\begin{align*}
\bar \phi_\cT := \EE[\phi(z_{0i})], \qquad \hat {\bar \phi}_\cT := \frac{1}{n_\cT} \sum_{i=1}^{n_\cT} \phi(z_{0i}).
\end{align*}
First define the in-sample ATE as 
\begin{align*}
{\Delta}_{\cT}^{\operatorname{in}} := \frac{1}{n_\cT} \sum_{i =1}^{n_\cT} h^\star(z_{0i}) = \hat{\bar{\phi}}_\cT^\top \eta^\star.
\end{align*}
By Hoeffding's inequality, with probability at least \(1-n^{-11}\), \({\Delta}_{\cT}^{\operatorname{in}}\) concentrates around \(\Delta_\cT^\star\) as
\[
\bigl|{\Delta}_{\cT}^{\operatorname{in}} -\Delta_\cT^\star\bigr| \lesssim \frac{M\sqrt{\log n}}{\sqrt{n_\cT}}.
\]
Recall the notation and setup defined in Appendix~\ref{section: groundwork}.
We define the squared error with respect to the in-sample ATE as 
\begin{align*}
\cE_{\operatorname{ATE}}(\hat\eta_\lambda ) := (\hat\eta_\lambda -\eta^\star)^\top \hat{\bar{\phi}}_\cT \hat{\bar{\phi}}_\cT^\top (\hat\eta_\lambda -\eta^\star).
\end{align*}
We proceed in two steps and work throughout on the good event $\Ecr$ defined in Appendix~\ref{section: groundwork}.

\textbf{Step 1. Squared-error bound for \(\hat h_\lambda\). } 
Since the squared error of the in-sample ATE can be written in our MSE bound with respect to the second moment \(\bSigma_\cT'\), where
\begin{align*}
\bSigma_\cT' = \hat{\bar \phi}_\cT \hat{\bar \phi}_\cT^\top,
\end{align*}
let \(\check\lambda\) denote the smallest \(\lambda \in \Lambda\).
Observe that 
\begin{align*}
\bigl| \hat{\bar \phi}_\cT^\top (\hat \eta_{\operatorname{final}} -\tilde \eta)\bigr| \;\le\; \bigl| \hat{\bar \phi}_\cT^\top (\hat \eta_{\check\lambda} -\tilde \eta)\bigr|  
\;\le\; \bigl| \hat{\bar \phi}_\cT^\top (\hat \eta_{\check\lambda} -\eta^\star)\bigr|  + \bigl| \hat{\bar \phi}_\cT^\top (\tilde \eta  -\eta^\star)\bigr|.
\end{align*}
By the Cauchy–Schwarz inequality, for some absolute constants $c_1, c_2>0$, we have
\begin{align*}
\bSigma_\cT' 
= \hat{\bar \phi}_\cT \hat{\bar \phi}_\cT^\top 
\;\preceq\; \hat \bSigma_\cT 
\;\preceq \; c_1(\bSigma_\cT + \frac{\log n}{n_\cT} \Ib)
\;\preceq\; c_2B \log n\Bigl(\bSigma_\cS + \frac{1}{n} \Ib\Bigr).
\end{align*}
Thus, \(\bSigma_\cT'\) satisfies the source--target overlap condition with degree \(B\log n\).
Moreover, \(\bSigma_\cT' = \hat{\bar \phi}_\cT \hat{\bar \phi}_\cT^\top\) is rank~1, i.e., it exhibits finite eigenvalue decay of dimension~1.
Recall that $\check\lambda \asymp \frac{\log n}{n}$.
Applying Theorem~\ref{theorem; MSE bound RA learner} with target second moment \(\bSigma_\cT'\) yields, with probability at least \(1-4n^{-11}\),
\begin{align*}
\bigl| \hat{\bar \phi}_\cT^\top (\hat \eta_{\check\lambda} -\eta^\star)\bigr|^2  &\lesssim  \frac{M^2}{n_{\operatorname{eff}}}.
\end{align*}
Here, the trace term is evaluated with \(\bSigma_\cT\) in \(\Sbar_\lambda\) replaced by \(\bSigma_\cT^\prime\). Since \(\bSigma_\cT^\prime\) is rank~1, \(\Sbar_\lambda\) is also rank~1 and hence its trace equals its only nonzero eigenvalue:
\begin{align*}
\Tr(\Sbar_\lambda)
&= \|\Sbar_\lambda\|_{\op} \\
&\leq cB\log n
\Bigl\|(\bSigma_\cS+\lambda\Ib)^{-\frac12}
\Bigl(\bSigma_\cS+\frac{1}{n}\Ib\Bigr)
(\bSigma_\cS+\lambda\Ib)^{-\frac12}\Bigr\|_{\op}
\lesssim B\log n,
\end{align*}
where the last inequality uses \(\lambda=\check\lambda\asymp \log n/n\).

For the imputation estimators, we additionally combine the source--target display above with Lemma~\labelcref{lemma; second moment relation under weak overlap}. This gives
\begin{align*}
\bSigma_\cT' 
\;\preceq\; c_3BR \log n\Bigl(\Sigmatreated + \frac{\xi\log n}{n} \Ib\Bigr),
\qquad
\bSigma_\cT' 
\;\preceq\; c_3BR \log n\Bigl(\Sigmacontrol + \frac{\xi\log n}{n} \Ib\Bigr),
\end{align*}
so the treated and control nuisance KRR problems have a rank-1 target second moment and overlap degree \(BR\log n\), up to logarithmic factors.
Furthermore, applying the standard covariate-shift KRR bounds used in \citet{wang2026pseudo,ma2023optimally} to these rank-1 target problems with overlap degree \(BR\log n\), we obtain the following with probability at least $1-n^{-11}$:
\begin{align*}
\bigl| \hat{\bar \phi}_\cT^\top (\tilde \eta  -\eta^\star)\bigr|^2 &\lesssim   \bigl| \hat{\bar \phi}_\cT^\top (\tilde \theta_1  -\theta_1^\star)\bigr|^2 +\bigl| \hat{\bar \phi}_\cT^\top (\tilde \theta_0  -\theta_0^\star)\bigr|^2 \lesssim  \frac{M^2}{n_{\operatorname{eff}}}.
\end{align*}
where $\lesssim$ hides constants and logarithmic factors.
Hence, we obtain the upper rate stated in the theorem.

\textbf{Step 2. Lower bound for ATE. }
Consider the same kernel as in the proof of Theorem~\ref{theorem; lower bound}.
Also consider the same setup: 
\begin{itemize}
\item  \(\pi(z) \equiv \frac{1}{R}\).
\item \(\cQ_\cS = \begin{cases}
\operatorname{Unif}([0,1]) &\quad \text{with probability \(\frac{1}{B}\)}, \\
\operatorname{Unif}([-1,0]) &\quad \text{with probability \(1- \frac{1}{B}\)}.
\end{cases}\)
\item \(\cQ_\cT = \operatorname{Unif}([0,1])\).
\item \(f_0^\star \equiv C_0 \;\;(\text{a constant})\).
\end{itemize}

This is equivalent to estimating the mean of \(f^\star_1(z)\) for \(z \sim \operatorname{Unif}([0,1])\) with \(\Theta \bigl(\tfrac{n}{BR}\bigr)\) samples.
Therefore, we obtain 
\begin{align*}
\inf_{\hat{\Delta}_\cT} \sup_{\substack{(f_0^\star, f_1^\star, h^\star): \|f_0^\star\|_\cF, \|f_1^\star\|_\cF \leq 1}} 
\EE\bigl[|\hat \Delta_\cT -\Delta^\star_\cT|^2] 
\;\gtrsim\; \frac{1}{n_{\operatorname{eff}}}.
\end{align*}

\end{proof}

\section{Proofs for Cross Fitting Algorithm}\label{subsection: proofs for cross fitting}
\noindent
With a slight modification of our analysis, the cross-fitting algorithm achieves the same rate.
We briefly provide a proof sketch for its performance.
Let \(\hat{\eta}_{\final}^{(1)}, \hat{\eta}_{\final}^{(2)}, \hat{\eta}_{\final}^{(3)}\) denote the final estimators for three dataset permutations.
Define the averaged estimator as
\[
\hat{\eta}_{\final} := \frac{1}{3} \bigl(\hat{\eta}_{\final}^{(1)} 
+\hat{\eta}_{\final}^{(2)} +\hat{\eta}_{\final}^{(3)}\bigr).
\]

By Theorem~\labelcref{theorem; main theorem}, with probability at least \(1-n^{-10}\), each estimator \(\hat{\eta}_{\operatorname{final}}^{(j)}\) for \(j =1,2,3\) achieves the same MSE bound as in Theorem~\labelcref{theorem; main theorem}.
Since
\begin{align*}
\cE_\cT(\hat{\eta}_{\final})
&= \|\hat{\eta}_{\final} -\eta^\star \|_{\bSigma_\cT}^2 \\
&\leq \left(\frac{1}{3} \sum_{j=1}^3 \|\hat{\eta}_{\final}^{(j)} -\eta^\star \|_{\bSigma_\cT}\right)^2 
\leq \max_{j \in \{1,2,3\}}\cE_\cT(\hat{\eta}_{\final}^{(j)}),
\end{align*}
with probability at least \(1-3n^{-10}\), we have 
\begin{align*}
\max_{j \in \{1,2,3\}}\cE_\cT(\hat{\eta}_{\final}^{(j)} )
&\lesssim  n_{\operatorname{eff}}^{-\alpha} \norm{h^\star}^{2(1-\alpha)}_\cF 
+ M^2 \Bigl(\frac{1}{n_{\operatorname{eff}}}+ \frac{R}{n_\cT}\Bigr),
\end{align*}
which is the same rate as in Theorem~\labelcref{theorem; main theorem}.

\hfill \BlackBox

\section{Technical Lemmas}

\subsection{Lemmas for Bias-variance Trade-off}
\noindent
We next present key lemmas for the optimal choice of \(\lambda\) under the polynomial eigenvalue decay condition. 
Under Assumption~\labelcref{Assumption; eigenvalue decay}, we define the effective dimension \(\db(\lambda) = \inf \{j \mid \mu_j < \lambda \}\), and the following inequality
\begin{align*}
\frac{\sum_{\mu_j < \lambda} \mu_j}{\lambda} \leq c \db(\lambda)
\end{align*}
holds for some constant \(c > 0\). Moreover, \(\db(\lambda) \lesssim \lambda^{-\frac{1}{2\ell}}\).

\begin{lemma}[Performance of optimal regularizer]\label{lemma; optimal trade-off lambda}
Let \(h > 0\) be an arbitrary positive constant and let \(\eta \in \HH\) with \(\lambda_{1} = h^{\alpha}B^{-(1-\alpha)}\|\eta\|_{\HH}^{-2\alpha} \ge \frac{\xi}{n}\). Then the following holds:
\begin{align*}
B \lambda_{1} \norm{\eta}^2_\HH + h \Tr(\Sbar_{\lambda_{1}}) \lesssim \norm{\eta}^{2(1-\alpha)}_\HH\, h^{\alpha} B^{\alpha}.
\end{align*}
\end{lemma}

\begin{proof}
Under Assumption~\ref{assumption; overlap source target}, we have
\begin{align*}
2(B \bSigma_\cS + B \lambda_{1} \Ib) \succeq \bSigma_\cT + B \lambda_{1} \Ib,
\end{align*}
and hence
\begin{align*}
\bSigma_\cS + \lambda_{1} \Ib \succeq \frac{1}{2}\Bigl(\frac{1}{B} \bSigma_\cT + \lambda_{1} \Ib\Bigr).
\end{align*}
Therefore, we have
\begin{align*}
\Tr( \Sbar_{\lambda_{1}} ) &= \Tr((\bSigma_\cS + \lambda_{1} \Ib)^{-1} \bSigma_\cT)\\
&\lesssim \Tr\Bigl(\bSigma_\cT \Bigl(\frac{1}{B} \bSigma_\cT + \lambda_{1} \Ib\Bigr)^{-1}\Bigr)\\
&\leq \sum_{j=1}^{\infty} \frac{B \mu_j}{\mu_j + B \lambda_{1}}.
\end{align*}  
Then the left-hand side of the statement is bounded as follows:
\begin{align*}
B \lambda_{1} \norm{\eta}_\HH^2 + h \sum_{j=1}^{\infty} \frac{B \mu_j}{\mu_j + B \lambda_{1}} 
&\lesssim B \lambda_{1} \norm{\eta}_\HH^2 + h \sum_{j=1}^{\db(\lambda_{1} B)} \frac{B \mu_j}{\mu_j + B \lambda_{1}}  +h \sum_{j=\db(\lambda_{1} B)+1}^{\infty} \frac{B \mu_j}{\mu_j + B \lambda_{1}} 
\\
&\lesssim B\Bigl(\lambda_{1} \norm{\eta}^2_\HH + h\, \db(\lambda_{1} B)\Bigr)\\
&\lesssim B\Bigl(\lambda_{1} \norm{\eta}^2_\HH + h\, (\lambda_{1} B)^{-\frac{1}{2\ell}}\Bigr)\\
&\lesssim B\Bigl(\norm{\eta}^2_\HH\, \lambda_{1} + h\, B^{-\frac{1}{2\ell}} \lambda_{1}^{-\frac{1}{2\ell}}\Bigr)\\
&\lesssim B\Bigl(\norm{\eta}^2_\HH\Bigr)^{\frac{1}{2\ell+1}} h^{\frac{2\ell}{2\ell+1}} B^{-\frac{1}{2\ell+1}} 
\quad \text{($\because$ evaluate at \(\lambda_{1} = h^{\alpha}B^{-(1-\alpha)}\|\eta\|_{\HH}^{-2\alpha}\))}\\
&= \Bigl(\norm{\eta}^2_\HH\Bigr)^{\frac{1}{2\ell+1}} h^{\frac{2\ell}{2\ell+1}} B^{\frac{2\ell}{2\ell+1}}\\
&= \Bigl(\norm{\eta}^2_\HH\Bigr)^{1-\alpha} h^{\alpha} B^{\alpha},
\end{align*}
where we used \(\db(\lambda_{1}) \lesssim \lambda_{1}^{-\frac{1}{2\ell}}\) under Assumption~\labelcref{Assumption; eigenvalue decay}.
\end{proof}

\begin{corollary}[Performance of optimal regularizer in grid]\label{corollary; optimal MSE in grid}
Under the same setup as in Lemma~\labelcref{lemma; optimal trade-off lambda}, set \(\lambda^\star = h^{\alpha}B^{-(1-\alpha)}\|\eta\|_{\HH}^{-2\alpha}\). Then, for any \(\lambda > 0\) with \(\lambda^\star \leq \lambda \leq 2\lambda^\star\),
\begin{align*}
B \lambda \norm{\eta}^2_\HH + h \Tr(\Sbar_{\lambda}) \lesssim \Bigl(\norm{\eta}_\HH^2\Bigr)^{1-\alpha} h^{\alpha} B^{\alpha}.
\end{align*}
\end{corollary}

\begin{proof}
Using Lemma~\ref{lemma; optimal trade-off lambda}, observe that 
\begin{align*}
B \lambda \norm{\eta}^2_\HH + h \Tr(\Sbar_{\lambda})
&\leq B\Bigl( \lambda \norm{\eta}^2_\HH + h\, (\lambda B)^{-\frac{1}{2\ell}}\Bigr) \\
&\leq B\Bigl(2\, \lambda^\star \norm{\eta}^2_\HH + h\, (\lambda^\star B)^{-\frac{1}{2\ell}}\Bigr) \\
&\leq 2B\Bigl(\lambda^\star \norm{\eta}^2_\HH + h\, (\lambda^\star B)^{-\frac{1}{2\ell}}\Bigr) \\
&\lesssim \Bigl(\norm{\eta}_\HH^2\Bigr)^{1-\alpha} h^{\alpha} B^{\alpha}.
\end{align*}
\end{proof}

\subsection{Concentration Inequalities}
\noindent
We first recall the trace-class lemmas from \citet*{wang2026pseudo}.
\begin{lemma}[Lemma D.1 from \citealt{wang2026pseudo}]\label{lemma; quadratic martingale}
Suppose that $\boldsymbol{x} \in \mathbb{R}^d$ is a zero-mean random vector with $\|\boldsymbol{x}\|_{\psi_2} \leq 1$. There exists a universal constant $C>0$ such that for any symmetric and positive semi-definite matrix $\boldsymbol{\Sigma} \in \mathbb{R}^{d \times d}$,
\begin{align*}
\mathbb{P}\left(\boldsymbol{x}^{\top} \boldsymbol{\bSigma} \boldsymbol{x} \leq C \operatorname{Tr}(\boldsymbol{\bSigma}) t\right) \geq 1-e^{-r(\boldsymbol{\bSigma}) t}, \quad \forall t \geq 1 .
\end{align*}
Here $r(\boldsymbol{\bSigma})=\operatorname{Tr}(\boldsymbol{\bSigma}) /\|\boldsymbol{\bSigma}\|_2$ is the effective rank of $\boldsymbol{\bSigma}$.
\end{lemma}
\begin{lemma}[Corollary D.1 from \citealt{wang2026pseudo}]\label{lemma; trace class concentration bounded}
Let $\left\{\boldsymbol{x}_i\right\}_{i=1}^n$ be i.i.d. random elements in a separable Hilbert space $\mathbb{H}$ with $\boldsymbol{\bSigma}=$ $\mathbb{E}\left(\boldsymbol{x}_i \otimes \boldsymbol{x}_i\right)$ being trace class. Define $\hat{\boldsymbol{\bSigma}}=\frac{1}{n} \sum_{i=1}^n \boldsymbol{x}_i \otimes \boldsymbol{x}_i$. Choose any constant $\gamma \in(0,1)$ and define an event $\mathcal{A}=\{(1-\gamma)(\boldsymbol{\bSigma}+\lambda \boldsymbol{I}) \preceq \hat{\boldsymbol{\bSigma}}+\lambda \boldsymbol{I} \preceq(1+\gamma)(\boldsymbol{\bSigma}+\lambda \boldsymbol{I})\}$.
If $\left\|\boldsymbol{x}_i\right\|_{\mathbb{H}} \leq L_x$ holds almost surely for some constant $L_x$, then there exists a constant $C \geq 1$ determined by $\gamma$ such that $\mathbb{P}(\mathcal{A}) \geq 1-\delta$ holds so long as $\delta \in(0,1 / 14]$ and $\lambda \geq \frac{C L_x^2 \log (n / \delta)}{n}$.    
\end{lemma}
In our applications, the bounded-kernel assumption gives \(\|x_i\|_{\HH}^2=K(z_i,z_i)\leq \xi\), so this lemma is invoked with \(L_x^2\leq \xi\).

\subsection{Lemmas for Model Selection}

\begin{lemma}[Theorem 5.2 from \citealt{wang2026pseudo}]\label{lemma; loss model selection}
Let $\left\{\boldsymbol{z}_i\right\}_{i=1}^n$ be deterministic elements in a set $\mathcal{Z} ; g^{\star}$ and $\left\{g_j\right\}_{j=1}^m$ be deterministic functions in $\mathcal{Z} ; \widetilde{g}$ be a random function on $\mathcal{Z}$. Define
\begin{align*}
\mathcal{L}(g)=\frac{1}{n} \sum_{i=1}^n\left|g\left(\boldsymbol{z}_i\right)-g^{\star}\left(\boldsymbol{z}_i\right)\right|^2
\end{align*}
for any function $g$ on $\mathcal{Z}$. Assume that the random vector $\widetilde{\boldsymbol{y}}=\left(\widetilde{g}\left(\boldsymbol{z}_1\right), \widetilde{g}\left(\boldsymbol{z}_2\right), \cdots, \widetilde{g}\left(\boldsymbol{z}_n\right)\right)^{\top}$ satisfies $\|\widetilde{\boldsymbol{y}}-\mathbb{E} \widetilde{\boldsymbol{y}}\|_{\psi_2} \leq V<\infty$. Choose any
\begin{align*}
\widehat{j} \in \underset{j \in[m]}{\operatorname{argmin}}\left\{\frac{1}{n} \sum_{i=1}^n\left|g_j\left(\boldsymbol{z}_i\right)-\widetilde{g}\left(\boldsymbol{z}_i\right)\right|^2\right\} .
\end{align*}

There exists a universal constant $C$ such that for any $\delta \in(0,1]$, with probability at least $1-\delta$ we have
\begin{align*}
\mathcal{L}\left(g_{\hat{j}}\right) \leq \inf _{\gamma>0}\left\{(1+\gamma) \min _{j \in[m]} \mathcal{L}\left(g_j\right)+C\left(1+\gamma^{-1}\right)\left(\mathcal{L}(\mathbb{E} \widetilde{g})+\frac{V^2 \log (m / \delta)}{n}\right)\right\} .
\end{align*}

Consequently,
\begin{align*}
\mathbb{E} \mathcal{L}\left(g_{\hat{j}}\right) \leq \inf _{\gamma>0}\left\{(1+\gamma) \min _{j \in[m]} \mathcal{L}\left(g_j\right)+C\left(1+\gamma^{-1}\right)\left(\mathcal{L}(\mathbb{E} \widetilde{g})+\frac{V^2(1+\log m)}{n}\right)\right\}.
\end{align*}    
\end{lemma}

\begin{lemma}[Lemma D.5 from \citealt{wang2026pseudo}]\label{lemma; key lemma in-sample population}
Let \(\{z_i\}_{i=1}^n\) be i.i.d. samples in a space \(\cZ\).
For any \(g:\cZ\to\RR\), define
\[
\|g\|_n=\left(n^{-1}\sum_{i=1}^n g^2(z_i)\right)^{1/2},
\qquad
\|g\|_{L^2}=\sqrt{\EE g^2(z_1)}.
\]
Let \(f:\cZ\to\RR\) be a random function that is independent of the \(z_i\)'s.
\begin{enumerate}
\item Suppose that \(\PP(|f(z_1)|>r)\leq \varepsilon\) and \(\EE |f(z_1)|^4\leq U^4\) hold for some deterministic \(r,U\geq 0\). Then, for all \(\delta\in(0,1)\),
\[
\PP\left(
\left|\|f\|_n-\|f\|_{L^2}\right|
\leq
r\sqrt{\frac{7\log(2/\delta)}{n}}+U\varepsilon^{1/4}
\right)
\geq 1-n\varepsilon-\delta.
\]
\item Define \(h(\cdot)=f(\cdot)/\|f\|_{L^2}\), and suppose that, conditioned on \(f\), \(h(z_i)\) is sub-Gaussian in the sense that
\[
\sup_{p\geq 1}\left\{p^{-1/2}\EE^{1/p}\bigl[|h(z_i)|^p\mid f\bigr]\right\}
\leq K\in[1,\infty).
\]
There exists a constant \(C\geq 1\) such that when \(\delta\in(0,1)\) and \(n\geq C^2K^4\log(2/\delta)\), we have
\[
\PP\left(
\left|\|f\|_n^2-\|f\|_{L^2}^2\right|
\leq
\|f\|_{L^2}^2\cdot CK^2\sqrt{\frac{\log(2/\delta)}{n}}
\right)
\geq 1-\delta.
\]
\end{enumerate}
\end{lemma}

\subsection{Other Lemmas}
\begin{lemma}\label{lemma; matrix inverse inequality}\label{lemma; psd inverse inequality}
For any psd trace class $A \succeq 0$, the following holds: 
\begin{align*}
A^{\frac{1}{2}}(A+\lambda I)^{-1} A^{\frac{1}{2}} \preceq I
\end{align*}
Also, if two invertible trace-class operators $A,B$ satisfy $A \succeq B$, then $A^{-1} \preceq B^{-1}$.
\end{lemma}

\begin{proof}
Take an eigendecomposition of $A$ as $A = P \Lambda P^\top$. 
Then, 
\begin{align*}
A^{\frac{1}{2}}(A+\lambda I)^{-1} A^{\frac{1}{2}} &= P \Lambda^{\frac{1}{2}} P^\top P (\Lambda + \lambda I)^{-1} P^\top  P \Lambda^{\frac{1}{2}} P^\top  \\
&= P \Lambda^{\frac{1}{2}}(\Lambda + \lambda I)^{-1}\Lambda^{\frac{1}{2}} P^\top \\
&\preceq PP^\top = I.
\end{align*}
For the second statement, first note that $I \preceq B^{-1 / 2} A B^{-1 / 2}$. 
The eigenvalues of the latter symmetric operator are thus $\geq 1$. Its inverse $B^{1 / 2} A^{-1} B^{1 / 2}$ has eigenvalues $\leq 1$, that is $B^{1 / 2} A^{-1} B^{1 / 2} \preceq I$. This gives $z^T B^{1 / 2} A^{-1} B^{1 / 2} z \leq\|z\|^2$ for $z \in \HH$. 
Setting $w=B^{1 / 2} z$, this gives $w^T A^{-1} w \leq z^T B^{-1} z$, which implies $A^{-1} \preceq B^{-1}$.  
\end{proof}

\begin{lemma}\label{lemma; trace simple inequality}
For any psd operators \(A,B,C\) with \(B \preceq C\), 
\begin{align*}
\Tr(AB) \leq \Tr(AC).
\end{align*}    
\end{lemma}

\begin{proof}
Observe that 
\begin{align*}
\Tr(AB) = \Tr(A^{\frac{1}{2}} BA^{\frac{1}{2}}) \leq \Tr(A^{\frac{1}{2}} CA^{\frac{1}{2}}) = \Tr(AC).
\end{align*}
\end{proof}

\begin{lemma}[Key constant of local polynomial regression]\label{lemma; key constant LPR}
For any \(k \in \NN\), the matrix
\[
\EE_{x \sim \operatorname{Unif}[0,1]}[(1,x, \dots, x^k)(1,x, \dots, x^k)^\top]
\]
has condition number (i.e., the ratio of the maximum eigenvalue to the minimum eigenvalue) \(\cO\!\left(\frac{(1+\sqrt{2})^{4k}}{\sqrt{k}}\right)\), and hence its minimum eigenvalue is lower bounded by \(\frac{\sqrt{k}}{c^k}\).
\end{lemma}

\begin{proof}
Note that \(\EE_{x \sim \operatorname{Unif}[0,1]}[(1,x, \dots, x^k)(1,x, \dots, x^k)^\top]\) is a Hilbert matrix, and this follows from well-known results in \citet{choi1983tricks}.
\end{proof}

\section{Supplements for Numerical and Real-world Studies}\label{sec:app:numerical}

\subsection{Benchmark Methods}\label{sec:app:bench}
\noindent
To ensure fair comparisons, we adopt cross-fitting on the benchmark methods \texttt{DR-CATE}, \texttt{ACW-CATE}, 
\textcolor{black}{and \texttt{R-Learner}}. 

\paragraph*{Separate regression (\texttt{SR}).} We randomly divide the source data into two parts, $\mathcal{D}_1$ and $\mathcal{D}_1'$. For $f_1^\star(z)$ and $f_0^\star(z)$, we fit candidate models using the subsets of $\mathcal{D}_1$ with $a = 1$ and $a = 0$, respectively. We then fit imputation models for $f_1^\star(z)$ and $f_0^\star(z)$ using $\widetilde\lambda_1= \widetilde\lambda_0= \frac{1}{5n}$ (the same as $\lambda_{0,1},\lambda_{0,0}$ in \texttt{COKE}) with the subsets of $\mathcal{D}_1'$ with $a = 1$ and $a = 0$, respectively. The model selection for each $f_a^\star(x)$ follows the pseudo-labeling technique described in \cite{wang2026pseudo}, where pseudo-labels are generated for all unlabeled target data and used for model selection over KRR estimators with tuning parameters ranging in ${\bm \Lambda}$. In our numerical studies, this range of tuning parameters is set to be the same as \texttt{COKE}, i.e., ${\bm \Lambda} = \{\frac{2^k}{5n} : k = 0, 1, \ldots, \lceil\log_2(5n)\rceil\}$. The final model is computed as $\hat h(z) = \hat f_1(z) - \hat f_0(z)$.

\paragraph*{\texttt{DR-CATE}.} We implement the two-stage DR-Learner, as described in \cite{kennedy2020towards}, using kernel ridge regression for all regression tasks. The source dataset is randomly divided into two subsets, $\mathcal{D}_1$ and $\mathcal{D}_1'$. Using $\mathcal{D}_1$, we estimate the propensity score $\pi(z)$ through logistic regression, and we also estimate the outcome regression functions $f_a^\star(z)$ (for $a = 0, 1$) using kernel ridge regression with model selection performed through hold-out validation.

Let $d_i = (z_i, a_i, y_i)$. We define the pseudo-outcome as $$\widehat\varphi(d)=\frac{a-\hat \pi(z)}{\hat \pi(z)(1-\hat \pi(z))}(y-\hat f_a(z))+\hat f_1(z)-\hat f_0(z).$$

Next, we regress the pseudo-outcome $\widehat\varphi(d)$ on $z$ using $\mathcal{D}_1'$ through kernel ridge regression with hold-out validation. We divide $\mathcal{D}_1'$ into two halves: training set $\cD_2$ and validation set $\mathcal{D}_3$. Then we fit KRR for $\widehat\varphi(d_i)$ against $z_i$ on $\cD_2$, with the candidate set of penalization parameters being the same as \texttt{COKE}, i.e., ${\bm \Lambda} = \{\frac{2^k}{5n} : k = 0, 1, \ldots, \lceil\log_2(5n)\rceil\}$ and select the best estimator using $\mathcal{D}_3$.

\paragraph*{\texttt{ACW-CATE}.} We implement the two-stage ACW estimator. The source dataset is randomly divided into two subsets, $\mathcal{D}_1$ and $\mathcal{D}_1'$. Using $\mathcal{D}_1$, we estimate the propensity score and the outcome regression functions as described previously in \texttt{DR-CATE}. We also estimate the density ratio $w(z) = p_\mathcal{T}(z)/p_\mathcal{S}(z)$. Define $S$ such that $S = 0$ indicates data from the source, and $S = 1$ indicates data from the target population. First, we estimate $\mathbb{P}(S = 1 \mid z)$ using both $\mathcal{D}_1$ and target data through logistic regression. The density ratio estimate is then given by:

\[
\widehat\omega(z) = \frac{n_\mathcal{S} \widehat{\mathbb{P}}(S = 1 \mid z)}{n_\mathcal{T} \widehat{\mathbb{P}}(S = 0 \mid z)}.
\]

We define the pseudo-outcome $\widehat\varphi$ as
$$\begin{aligned}\widehat\varphi(d)=&\ (1-S)\cdot\frac{n_{1}+n_\mathcal{T}}{n_{1}}\widehat\omega(z)\left\{\frac{a-\hat \pi(z)}{\hat \pi(z)(1-\hat \pi(z))}(y-\hat f_a(z))\right\}\\&\ +S\cdot\frac{n_{1}+n_\mathcal{T}}{n_\mathcal{T}}\left\{\hat f_1(z)-\hat f_0(z)\right\}.\end{aligned}$$

Next, we use kernel ridge regression to regress the pseudo-outcome $\widehat\varphi$ on $z$ using the other half of the source data $\mathcal{D}_1'$ combined with all target data $\mathcal{D}_\mathcal{T}$. We divide these data into two halves for hold-out validation: the training set $\cD_2$, containing half of $\mathcal{D}_1'$ and half of $\mathcal{D}_\mathcal{T}$, and the validation set $\mathcal{D}_3$. Then we fit KRR for $\widehat\varphi(d_i)$ against $z_i$ on $\cD_2$, with the candidate set of penalization parameters being the same as \texttt{COKE}, i.e., ${\bm \Lambda} = \{\frac{2^k}{5n} : k = 0, 1, \ldots, \lceil\log_2(5n)\rceil\}$ and select the best estimator using $\mathcal{D}_3$.

\paragraph*{\textcolor{black}{\texttt{R-Learner}.}} 
\textcolor{black}{
We implement the R-Learner following the residualized loss formulation in \cite{nie2021quasi}. The source dataset is randomly divided into two subsets, $\mathcal{D}_1$ and $\mathcal{D}_1'$. Using $\mathcal{D}_1$, we estimate the nuisance functions $m^\star(z)=\mathbb{E}[y\mid z]$ and $\pi^\star(z)=\mathbb{P}(a=1\mid z)$. The outcome regression $m^\star(z)$ is estimated by kernel ridge regression, with the tuning parameter selected by hold-out validation within $\mathcal{D}_1$. The propensity score $\pi^\star(z)$ is estimated by logistic regression.
}

\textcolor{black}{
Next, using $\mathcal{D}_1'$, we form the residualized outcome and residualized treatment for each observation $d_i=(z_i,a_i,y_i)$:
\[
\widetilde y_i = y_i-\widehat m(z_i), \qquad \widetilde a_i = a_i-\widehat \pi(z_i).
\]
We further split $\mathcal{D}_1'$ into a training set $\cD_2$ and a validation set $\mathcal{D}_3$. On $\cD_2$, we estimate the CATE function $\tau(z)$ by kernel ridge regression under the empirical R-loss,
\[
\frac{1}{|\cD_2|}\sum_{i\in \cD_2}(\widetilde y_i-\widetilde a_i \tau(z_i))^2 + \lambda \|\tau\|_{\cF}^2.
\]
The tuning parameter $\lambda$ is selected using the validation R-loss on $\mathcal{D}_3$, with the same candidate set of penalization parameters as \texttt{COKE}, i.e., ${\bm \Lambda} = \{\frac{2^k}{5n} : k = 0, 1, \ldots, \lceil\log_2(5n)\rceil\}$. Finally, we repeat the procedure after swapping $\mathcal{D}_1$ and $\mathcal{D}_1'$, and average the two resulting estimators.
}

\subsection{Cross-fitting}\label{sec:app:cross-fitting}

Our cross-fitting procedure consists of the following steps: (1) splitting the source data $\cD$ into $\cD_{1}$ and $\cD_2$; (2) implementing the training procedures in Algorithm \labelcref{algorithm: main} twice separately on the permutations $\{\cD_{1},\cD_2\}$ and $\{\cD_2,\cD_{1}\}$; (3) averaging the two estimators resulted from (2) as the final output.

\subsection{Supplementary Results}

\subsubsection{Simulation}\label{sec:app:simu}
\noindent
Under the default setting with $q = 1$, we compare the performance of \texttt{COKE} with and without cross-fitting across varying $S_B$, as shown in Figure \labelcref{fig:CF}. Cross-fitting improves estimation accuracy by leveraging an additional estimate obtained by swapping the roles of $\mathcal{D}_1$ and $\mathcal{D} \setminus \mathcal{D}_1$, and averaging the two estimates. The mean squared error for the cross-fitted version is consistently lower across all $S_B$ values, showing its effectiveness in reducing estimation error. For example, at $S_B = 20$, cross-fitting reduces the mean squared error by approximately 15.0\%.

\begin{figure}[ht]
\centering
\includegraphics[width=8cm]{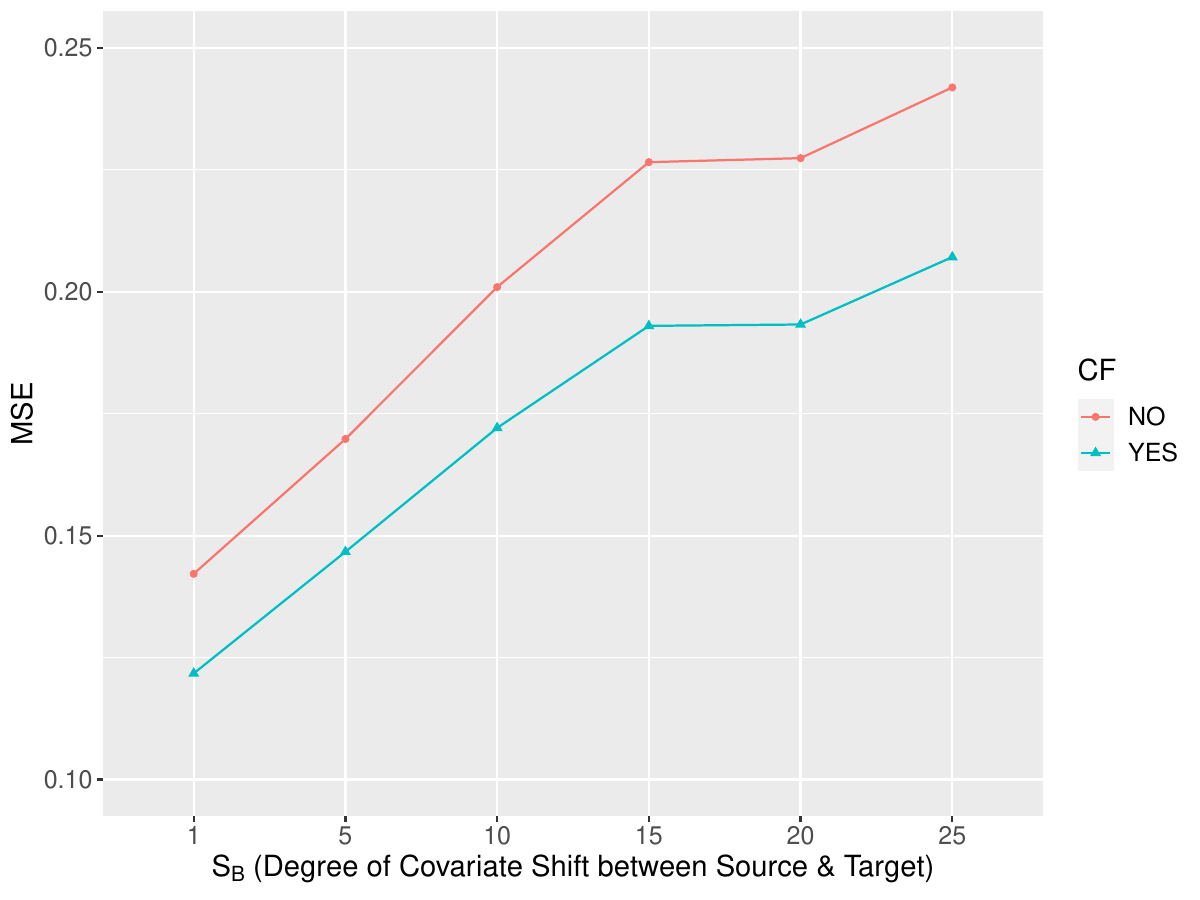}
\caption{Comparison of mean squared error for \texttt{COKE} with and without cross-fitting under the default setting with $q = 1$ across varying $S_B$ (degree of covariate shift between source and target). Cross-fitting reduces estimation error consistently by approximately $13.5\%$--$15\%$.}
\label{fig:CF}
\end{figure}

\subsubsection{Real Example}\label{sec:app:real:results}

In Figure \labelcref{fig:densityratio}, we plot the histograms of $\log_{10}\{\widehat{\omega}(z)\}$ separately on the source and target samples. The source sample has a mean of $-0.889$ and a standard deviation of $0.956$ for $\log_{10}\{\widehat{\omega}(z)\}$, whereas the target sample has a mean of $0.732$ with a standard deviation of $0.748$. The effective sample size of the source sample is $399.01$. Through this plot, one can see strong covariate shift and weak overlap between the source and target. For example, at the mode of $\log_{10}\{\widehat{\omega}(z)\}$ on the target sample with a density larger than $0.6$, the source sample has a density less than $0.02$. Also, on the left tail of $\log_{10}\{\widehat{\omega}(z)\}$ on the source sample, the target data shows nearly zero density.

\begin{figure}[htb!]
\centering
\includegraphics[width=8cm]{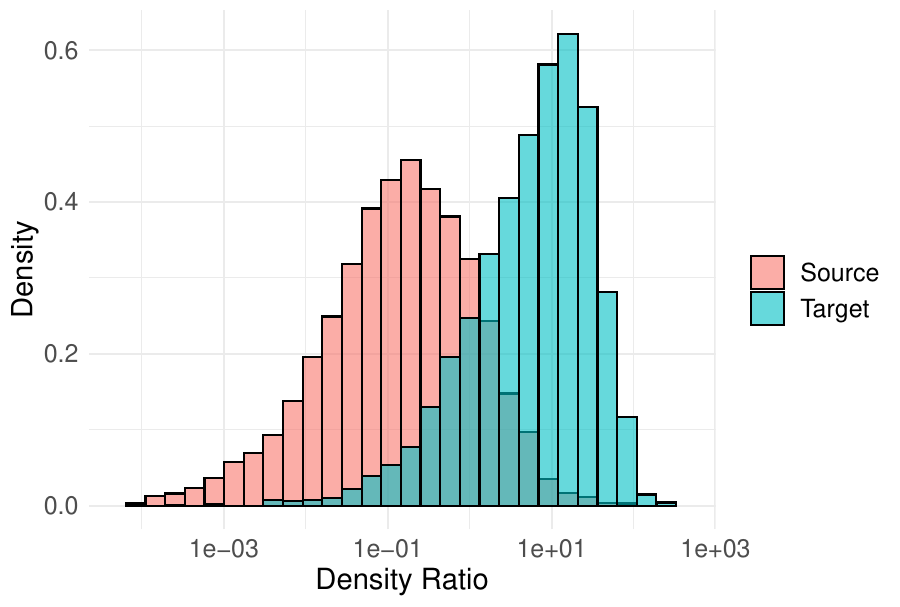}
\caption{Empirical distribution of the logarithms of the estimated density ratio (using the logistic regression) between the source and target in the 401(k) study. 
}
\label{fig:densityratio}
\end{figure}

\begin{table}[htb!]
\centering
\begin{tabular}{|c|ccccc|}
\hline Metrics & \texttt{COKE} & \texttt{SR} & \texttt{DR-CATE} & \texttt{ACW-CATE} & \texttt{R-Learner}\\
\hline Spearman Cor with $\hat{s}_{0i}$ & ${0.54}_{0.012}$ & $0.39_{0.014}$ & $0.30_{0.013}$ & $0.36_{0.014}$ & $0.33_{0.015}$ \\
\hline Pearson Cor with $\hat{s}_{0i}$ & ${0.48}_{0.029}$  & $0.38_{0.032}$ & $0.15_{0.060}$ & $0.30_{0.028}$ & $0.23_{0.033}$ \\ 
\hline
\end{tabular}
\caption{\label{tab:res:401:datasplit} Spearman and  Pearson correlation coefficients (subscribed with their empirical standard errors) between the empirical gold-standard $\hat{s}_{0i}$ and the CATE predictors obtained by data-splitting in the 401(k) study. %The nuisance and CATE models in $\hat{s}_{0i}$ are obtained using generalized additive models.
}
\end{table}

\begin{figure}[htb!]
\centering
\includegraphics[width=15cm]{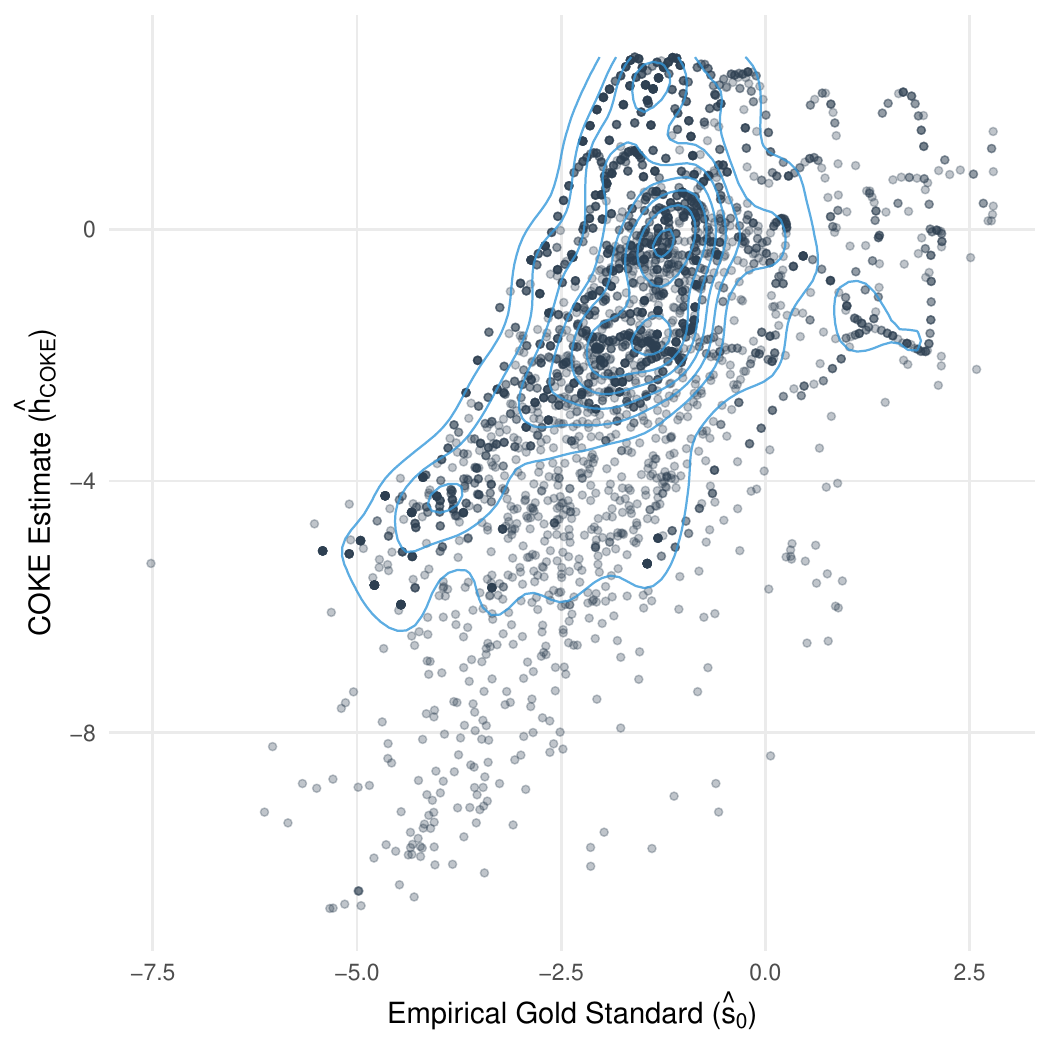}
\caption{{\color{black} Joint distribution of the COKE CATE estimates ($\hat{h}_{\mbox{\texttt{COKE}}}$) and the empirical gold-standard predictors ($\hat{s}_0$) for the 401(k) study. The scatter plot is overlaid with the density contours to illustrate the joint probability mass. The visualization demonstrates the concordance between the estimated and gold-standard treatment effects across the full range of the data.}}
\label{fig:scat401k}
\end{figure}

{\color{black}
Figure \ref{fig:densityrationhanes} plots the empirical distribution of the log-estimated density ratios, $\log_{10}\{\hat{\omega}(z)\}$, separately for the source (NHANES 2001) and target (NHANES 2015) samples. The visible discrepancy between the two histograms indicates a clear covariate shift between the cohorts (especially the ``tails''). This lack of intersection visually underscores the distributional shifts that occurred over the 14-year gap, reinforcing the need for our overlap-adaptive transfer learning approach.}

\begin{figure}[htb!]
\centering
\includegraphics[width=9cm]{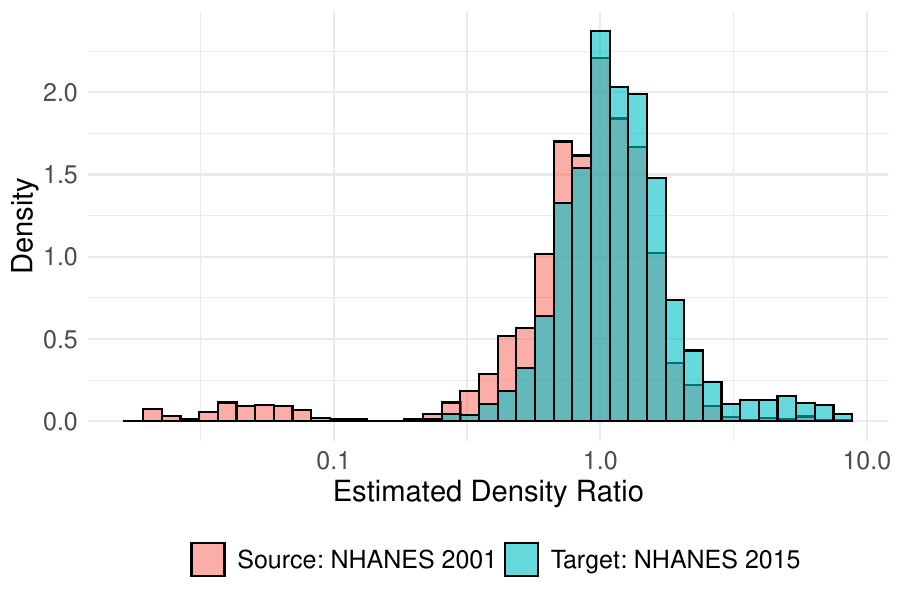}
\caption{Empirical distribution of the logarithms of the estimated density ratio (using the logistic regression) between the source and target in the NHANES study. 
}
\label{fig:densityrationhanes}
\end{figure}

\begin{table}[htb!]
\centering
\centering
\begin{tabular}{|c|ccccc|}
\hline Metrics & \texttt{COKE} & \texttt{SR} & \texttt{DR-CATE} & \texttt{ACW-CATE} & \texttt{R-Learner} \\
\hline Spearman Cor with $\hat{s}_{0i}$ & ${0.51}_{0.012}$ & $0.40_{0.014}$ & $0.33_{0.014}$ & $0.22_{0.015}$ & $0.21_{0.014}$ \\
\hline Pearson Cor with $\hat{s}_{0i}$ & ${0.50}_{0.011}$  & $0.43_{0.014}$ & $0.36_{0.014}$ & $0.22_{0.016}$ & $0.24_{0.013}$  \\ 
\hline
\end{tabular}
\caption{\label{tab:res:nhanes:datasplit} Spearman and  Pearson correlation coefficients (subscribed with their empirical standard errors) between the empirical gold-standard $\hat{s}_{0i}$ and the CATE predictors obtained by data-splitting in the NHANES study. %The nuisance and CATE models in $\hat{s}_{0i}$ are obtained using generalized additive models.
}
\end{table}

\begin{figure}[htb!]
\centering
\includegraphics[width=15cm]{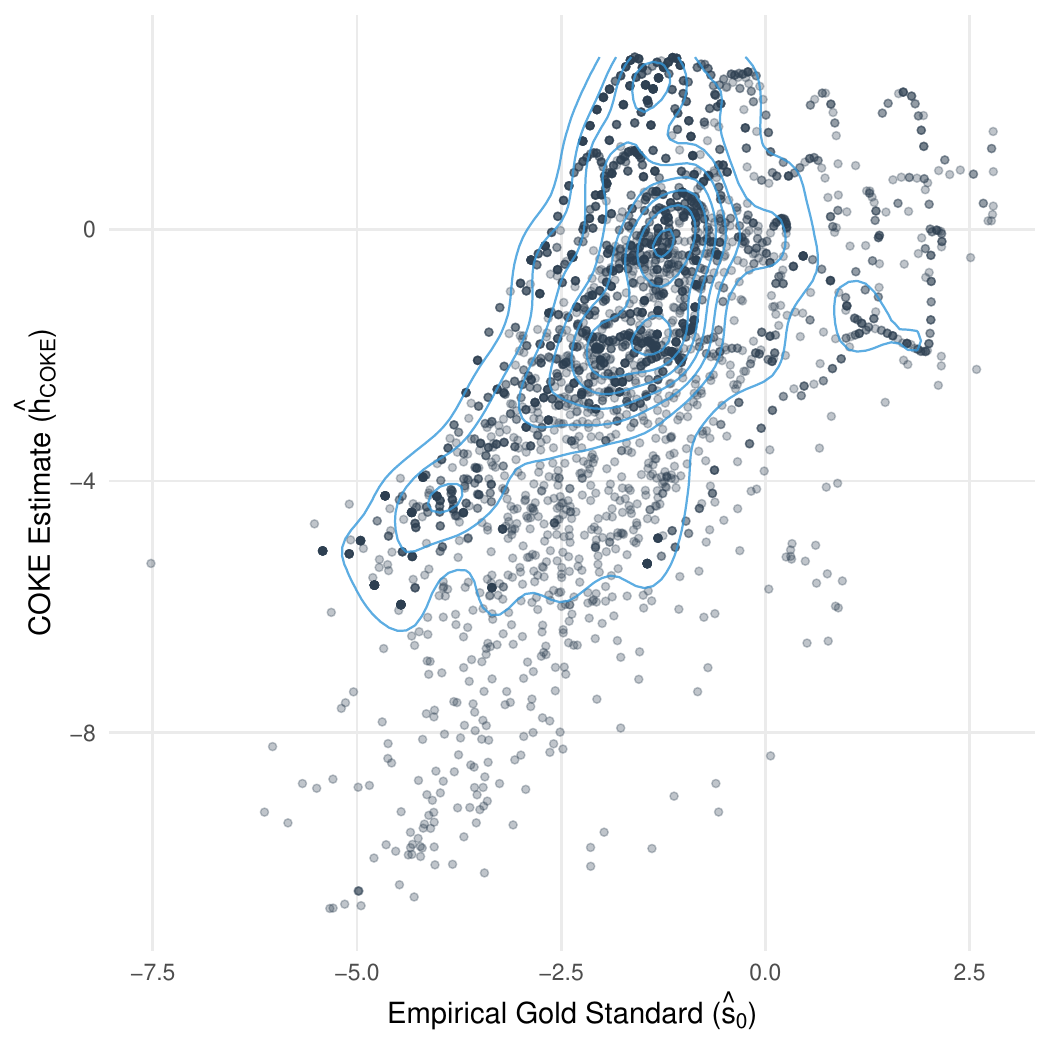}
\caption{{\color{black} Joint distribution of the COKE CATE estimates ($\hat{h}_{\mbox{\texttt{COKE}}}$) and the empirical gold-standard predictors ($\hat{s}_0$) for the NHANES study. The scatter plot is overlaid with the density contours to illustrate the joint probability mass. The visualization demonstrates the concordance between the estimated and gold-standard treatment effects across the full range of the data.}}
\label{fig:scatnhanes}
\end{figure}

\end{document}